\documentclass[10pt,twocolumn,showpacs,preprintnumbers,amsmath,amssymb,aps,prx,longbibliography,superscriptaddress,nofootinbib]{revtex4-2}
\usepackage{braket,graphicx}
\usepackage{amssymb,amsmath,amsthm}
\newtheorem{theorem*}{Theorem*}
\usepackage{bm,dsfont}

\usepackage[unicode=true,
bookmarks=true,bookmarksnumbered=true,bookmarksopen=false,
breaklinks=true,pdfborder={0 0 0},backref=false,colorlinks=true]{hyperref}
\hypersetup{
linkcolor=[RGB]{201,51,0},
citecolor=[RGB]{0,153,38},
urlcolor=[RGB]{0,115,230}}
\usepackage{bookmark}

\usepackage{cleveref}
\crefname{appendix}{App.}{Apps.}
\crefname{equation}{Eq.}{Eqs.}
\crefname{figure}{Fig.}{Figs.}
\crefname{table}{Tab.}{Tabs.}
\crefname{section}{Sec.}{Secs.}

\newtheorem{theorem}{Theorem}

\newcommand{\vect}[1]{{\bm{#1}}}
\newcommand{\dd}{\mathrm{d}}
\newcommand{\ii}{\mathrm{i}}
\newcommand{\ee}{\mathrm{e}}
\newcommand{\EE}{\mathop{\mathbb{E}}}

\newcommand{\Tr}{\mathop{\mathrm{Tr}}}
\newcommand{\id}{\mathds{1}}
\newcommand{\tabref}[1]{Tab.~\ref{#1}}
\newcommand{\figref}[1]{Fig.~\ref{#1}}

\begin{document}
\title{Measurement-Based Quantum Diffusion Models}

\author{Xinyu Liu}
\affiliation{Department of Physics, California Institute of Technology, Pasadena, California 91125, USA}

\author{Jingze Zhuang}
\affiliation{Institute for Interdisciplinary Information Sciences, Tsinghua University, Beijing 100084, China}

\author{Wanda Hou}
\affiliation{Department of Physics, University of California, San Diego, La Jolla, California 92093, USA}

\author{Yi-Zhuang You}
\affiliation{Department of Physics, University of California, San Diego, La Jolla, California 92093, USA}

\date{\today}

\begin{abstract}
We introduce measurement-based quantum diffusion models that bridge classical and quantum diffusion theory through randomized weak measurements. The measurement-based approach naturally generates stochastic quantum trajectories while preserving purity at the trajectory level and inducing depolarization at the ensemble level. We address two quantum state generation problems: trajectory-level recovery of pure state ensembles and ensemble-average recovery of mixed states. For trajectory-level recovery, we establish that quantum score matching is mathematically equivalent to learning unitary generators for the reverse process. For ensemble-average recovery, we introduce local Petz recovery maps for states with finite correlation length and classical shadow reconstruction for general states, both with rigorous error bounds. Our framework establishes Petz recovery maps as quantum generalizations of reverse Fokker-Planck equations, providing a rigorous bridge between quantum recovery channels and classical stochastic reversals. This work enables new approaches to quantum state generation with potential applications in quantum information science.
\end{abstract}

\maketitle

\section{Introduction}
Diffusion-based generative models have achieving remarkable success in generating high-quality images, text, and other classical data structures through noise injection and denoising~\cite{ho2020denoising, song2020score, dhariwal2021diffusion}. This success has naturally motivated the development of quantum analogs---quantum diffusion models that aim to generate quantum states rather than classical data. Recent years have witnessed growing interest in this direction, with several frameworks proposed~\cite{Parigi2023Q2308.12013,Cacioppo2023Q2311.15444,Chen2024Q2401.07039,Kolle2024Q2401.07049,Zhang2024G2310.05866,Zhu2024Q2404.06336,tang2025quadim,Kwun2025M2411.17608,Liu2025B2505.18621,Huang2025C2506.19270,Zhang2026P,Cui2025Q2508.12413}. These approaches have demonstrated the ability to generate both pure and mixed quantum states with high fidelity, unlock a range of applications in quantum information science---from state preparation to quantum error correction.

However, despite these advances, a fundamental gap remains between classical and quantum diffusion models. Classical diffusion theory is built upon a rich mathematical framework \cite{Yang2022D2209.00796} (see \tabref{tab:eqns}) that connects stochastic differential equations (SDEs), ordinary differential equations (ODEs) and partial differential equations (PDEs) through the \emph{score function}---the gradient of log probability. This framework provides not only a unified understanding of forward and reverse processes but also principled training objectives through score matching. While some existing quantum diffusion models \cite{Zhang2024G2310.05866} have implemented forward diffusion and backward denoising processes, the theoretical correspondence between quantum SDEs and PDEs remains unclear, and training objectives are often constructed heuristically rather than derived from first principles.

\begin{table}[ht]
\centering
\renewcommand{\arraystretch}{1.4}
\begin{tabular}{l p{0.42\linewidth} p{0.42\linewidth}}
\hline
 & \textbf{Classical Diffusion} & \textbf{Quantum Diffusion} \\
\hline
{SDE:} & $\dd \vect{x} = \vect{f}\,\dd t + g\,\dd \vect{w}$
    & $\dd\ket{\psi} = (-\frac{\gamma}{2}\delta O^{2}\,\dd t+ \sqrt{\gamma}\,\delta O\,\dd w)\ket{\psi}$ 
\\
{ODE:} & $\dd \vect{x} = (\vect{f} - \frac{1}{2}\vect{s})\dd t$
    & $\dd\ket{\psi} = -\ii\,H(\psi)\ket{\psi}\,\dd t$ 
\\
{PDE:} & $\partial_t p = -\nabla\cdot\big((\vect{f} - \frac{1}{2}\vect{s})p\big)$
    & $\partial_t \bar{\rho} = -\frac{\gamma}{2}[O,[O,\bar{\rho}]]$ 
\\
\hline
\end{tabular}
\caption{Comparison between classical and quantum formulations of diffusion. 
Classical diffusion is driven by stochastic noise, while quantum diffusion is driven by repeated measurement of a random observable $O$ with measurement strength $\gamma$. In classical diffusion, $p=p(\vect{x})$ denotes the probability distribution of classical random variable $\vect{x}$, and $\vect{s}(\vect{x})= g^2\,\nabla_\vect{x}\log p(\vect{x})$ is the score function. In quantum diffusion, $\bar{\rho}=\int\dd\psi\,p(\psi) \ket{\psi}\bra{\psi}$ characterizes the average state (density matrix), $\delta O = O - \bra{\psi} O \ket{\psi}$ and $H(\psi)$ is a state-dependent control Hamiltonian.}
\label{tab:eqns}
\end{table}

In this work, we bridge this theoretical gap by establishing a complete correspondence between classical and quantum diffusion models. Our key insight is that quantum measurement naturally provides the stochastic processes needed for diffusion. We propose using \emph{randomized weak measurements} as the forward diffusion process, which naturally gives rise to stochastic differential equations for quantum state evolution. This approach not only provides a physically realizable implementation but also establishes the missing theoretical connections.

We address two distinct but related problems in quantum state generation: the recovery of pure state ensembles and the recovery of average states (see \figref{fig:roadmap}). For pure state ensembles, we demonstrate that score matching in the quantum setting is mathematically equivalent to learning unitary generators (control Hamiltonian) for the reverse process. This equivalence provides a principled training objective that was previously missing in quantum diffusion literature. For average state recovery, we introduce two complementary approaches: local Petz recovery maps for states with finite correlation length and classical shadow reconstruction for general states. Both approaches come with rigorous error bounds that scale favorably with system size.

\begin{figure}[ht]
    \centering
    \includegraphics[width=\linewidth]{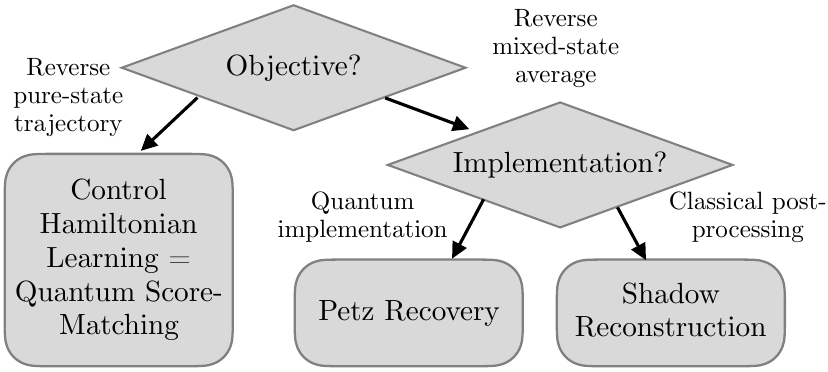}
    \caption{Road map for reverse quantum diffusion.}
    \label{fig:roadmap}
\end{figure}

A key advantage of our measurement-based approach is that measurement serves dual purposes: it drives the forward stochastic process while simultaneously extracting information about the quantum state. This information can be used to condition the reverse process, enabling the construction of state-dependent control Hamiltonian or quantum processes. This capability is unique to measurement-based approaches and is not available in other quantum diffusion frameworks that rely on noise channels or unitary scrambling. Admittedly, several works have explored measurement-driven quantum processes for recovering target quantum state ensembles~\cite{Zhang2024G2310.05866,Kwun2025M2411.17608,Zhang2026P,Cui2025Q2508.12413}, underscoring the growing interest in using measurements to induce random state ensembles and designs~\cite{Ho2022E2109.07491,Cotler2023E2103.03536,Choi2023P2103.03535,Zhang2024H2411.03587}. However, the ability to condition on measurement outcomes in this work further enables deterministic, flow-based (unitary) reversal at the level of individual pure states, rather than, as in most existing literature\cite{lima2025review}, merely reconstructing the density matrix in an ensemble-averaged sense. 

Our theoretical framework establishes several important connections. We show that the average behavior of the measurement-induced SDE is described by a Lindblad master equation, providing the quantum analog of the Fokker-Planck equation in classical diffusion. We prove that Petz recovery maps serve as the quantum generalization of reverse Fokker-Planck equations. 

\section{Forward Quantum Diffusion}
\subsection{Randomized Weak Measurements}
Consider an $n$-qubit system initially prepared in a pure state $\ket{\psi_0}$. Under weak, continuously monitored measurements, the \emph{conditional} state $\ket{\psi_t}$ remains pure at all times. Randomness of the measurement outcomes induces a \emph{stochastic} trajectory of pure states (a quantum trajectory), given by a nonlinear stochastic differential equation (SDE): \cite{Ghirardi1990M,Adler2013O,Giachetti2023E2306.12166}
\begin{equation}\label{eq:sde_pure}
\dd|\psi_t\rangle = \Big( -\frac{\gamma}{2} \delta O_t^2 \dd t + \sqrt{\gamma} \delta O_t \dd w_t \Big) |\psi_t\rangle,
\end{equation}
where $\delta O_t:=O_t-\bra{\psi_t}O_t\ket{\psi_t}$ is $\ket{\psi_t}$-dependent, with $O_t$ being the observable to be measured at time $t$. $\gamma>0$ parameterize the \emph{measurement strength}. $\dd w_t$ describes a standard Wiener process satisfying $\mathbb{E}[\dd w_t] = 0$ and $\mathbb{E}[\dd w_t^2] = \dd t$, which captures the stochastic fluctuations of the \emph{weak measurement outcome} $\dd o_t$ around its expectation value:
\begin{equation}\label{eq:weak_outcome}
\dd o_t =\bra{\psi_t}O_t\ket{\psi_t}\dd t+\frac{\dd w_t}{2\sqrt{\gamma}}.
\end{equation}
Operationally, $\dd o_t$ denotes the rescaled signal accumulated from weak measurements of $O_t$ over the interval $\dd t$. For a single weak measurement within the $\dd t$ interval, the observed eigenvalue of $O_t$ corresponds to $ 2\sqrt{{\gamma}/{\dd t}}\,\dd o_t$. 

In classical diffusion models, the forward process is typically designed as a trivializing transformation, mapping a complex target distribution to a simple reference distribution from which sampling is straightforward. 
To construct an analogous trivializing map for quantum diffusion, we consider a \emph{randomized measurement} protocol, in which the observable $O_t$ is drawn at each time step from a stochastic process independent of the evolving state $\ket{\psi_t}$. Randomized measurements generate an ensemble of pure states  
\begin{equation}
    \mathcal{E}_\psi=\{\ket{\psi}\!\bra{\psi}: \psi \sim p(\psi)\},
\end{equation}
which can be either (i) described by the probability distribution $p(\psi)$ over the projective Hilbert space, or (ii) characterized by its first moment---the ensemble-averaged density matrix  
\begin{equation}
    \bar{\rho} = \EE_{\psi \sim p(\psi)} \ket{\psi}\!\bra{\psi}.
\end{equation}

It is possible to choose the distribution of $O_t$ such that, under randomized measurements, any initial ensemble evolves toward a uniform distribution over the Hilbert space. In this limit, the ensemble-averaged density matrix approaches the maximally mixed state,  
$\mathbb{E}_{\mathrm{traj.}}\!\ket{\psi_t}\!\bra{\psi_t} \xrightarrow{t \to \infty} \id /2^n$.
This does not contradict the fact that each conditional state $\ket{\psi_t}$ along a single trajectory remains pure; the mixing arises solely from averaging over the classical randomness of the measurement record. Thus, randomized weak measurement can preserve purity at the trajectory level while inducing complete depolarization at the ensemble level, realizing a trivializing forward quantum diffusion.

Consequently, reverse quantum diffusion admits two natural generative objectives:
\begin{enumerate}
\item[(i)] \textbf{Trajectory-level recovery:} sample from or reconstruct the trajectory-level pure-state distribution $p(\psi_t)$ induced by the forward process;
\item[(ii)] \textbf{Ensemble-average recovery:} reconstruct the averaged state $\bar{\rho}_t = \mathbb{E}_{\psi_t} \ket{\psi_t}\!\bra{\psi_t}$ at each diffusion time $t$.
\end{enumerate}
Whereas much of the literature~\cite{lima2025review} target (ii) of reconstructing $\bar{\rho}_t$, the trajectory-level objective (i) is more demanding, as it requires a classical decoder to infer the latent pure state from each trajectory's stochastic measurement record. Crucially, prior approaches cast denoising as a state-preparation problem, seeking optimal protocols that act universally on all states (regardless of measurement record). By contrast, our measurement-based framework formulates denoising as a control problem that learns control Hamiltonian conditioned on the measurement record of each trajectory.

In this work, we address both objectives, as illustrated in \figref{fig:roadmap}.  
For trajectory-level recovery (i), we develop a \emph{control Hamiltonian learning} approach that learns a time-reversal unitary evolution from measurement records.  
For ensemble-average recovery (ii), we consider two complementary strategies: a \emph{Petz recovery} map for quantum implementation, and \emph{classical shadow reconstruction} for purely classical post-processing. 

\subsection{Measurement-Based Forward Diffusion}
According to \cref{eq:sde_pure}, the state projection operator $\rho_t=\ket{\psi_t}\bra{\psi_t}$ (pure-state density matrix) evolves along the quantum trajectory as: \cite{Jacobs2006Aquant-ph/0611067,Adler2013O}
\begin{equation}\label{eq:sde_mixed}
\dd\rho_t
= -\frac{\gamma}{2} [O_t, [O_t, \rho_t]] \,\dd t
  + \sqrt{\gamma} \,\{ \delta O_t, \rho_t \} \,\dd w_t.
\end{equation}
The first term in \cref{eq:sde_mixed} represents a deterministic drift towards eigenbasis of $O_t$, arising from the measurement-induced decoherence. The second term captures the stochastic back-action generated by the quantum state collapse under measurement.

Let $\vect{P}=\{P_i\}_{i=1}^{4^n}$ denote the $n$-qubit Pauli-operator basis with $\Tr(P_i P_j)=2^n\delta_{ij}$.
Any state $\rho_t$ or observable $O_t$ can be expanded as
\begin{equation}
\label{eq:pauli-expansion2}
\rho_t = \frac{\vect{z}_t\cdot\vect{P}}{2^n},\quad
O_t=\vect{x}_t\cdot \vect{P}, 
\end{equation}
with $z_{t,i}=\Tr(\rho_t P_i)=\bra{\psi_t}P_i\ket{\psi_t}$, $x_{t,i}=\Tr(O_t P_i)/2^n$. Then \cref{eq:sde_mixed} implies
\begin{equation}\label{eq:sde_coeff}
\dd\vect{z}_t
= \vect{f}(\vect{z}_t, t) \,\dd t
+ \vect{g}(\vect{z}_t, t) \,\dd w_t,
\end{equation}
with the drift $\vect{f}$ and the noise $\vect{g}$ functions given by (Einstein summation convention implied)
\begin{equation}
\begin{split}
f_l(\vect{z}_t,t) &= -\gamma(c_{ijm} c_{klm} - c_{ikm} c_{jlm}) x_{t,i}x_{t,j}z_{t,k}, \\
g_l(\vect{z}_t,t)  &= \sqrt{\gamma}((c_{ijl} + c_{jil})x_{t,i}z_{t,j} - 2 x_{t,i} z_{t,i} z_{t,l}),
\end{split}
\end{equation}
and $c_{ijk} = \mathrm{Tr}(P_i P_j P_k)/\mathrm{Tr}\,\id$ being the operator product expansion (OPE) coefficient.

The SDE \cref{eq:sde_coeff} for $\vect{z}_t$ takes the form of a classical Langevin equation, linking measurement-based quantum diffusion directly to classical diffusion.  
Here, the Pauli expectation values $\vect{z}_t$ serve as stochastic variables encoding the quantum state $\rho_t$.  
If accessible, the time series $\{\vect{z}_t\}$ could train a denoising model for reversal.  
However, in quantum systems, $\vect{z}_t$ (or $\rho_t$) is not directly observable, and must be inferred from the measurement record $\{(O_t,o_t)\}$. We will comeback to this problem later in Sec.~\ref{sec:classical_decoder}.

\subsection{Kraus Operator Formulation}\label{sec:Kraus}

The stochastic evolution of $\ket{\psi_t}$ in \cref{eq:sde_pure} can be equivalently expressed in the Kraus operator formalism:
\begin{equation}
\ket{\psi_{t+\dd t}}
= \frac{K_{\dd t}(O_t,\dd o_t)\,\ket{\psi_t}}
{\|K_{\dd t}(O_t,\dd o_t)\,\ket{\psi_t}\|},
\end{equation}
where $K_{\dd t}$ denotes the \emph{Kraus operator} for a weak measurement of observing $O_t$ and obtaining the measurement outcome signal $\dd o_t$:
\begin{equation}\label{eq:K_dt}
K_{\dd t}(O_t,\dd o_t)
= \Big(\frac{2\gamma}{\pi\dd t}\Big)^{\frac{1}{4}}\!
\exp\!\Big( -\frac{\gamma}{\dd t}\,\big(\dd o_t - O_t\,\dd t\big)^2 \Big).
\end{equation}
The probability density for observing $\dd o_t$ is given by  $p(\dd o_t|O_t,\psi_t) = \|K_{\dd t}(O_t,\dd o_t)\,\ket{\psi_t}\|^2$,
which follows a Gaussian distribution centered at $\bra{\psi_t}O_t\ket{\psi_t}\dd t$ with variance $(4\gamma)^{-1}\dd t$, consistent with \cref{eq:weak_outcome}.

The randomized measurement not only drives the diffusion of the conditional pure state $\ket{\psi_t}$, but also generates a \emph{weak measurement record} along the trajectory, represented by the time series
\begin{equation}
\mathcal{O} = \{(O_{t}, o_{t}) \mid 0 \le t \le T\}.
\end{equation}


\noindent The accumulated Kraus operator along the measurement trajectory $\mathcal{O}$ is
\begin{equation}\label{eq:K_t}
K_t(\mathcal{O}) =
\mathcal{T}\!\!\prod_{t'=0}^{t-\dd t} K_{\dd t}(O_{t'},\dd o_{t'}),
\end{equation}
where $\mathcal{T}$ denotes time-ordering. 

\subsection{Pauli Twirled Measurement Channel}
Averaging over all possible measurement trajectories defines the weak measurement channel $\mathcal{F}_t$:
\begin{equation}
\bar{\rho}_t = \mathcal{F}_t(\bar{\rho}_0)
= \EE_{\mathcal{O}}
K_t(\mathcal{O})\,\bar{\rho}_0\,K_t^\dagger(\mathcal{O}),
\end{equation}
which describes the ensemble-averaged evolution under the forward diffusion process.  
Physically, $\mathcal{F}_t$ captures the decoherence induced by weak measurements after averaging over the classical randomness of observable choices $O_t$ and measurement outcomes $o_t$ in $\mathcal{O}=\{(O_t,o_t)\}$.  

Averaging the stochastic evolution in \cref{eq:sde_mixed} over all measurement trajectories transforms the SDE for the pure-state projector $\rho_t$ into a deterministic PDE for the average state $\bar{\rho}_t$.  
The resulting evolution is governed by a Lindblad master equation,
\begin{equation}\label{eq:pde}
\partial_t \bar{\rho}_t = \mathcal{L}[\bar{\rho}_t]
= -\frac{\gamma}{2} \EE_{O_t} [O_t, [O_t, \bar{\rho}_t]],
\end{equation}
where the specific form of the Liouvillian super-operator $\mathcal{L}$ depends by the observable distribution $p(O_t)$. Equivalently, the channel can be written as $\mathcal{F}_t = \ee^{t\mathcal{L}}$.

The measurement channel $\mathcal{F}_t$ is said to be \emph{Pauli twirled} \cite{Dankert2009Equant-ph/0606161,Magesan2011S1009.3639,kuo2020medul,hu2023cstwl} if it is weakly symmetric \cite{Buca2012A1203.0943,Albert2014S1310.1523,Albert2018L1802.00010} under local Clifford transformations, i.e. $\forall U\in C_1^{\otimes n}$ ($C_1$ denotes the single-qubit Clifford unitary group):
\begin{equation}
\mathcal{F}_t(U\bar{\rho}_0 U^\dagger)=\mathcal{F}_t(\bar{\rho}_0)=U\mathcal{F}_t(\bar{\rho}_0 )U^\dagger.
\end{equation}
Formally, Pauli twirled measurement channels can be constructed by requiring the observable distribution $p(O_t)$ to be invariant under local Clifford unitaries, i.e. $p(O_t)=p(U^\dagger O_tU)$. The nice property of Pauli twirled channels is that they are diagonal in the Pauli basis, such that given the decomposition $\bar{\rho}_t=\vect{z}_t\cdot\vect{P}/2^n$, the action of the channel $\bar{\rho}_t=\mathcal{F}_t(\bar{\rho}_0)$ simply amounts to rescaling each Pauli expectation value
\begin{equation}
z_{t,i}=w_{\mathcal{F}_t}(P_i)z_{0,i}
\end{equation}
by the \emph{Pauli weight} $w_{\mathcal{F}_t}(P):=\Tr( P\mathcal{F}_t(P))\,/\,\Tr\id$ that characterizes the decay of each Pauli component under the channel. 

\subsection{Randomized Pauli Measurement}
To enable concrete analysis and numerical simulation, we focus on a discrete-time weak-measurement scheme with small but finite time steps of duration $\delta t>0$. 
At each step, the measured observable $O_t$ is drawn uniformly at random from single-qubit Pauli operators
\begin{equation}
    \mathcal{P}_1 = \{\,\sigma_{x,j},\sigma_{y,j}, \sigma_{z,j} \mid j=1,\dots,n\,\}.
\end{equation}
Unless otherwise stated, all subsequent results will refer to this single-qubit measurement setting.

For an $n$-qubit system, the protocol applies one such on-site weak measurement to every qubit simultaneously: at each step, an independent single-qubit Pauli $O_t^{(j)} \in \{\sigma_{x,j}, \sigma_{y,j}, \sigma_{z,j}\}$ is drawn uniformly for each qubit $j$, paired with its own Wiener increment $\dd w_t^{(j)}$. Because on-site Paulis on different qubits commute, the $n$ measurements can be performed simultaneously, and the per-step Kraus operator factorizes as a tensor product of on-site Kraus operators. The trajectory-level SDE \cref{eq:sde_pure} correspondingly extends to
\begin{equation}\label{eq:sde_multisite}
\dd\ket{\psi_t}=\sum_{j=1}^{n}\Big(\!-\tfrac{\gamma}{2}\,(\delta O_t^{(j)})^2\,\dd t+\sqrt{\gamma}\,\delta O_t^{(j)}\,\dd w_t^{(j)}\Big)\ket{\psi_t}.
\end{equation}
This site-factorized structure makes the forward channel a true \emph{trivializing map}: the ensemble-averaged state on $n$ qubits converges to the maximally mixed \emph{product} state $I/2^n$, with no residual entangling correlations preserved.

As Pauli observables $O_t\in\mathcal{P}_1$ satisfy $O_t^2=\id$ with eigenvalues $\pm1$, the measurement signal $\delta o_t=\pm\frac{1}{2}\sqrt{\delta t/\gamma}$ also take binary values correspondingly. In this case, the differential Kraus operator in \cref{eq:K_dt} reduces to the following discrete form after proper normalization,
\begin{equation}
K_{\delta t}(O_t,\delta o_t)
\;=\;
\frac{1}{\sqrt{2}}\,
\exp\!\Big(-\gamma\,\delta t +2\gamma\,O_t\,\delta o_t\Big).
\end{equation}
Averaging over outcomes and the random choice of $O_t$ yields the single-step measurement channel
\begin{equation}
\label{eq:channel_delt}
\begin{split}
\bar{\rho}_{t+\delta t}&=\mathcal{F}_{\delta t}(\bar{\rho}_t)\\&=\mathbb{E}_{O_t,\delta o_t} K_{\delta t}(O_t,\delta o_t)\,\bar{\rho}_t\,K_{\delta t}^\dagger(O_t,\delta o_t).
\end{split}
\end{equation}
This acts as a product of single-qubit depolarization channel acting identically and independently across all qubits.
The corresponding finite-difference Lindblad equation reads $\delta\bar{\rho}_{t}=\mathcal{L}[\bar{\rho}_t]\delta t$, with the Liouvillian super-operator decomposes to each qubit $\mathcal{L}=\frac{1}{n}\sum_{j=1}^n \mathcal{L}_j$,
\begin{equation}
\mathcal{L}_j[\bar{\rho}_t]=-\frac{4\gamma}{3}\Big(\bar{\rho}_t-(\Tr\bar{\rho}_t)\frac{\id}{2}\Big),
\end{equation}
which precisely generates local depolarization.

Because the channel is Pauli-twirled, its action on the $n$-qubit Pauli basis is diagonal. 
The Pauli weight depends only on the \emph{operator weight} $|P|$ (the number of qubits on which $P$ acts nontrivially) and takes the closed form
\begin{equation}\label{eq-measurement-channel-pauli-weight}
w_{\mathcal{F}_t}(P)
= \exp\!\Big(-\frac{4\gamma}{3n}\,|P|\,t\Big),
\end{equation}
describing exponential decay of each nontrivial Pauli component at a uniform rate $4\gamma/3n$ per qubit.

\section{Reverse Quantum Diffusion}

\subsection{Control Hamiltonian Learning}

\subsubsection{Classical Decoder}\label{sec:classical_decoder}
The reverse diffusion task can be formulated as a control problem: given the current state of the system, a denoising model must condition on this state and propose control operations that drive the system backward in time. In classical systems, the state variables coincide with physical observables, which can be directly accessed, copied, and processed by the denoising model to generate classical control actions. In quantum systems, however, there is a fundamental separation between the underlying quantum state and measurable observables. The quantum state cannot be directly accessed; it can only be inferred from repeated measurements. Yet, quantum trajectories are not reproducible, making it impossible to reconstruct intermediate states along the trajectory via quantum state tomography without severe post-selection overhead. As a result, a quantum denoising algorithm does not have direct access to the state and must base its control on indirect and noisy inference, which is substantially more challenging than in the classical case.

To address this difficulty, we introduce a \emph{classical decoder} that works alongside the quantum denoising model. The decoder is a classical algorithm that infers a per-trajectory pure-state estimate $\ket{\hat\psi_t}$---and hence the Pauli expectation values $\hat{\vect{z}}_t = \bra{\hat\psi_t}\vect{P}\ket{\hat\psi_t}$---directly from the measurement record $\{(O_t, o_t)\}$, with no oracle access to the underlying quantum state.

We adopt a decoder built on classical-shadow tomography (\cref{sec:shadow}): the records are first inverted into an estimate $\hat{\rho}_0$ of the ensemble-average initial density matrix; a pure-state initial guess $\ket{\hat\psi_0}$ is then drawn from the maximum-entropy unraveling of $\hat{\rho}_0$ (the distribution over pure states with first moment fixed at $\hat{\rho}_0$ and otherwise of maximum entropy); the Kraus-operator formulation of \cref{sec:Kraus} finally replays this guess forward through the same record,
\begin{equation}
\ket{\hat\psi_t} = \frac{K_t(\mathcal{O})\,\ket{\hat\psi_0}}{\|K_t(\mathcal{O})\,\ket{\hat\psi_0}\|},
\end{equation}
yielding a per-trajectory state estimate consistent with both the measurement record and the inferred ensemble. This pipeline is device-realistic and inherits the favorable sample-complexity guarantees of classical-shadow tomography. 

\subsubsection{Quantum Score Matching}
Given the trajectories $\{\vect{z}_t\}$ decoded from the weak-measurement record, reversing the diffusion is a quantum control problem: learn a control Hamiltonian $H_\theta(\vect{z}_t,t)$ that drives $\ket{\psi_t}$ backward along its trajectory. Pure states stay pure under the forward SDE \cref{eq:sde_pure}, so the time-reversed process must preserve purity and is therefore generated by a unitary evolution---a deterministic flow.

A general SDE $\dd\vect{x}=\vect{f}\,\dd t+\vect{g}\,\dd w$ propagates a time-evolving marginal $p(\vect{x},t)$. Its \emph{score} $\vect{s}(\vect{x},t):=\nabla\log p(\vect{x},t)$ plays two structural roles: it supplies the drift correction needed to reverse the SDE in time, and it converts the SDE into an equivalent deterministic \emph{probability-flow} ODE $\dd\vect{x}=\vect{v}\,\dd t$ realizing the same marginal, with $\vect{v}=\vect{f}-\tfrac{1}{2}\vect{s}$~\cite{song2020score,tong2024simulationfree,lu2022maximum}. The four resulting representations of the dynamics (forward/reverse $\times$ SDE/ODE), summarized in \cref{tab:sde_ode_duality}, are therefore tied by the \emph{score--flow relation}
\begin{equation}\label{eq:score_flow}
\vect{s}(\vect{x},t) = 2\,\vect{f}(\vect{x},t) - 2\,\vect{v}(\vect{x},t).
\end{equation}

\begin{table}[t]
\caption{Stochastic--deterministic duality of forward and reverse processes. Both pairs realize the same marginal $p(\vect{x},t)$; the score $\vect{s}$ and flow $\vect{v}$ are tied by \cref{eq:score_flow}. $\bar t=T-t$ denotes reverse time, and $\dd w$ a Wiener increment.}
\label{tab:sde_ode_duality}
\centering
\renewcommand{\arraystretch}{1.35}
\begin{tabular}{lll}
\hline\hline
 & \textbf{Diffusion (SDE)} & \textbf{Flow (ODE)} \\
\hline
Forward & $\dd\vect{x} = \vect{f}\,\dd t + \vect{g}\,\dd w$
        & $\dd\vect{x} = (\vect{f} - \tfrac{1}{2}\vect{s})\,\dd t = \vect{v}\,\dd t$ \\
Reverse & $\dd\vect{x} = (-\vect{f} + \vect{s})\,\dd\bar t + \vect{g}\,\dd w$
        & $\dd\vect{x} = (-\vect{f} + \tfrac{1}{2}\vect{s})\,\dd\bar t = -\vect{v}\,\dd\bar t$ \\
\hline\hline
\end{tabular}
\end{table}

The score $\vect{s}$ and flow velocity $\vect{v}$ both depend on the unknown marginal $p(\vect{x},t)$ and are intractable to evaluate directly; we therefore approximate them by parameterized models $\vect{s}_\theta$ and $\vect{v}_\theta$ learned from sampled trajectory data. For increments along trajectories of the forward SDE, a standard \emph{denoising score-matching} (DSM) objective is~\cite{ho2020denoising,song2020score}, aggregated over training times~$t$,
\begin{equation}\label{eq:loss_score}
\mathcal{L}^{\mathrm{sm}}_\theta = \sum_t \tfrac{1}{2}\,\EE_{\vect{x}_t,\vect{x}_{t-\dd t}}\Big\|\tfrac{\vect{x}_t-\vect{x}_{t-\dd t}}{\dd t} - \vect{f}(\vect{x}_t,t) + \vect{s}_\theta(\vect{x}_t,t)\Big\|^2,
\end{equation}
where $\EE_{\vect{x}_t,\vect{x}_{t-\dd t}}$ denotes an expectation over consecutive pairs along forward-SDE trajectories (and, implicitly, over trajectories as sampled), fitting $\vect{s}_\theta$ to the SDE drift residual. Our backward model is a unitary (deterministic) flow, so at the classical level it parallels sampling along the probability-flow ODE associated with the same marginals~\cite{song2020score}. Substituting \cref{eq:score_flow} into \cref{eq:loss_score} reparameterizes DSM as a regression on $\vect{v}_\theta$ trained from forward-SDE increments,
\begin{equation}\label{eq:loss_crossed_classical}
\mathcal{L}_\theta = \sum_t \tfrac{1}{2}\,\EE_{\vect{x}_t,\vect{x}_{t-\dd t}}\Big\|\tfrac{\vect{x}_t-\vect{x}_{t-\dd t}}{\dd t} + \vect{f}(\vect{x}_t,t) - 2\,\vect{v}_\theta(\vect{x}_t,t)\Big\|^2,
\end{equation}
where the prefactor $2$ on $\vect{v}_\theta$ is inherited directly from the probability-flow identity $\vect{v}=\vect{f}-\tfrac{1}{2}\vect{s}$ underlying \cref{eq:score_flow}.

In the quantum protocol, the same structural asymmetry appears between a \emph{stochastic} forward channel and a \emph{deterministic} reverse channel. Randomized weak measurements realize the forward pure-state diffusion \cref{eq:sde_pure}: each trajectory samples a measurement record whose noise drives the unraveling. The goal of learning is nonetheless to synthesize a \emph{unitary} backward map on $\ket{\psi_t}$, the unique deterministic dynamics compatible with preserving purity when reversing the diffusion. We are therefore in the same forward-SDE / probability-flow-ODE split addressed classically by reparameterizing DSM through \cref{eq:loss_crossed_classical}.

Concretely, we lift that objective to Hilbert space along the pure-state quantum SDE \cref{eq:sde_pure}: $\vect{x}\in\mathbb{R}^d$ becomes $\ket{\psi}\in\mathbb{C}^{2^n}$, the Euclidean norm becomes the state-vector $L^2$ norm, and the drift and flow become state-vector-valued, $\ket{f(\psi,t)} = -\tfrac{\gamma}{2}\,\delta O^2\ket{\psi}$ and $\ket{v_\theta(\psi,t)} = -\ii\,H_\theta(\vect{z},t)\ket{\psi}$. Substituting into the Hilbert-space analogue of \cref{eq:loss_crossed_classical} and expanding to first order in $\dd t$ yields the mean-infidelity loss
\begin{equation}\label{eq:loss_unitary}
\mathcal{L}_\theta = \sum_t \Bigl(\,1 - \EE_{\psi_{t-\dd t},\,\psi_{t}}\,\big|\!\bra{\psi_{t-\dd t}}V_\theta(\vect{z}_t,t)\ket{\psi_t}\!\big|^2\Bigr),
\end{equation}
with $\EE_{\psi_{t-\dd t},\,\psi_{t}}$ averaging over consecutive pure states along unravelings (and, implicitly, over trajectories), and with the same $\sum_t$ convention as in \cref{eq:loss_score,eq:loss_crossed_classical}. The \emph{score operator}
\begin{equation}\label{eq:V_theta}
V_\theta(\vect{z}_t,t) := \exp\!\Big[\big(-\tfrac{\gamma}{2}\,\delta O^2 + 2\ii\,H_\theta(\vect{z}_t,t)\big)\,\dd t\Big]
\end{equation}
implements one $\dd t$-step of the reverse process, with $\ket{\psi_t} = K_{\dd t}(O,\dd o)\ket{\psi_{t-\dd t}}$ a forward-diffused state and $\vect{z}_t = \bra{\psi_t}\vect{P}\ket{\psi_t}$ its Pauli expectation values, both reconstructed from the measurement record by classical decoding. The trained reverse process is generated by applying $V_\theta^\dagger(\vect{z}_t,t)$ at each $\dd t$-step from $t=T$ back to $t=0$.

\subsubsection{Error Bound Analysis}
Given a time-dependent control Hamiltonian $H(t)$, the reverse evolution along each pure-state trajectory from $t = T$ back to $t = 0$ is implemented by the unitary map
\begin{equation}
U[H] \;=\; \mathcal{T}^{\dagger}\exp\!\Big(\ii\int_T^0 H(\vect{z}_t,t)\,\dd t\Big),
\end{equation}
where $\mathcal{T}^{\dagger}$ denotes reverse time-ordering.
Acting with $U[H]$ on each final state $\ket{\psi_T}$ generates a \emph{reverse-propagated ensemble}
\begin{equation}
\mathcal{E}_0[H] = \big\{U[H]\ket{\psi_T}: \ket{\psi_T} \in \mathcal{E}_T \big\},
\end{equation}
where $\mathcal{E}_T$ denotes the final state ensemble $t = T$.
The corresponding \emph{average state} at $t=0$ is obtained by averaging over the ensemble $\mathcal{E}_0[H]$,
\begin{equation}
\bar{\rho}_0[H] =\mathbb{E}_{\psi_T}
U[H]\ket{\psi_T}\!\bra{\psi_T}U^\dagger[H].
\end{equation}
To derive an error bound for the entire pure-state ensemble $\mathcal{E}_0[H]$ produced by the backward unitary evolution, we employ the Wasserstein distance as a metric on the space of pure-state ensembles, since the Wasserstein distance compares full probability distributions rather than just their averages. Suppose the learned generative model (with control Hamiltonian $H_\theta$) produces an ensemble $\mathcal{E}_0[H_\theta] = \{\,|\phi\rangle\langle\phi| : \phi \sim p_\theta(\phi)\,\}$ of pure states drawn from a distribution $p_\theta$, and that the true ensemble of pure states is $\mathcal{E}_0[H_{\mathrm{true}}] = \{\,|\psi\rangle\langle\psi| : \psi \sim p_{\mathrm{true}}(\psi)\,\}$ for some distribution $p_{\mathrm{true}}$. The Wasserstein-1 distance (with trace distance as the cost function) between $p_\theta$ and $p_{\mathrm{true}}$ is defined as \cite{Zhang2024G2310.05866}

\begin{equation}\label{seq-Wp-distance}
\begin{split}
&W_1(p_\theta, p_{\mathrm{true}})
\\
&=  \inf_{\pi \in \Pi(p_\theta, p_{\mathrm{true}})} 
\int D(|\phi\rangle,\,|\psi\rangle) \;
\pi\big(|\phi\rangle,\, |\psi\rangle \big)\;\dd\phi\, \dd\psi ,
\end{split}
\end{equation}
where $\Pi(p_\theta, p_{\mathrm{true}})$ denotes the set of all couplings (joint distributions) $\pi(\phi,\psi)$ having marginals $p_\theta(\phi)$ and $p_{\mathrm{true}}(\psi)$. In other words, $\pi(\phi,\psi)$ satisfies $\int \pi(\phi,\psi)\,d\phi = p_{\mathrm{true}}(\psi)$ and $\int \pi(\phi,\psi)\,d\psi = p_\theta(\phi)$. The cost function $D(|\phi\rangle, |\psi\rangle)$ is the trace distance between the pure states $|\phi\rangle$ and $|\psi\rangle$.

\begin{theorem}\label{thm:theorem1}
Distribution Convergence Theorem. Let $p_{\mathrm{true}}$ be the true initial-state distribution, and let $p_\theta$ be the distribution obtained by running the reverse diffusion process for time $T$ under the learned control unitary $V_\theta$. Assume $V_\theta$ satisfies  
$\EE_{\psi_{t-\dd t},\psi_t}\sqrt{1-
\left|\langle \psi_{t-dt}|V_\theta(\boldsymbol{z}_t, t)|\psi_t\rangle \right|^2} \le \epsilon dt$,
and that it is Lipschitz continuous such that for any pure state $|\alpha\rangle$, 
$D(V_\theta^\dagger(\boldsymbol{z}, t)|\alpha\rangle,V_\theta^\dagger(\boldsymbol{z}', t)|\alpha\rangle)\le L_Vdt D(|\psi\rangle,|\psi^{\prime}\rangle)$, where $\boldsymbol{z}$ and $\boldsymbol{z}'$ are the state vectors of $|\psi\rangle$ and $|\psi^{\prime}\rangle$, respectively. Then by choosing a stopping time $T = \frac{C}{L_v}\ln(1/\epsilon)$ with $0 < C < 1$, the Wasserstein-1 distance between $p_\theta$ and $p_{\mathrm{true}}$ is bounded by 
\begin{equation}\label{eq-w-bound}
W_{1}(p_{0}, p_{\mathrm{true}}) 
\leq \frac{\sqrt{2}}{L_{v}}  \epsilon^{1-C} 
+ \epsilon^{-C} \, E_{\mathrm{diff}}(T)
\end{equation}
with $E_{\mathrm{diff}}(T)$ decaying exponentially in $T$ as $ne^{-\frac{4\gamma}{3n}T}$. If we choose a sufficiently large measurement strength $\gamma$ such that $\frac{4\gamma}{3n L_{v}} \ge 1$, the leading-order behavior of the bound for small $\epsilon$ is dominated by the first term in \cref{eq-w-bound}, and $
W_{1}(p_{0}, p_{\mathrm{true}}) \;\to\; 0$ as $\epsilon \to 0.$
\end{theorem}

\emph{Proof sketch}. The $1$-Wasserstein distance $W_1(p_{\theta}, p_{\mathrm{true}})$ between the learned and true distributions accumulates over time due to the small per-step discrepancy $\epsilon dt$ between their reverse diffusion processes. At each infinitesimal time step, this deviation is amplified at most by the Lipschitz factor $e^{L_V dt}$. Meanwhile, $E_{\mathrm{diff}}(T)$ quantifies how far the distribution obtained by forward diffusion up to time $T$ is from a uniform random product-state distribution, and this deviation is exponentially small in $T$. Combining these results and choose $T = \frac{C}{L_v}\ln(1/\epsilon)$ with $0 < C < 1$ yields the bound stated in the theorem. Full proof details are provided in SM \cite{suppl}.

The above theorem means that the distribution distance $W_1 \to 0$ when the training error $\epsilon$ is made arbitrarily tiny.

\subsubsection{Numerical Demonstration} 

To demonstrate the effectiveness of control-Hamiltonian learning, we consider a single-qubit example, as illustrated in \cref{1qubit-distribution} and \cref{1qubit-trajectory}, along with two two-qubit examples shown in \cref{2qubit-bell}. The first two-qubit example takes the two-qubit spin singlet $\psi_B = (|\uparrow\downarrow\rangle - |\downarrow\uparrow\rangle)/\sqrt{2}$ as a fixed target initial state---the recovery problem here is to reconstruct this single rank-one Bell state from the forward-diffused records, providing the most stringent test of the controller on a maximally entangled target. The details of the two-qubit Heisenberg model can be found in the caption of \cref{2qubit-bell}.

\begin{figure}[htbp]
\begin{center}
\includegraphics[width=0.8\linewidth]{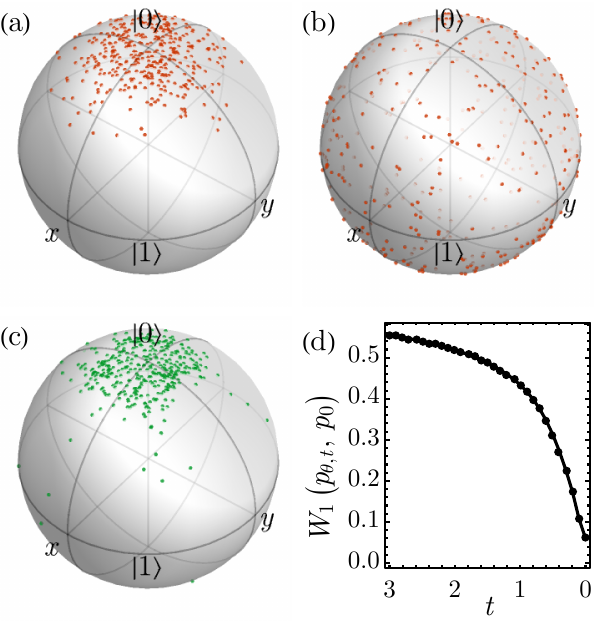}
\caption{Single qubit example with $T=3, \gamma=1, dt=0.01.$ (a)  Training-state ensemble $\mathcal{E}_\psi=\{\ket{\psi}\!\bra{\psi}: \psi \sim p(\psi)\}$ with $|\psi\rangle = (|0\rangle + |\delta\psi\rangle)/\text{normalization},$ where $|\delta\psi\rangle=\alpha_0|0\rangle+\alpha_1|1\rangle$ and $\alpha_{0,1}$ are complex Gaussian variables (standard deviation 0.2), normalized to unit norm. (b) Initial distribution of 400 samples for backward diffusion. (c) Final distribution of these samples after backward diffusion. (d) Decay of the Wasserstein distance $W(p_{\theta,t}, p_{\mathrm{true},0})$ versus time $t.$}
\label{1qubit-distribution}
\end{center}
\end{figure}

\begin{figure}[htbp]
\begin{center}
\includegraphics[width=0.92\linewidth]{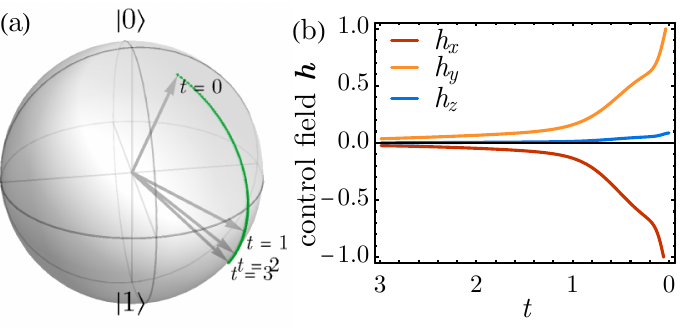}
\caption{One example trajectory for the single qubit example in \cref{1qubit-distribution}. (a) Single backward-diffusion trajectory. (b) Corresponding time evolution of the coefficients in $H_\theta=h_x\sigma_x + h_y\sigma_y + h_z\sigma_z$.}
\label{1qubit-trajectory}
\end{center}
\end{figure}

\begin{figure}[htbp]
\begin{center}
\includegraphics[width=0.64\linewidth]{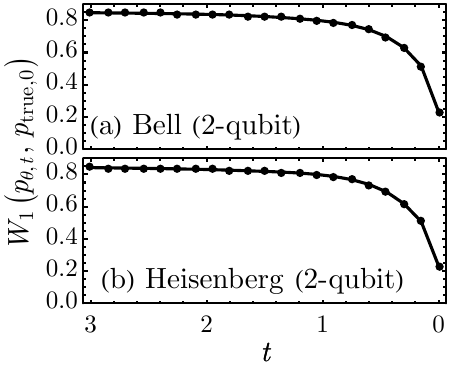}
\caption{Two qubit examples with $T=3, \gamma=1, dt=0.01.$ (a) $W_1(p_{\theta,t}, p_{\mathrm{true},0})$ versus time $t$ for the two-qubit spin singlet (Bell state) $|\psi_B\rangle = \frac{1}{\sqrt{2}}\left( | \uparrow\downarrow \rangle - | \downarrow\uparrow \rangle \right)$ taken as the fixed target initial state. (b) $W_1(p_{\theta,t}, p_{\mathrm{true},0})$ versus time $t$ for a thermal state ensemble of a two-qubit Heisenberg model $H=J\boldsymbol{\sigma}_1\cdot\boldsymbol{\sigma}_2+B_x\sigma_{1,x}+B_z\sigma_{2,x}$ with $B_x=B_z=0.5$ and $J=1$, in a temperature range $T_H\in[0,0.5]$.}
\label{2qubit-bell}
\end{center}
\end{figure}

\subsection{Classical Shadow Reconstruction}\label{sec:shadow}

\subsubsection{Problem Statement}
We now pivot from trajectory-wise reversal to \emph{ensemble-average} recovery.  
In this setting, the data are no longer a single measurement record $\mathcal{O}=\{(O_t,o_t)\}$ tied to one quantum trajectory, but a collection of independent and identically distributed records $\{\mathcal{O}^{(s)}\}_{s=1}^M$ generated by repeatedly performing randomized weak-measurement protocol on the unknown initial state $\bar{\rho}_0$.  
The learning task is to reconstruct the $\bar{\rho}_0$ given the measurement record dataset.  

This is a problem of learning quantum states from classical data.  
When the objective is inference---rather than coherent state preparation on a quantum device---\emph{classical shadow tomography} \cite{Huang2020P2002.08953} provides an efficient solution: it associates to each measurement record a classical shadow that is an unbiased randomized estimator of $\bar{\rho}_0$, and averaging classical shadows yields a reconstruction of $\bar{\rho}_0$ with rigorous concentration guarantees and favorable sample complexity.

\subsubsection{Measurement-and-Prepare Channel}
For each measurement record $\mathcal{O}$ obtained from forward diffusion over time $t$,  
define the \emph{measurement operator}
\begin{equation}
\sigma_t(\mathcal{O}) = K_t^\dagger(\mathcal{O})\, K_t(\mathcal{O})    
\end{equation}
in terms of the Kraus operator $K_t(\mathcal{O})$ defined in \cref{eq:K_t},
such that the probability of observing $\mathcal{O}$ on an initial state $\bar{\rho}_0$ is
\begin{equation}
    p(\mathcal{O}\mid \bar{\rho}_0) = \Tr( \sigma_t(\mathcal{O})\,\bar{\rho}_0 ),
\end{equation}
referred to as the \emph{posterior distribution} of $\mathcal{O}$. 

Classical shadow tomography begins by repeatedly measuring identical copies of $\bar{\rho}_0$,  
thereby generating measurement records $\mathcal{O}$ drawn from the posterior distribution.  
Collecting $M$ such measurement shots produces the dataset  
$\{\mathcal{O}^{(s)}\}_{s=1}^M$.

In the classical post-processing stage, each record $\mathcal{O}$ is mapped to its corresponding measurement operator $\sigma_t(\mathcal{O})$, regarded as a \emph{classical snapshot}.  
This process can be described as a \emph{measure-and-prepare channel}:
\begin{equation}
\begin{split}
\mathcal{M}_t(\bar{\rho}_0) &= 
\sum_{\mathcal{O}} \sigma_t(\mathcal{O}) \,\Tr\!\left[ \sigma_t(\mathcal{O})\,\bar{\rho}_0 \right]\\
&= \EE_{\mathcal{O} \sim p(\mathcal{O} \mid \bar{\rho}_0)} \sigma_t(\mathcal{O}).    
\end{split}
\end{equation}
Here we intentionally do not normalize $\mathcal{M}_t$ (so it is not trace-preserving), as this choice simplifies the analysis.

The map $\mathcal{M}_t$ is a linear transformation from the target state $\bar{\rho}_0$ to the average classical snapshot.  
If the measurement scheme is \emph{tomographically complete}---as in the randomized Pauli measurement case---$\mathcal{M}_t$ is invertible, and the state can be reconstructed as
\begin{equation}
\bar{\rho}_0
= \EE_{\mathcal{O} \sim p(\mathcal{O} \mid \bar{\rho}_0)}
\,\mathcal{M}_t^{-1}\!\left[ \sigma_t(\mathcal{O}) \right].
\end{equation}
Estimating the ensemble average requires samples from the quantum device,  
while the post-processing step constructs classical snapshots and applies $\mathcal{M}_t^{-1}$.  
This is the basic principle of {classical shadow reconstruction}.  
The problem thus reduces to determining the inverse map $\mathcal{M}_t^{-1}$, also called the \emph{reconstruction map}.

\subsubsection{Channel Inversion and Pauli Weight}

In general, computing $\mathcal{M}_t^{-1}$ is challenging,  
but our weak-measurement scheme allows an exact solution thanks to two key features:
\begin{enumerate}
    
\item[(i)] \emph{Pauli Twirling}.  
   In randomized Pauli measurement, the observable distribution is invariant under local Clifford unitaries, implying that $\mathcal{M}_t$ is Pauli-twirled.  
   Therefore, $\mathcal{M}_t$ is diagonal in the Pauli basis:
   \begin{equation}
   \mathcal{M}_t(P) = w_{\mathcal{M}_t}(P)\, P,
   \end{equation}
   where the \emph{Pauli weight} (as the channel eigenvalue) is defined as
   \begin{equation}
   w_{\mathcal{M}_t}(P) = \frac{\Tr ( P\, \mathcal{M}_t(P) )}{\Tr \id}.
   \end{equation}

\item[(ii)] \emph{Locality}.  
   The randomized Pauli measurement is a local measurement scheme,  
   meaning that weak measurements act independently on each qubit.  
   This locality enables us to compute the Pauli weight for each qubit separately and combine the results multiplicatively.  
   The final closed-form expression is
   \begin{equation}
   w_{\mathcal{M}_t}(P) = w_{\mathcal{M}_t}(\id)
   \left( \frac{1 - \ee^{-\frac{16\gamma t}{3n}}}{3 + \ee^{-\frac{16\gamma t}{3n}}} \right)^{|P|},
   \end{equation}
   where $|P|$ denotes the operator weight (support size) of the Pauli string $P$, and the factor
   \begin{equation}
       w_{\mathcal{M}_t}(\id)=\left(3\,\ee^{\frac{4\gamma t}{3n}}+\ee^{-\frac{4\gamma t}{n}}\right)^n
   \end{equation}
   is a time-dependent function.
\end{enumerate}
With these results, the reconstruction map is
\begin{equation}
\mathcal{M}_t^{-1}(P) = \frac{1}{w_{\mathcal{M}_t}(P)}\,P.
\end{equation}

\subsubsection{Estimating Observables and Sample Complexity}
For finite sample size $M$, the Pauli expectation value can be estimated as
\begin{equation}
z_{0,i} \equiv \Tr\!\left( \bar{\rho}_0\,P_i \right)
\approx \frac{1}{M} \sum_{s=1}^M \frac{\Tr\!\left( \sigma_t(\mathcal{O}^{(s)})\,P_i \right)}{w_{\mathcal{M}_t}(P_i)},
\end{equation}
from which the initial ensemble-average state $\bar{\rho}_0$ can be formally recovered as 
\begin{equation}
\bar{\rho}_0=\frac{\vect{z_0}\cdot\vect{P}}{2^n}\approx \frac{1}{2^n M} \sum_{s=1}^M \frac{\Tr\!\left( \sigma_t(\mathcal{O}^{(s)})\,P_i \right)\,P_i}{w_{\mathcal{M}_t}(P_i)},
\end{equation}
effectively realizing the reverse diffusion via classical post-processing the forward diffusion data $\{\mathcal{O}^{(s)}\}_{s=1}^{M}$.

The statistical variance associated with $z_{0,i}$ scales as
\begin{equation}
\mathrm{Var}[z_{0,i}] \sim \frac{\|P_i\|_{\mathrm{sh}}^2}{M},
\end{equation}
where $\|P\|_{\mathrm{sh}}^2$ denotes the \emph{shadow norm} and is given by
\begin{equation}
\|P\|_{\mathrm{sh}}^2=\frac{\Tr\!\left(P\mathcal{M}_t^{-1}(P)\right)}{\Tr\!\left(\mathcal{M}_t^{-1}(\id)\right)}=\frac{w_{\mathcal{M}_t}(\id)}{w_{\mathcal{M}_t}(P)}.
\end{equation}
Therefore, to achieve a target variance $\epsilon^2$ for a given Pauli observable $P$, the required sample size scales as
\begin{equation}
M \sim \frac{\|P\|_{\mathrm{sh}}^2}{\epsilon^2}.
\end{equation}

Physically, the Pauli weight $w_{\mathcal{M}_t}(P)$ quantifies the rate that measurement extracts information about the Pauli observable $P$.  
In reconstruction, the expectation value of $P$ must be reweighted by $1/w_{\mathcal{M}_t}(P)$, which necessarily amplifies statistical fluctuations.  
A smaller Pauli weight indicates less information per sample and thus a larger amplification factor and a larger sample complexity.

The shadow norm exhibits characteristic time dependence:  
for short times ($t\to 0$), it diverges as $\mathcal{O}(1/t)$, reflecting that weak measurements initially acquire little information about $\bar{\rho}_0$, hence requiring many samples for accurate reconstruction.  
As $t$ increases, the information per trajectory grows linearly before saturating.  
In the long-time limit ($t\to\infty$), the shadow norm converges to $3^{|P|}$,  
matching the scaling of projective Pauli measurements.

\subsection{Petz Recovery}

\subsubsection{Global Petz Recovery}
At the average state level, reversing quantum diffusion amounts to inverting the weak measurement channel $\mathcal{O}_t$ that governs the evolution of the average state $\bar{\rho}_t$. 
A general, near-optimal method for channel inversion is the \emph{Petz recovery map} \cite{petz1986sufficient,Kwon_2019,Kwon2022R2104.03360}, which guarantees exact recovery when $\mathcal{F}_t$ is reversible on the support of $\bar{\rho}_0$ and provides strong theoretical performance guarantees for approximate cases.

For an infinitesimal time step $\dd t$, the Petz recovery map can be expressed in the ``twirled'' integral form
\begin{equation}\label{eq:petz-twirled}
\begin{split}
&\widetilde{\mathcal{R}}_{\dd t}(\sigma)=\int_{-\infty}^\infty f(\tau) \mathcal{R}_{\dd t}^\tau(\sigma)\dd\tau\ ,\\ &\mathcal{R}_{\dd t}^\tau(\sigma)=\bar{\rho}^{\frac{1-i\tau}{2}}_{t}\mathcal{F}_{\dd t}^\dagger \left(\bar{\rho}_{t+\dd t}^{\frac{-1+i\tau}{2}}\sigma \bar{\rho}_{t+\dd t}^{\frac{-1-i\tau}{2}}\right)\bar{\rho}^{\frac{1+i\tau}{2}}_{t}\ ,
\end{split}
\end{equation}
where $f(\tau) = \frac{1}{2(\cosh(\pi\tau)+1)}$ and $\mathcal{F}_{\dd t}^\dagger=\ee^{\dd t\mathcal{L}^\dagger}$ denotes the adjoint channel of $\mathcal{F}_{\dd t}$, and the Liouvillian super-operator $\mathcal{L}=\mathcal{L}^\dagger$ is self-adjoint as defined in \cref{eq:pde}. Physically, \cref{eq:petz-twirled} describes how to undo the decoherence from $\mathcal{F}_{dt}$ using the eigenbasis of $\bar{\rho}_t$ as a prior. The forward channel $\mathcal{F}_{dt}$ slightly suppresses off-diagonal elements of $\bar{\rho}_t$ in its eigenbasis, corresponding to a drift toward the maximally mixed state. The Petz map applies the opposite drift by reweighting the eigencomponents of $\sigma$ relative to $\bar{\rho}_t$, thereby restoring the lost coherence. Iterating $\widetilde{\mathcal{R}}_{dt}$ over the full evolution time yields a reconstruction of the initial average state from the final diffused state.

\subsubsection{Local Petz Recovery}
While the \emph{global} Petz map is an exact recovery channel in theory, it is generally intractable for large many-body systems due to its fully nonlocal action.  
A practical alternative arises when the ensemble exhibits predominantly short-range correlations, characterized by a finite \emph{Markov length} $\xi$: correlations between a region $A$ and the distant complement $C$ decay rapidly once a buffer region $B$ (of width $\xi$) is included, as shown in \figref{fig:petz}(a).  
Formally, this corresponds to the conditional mutual information (CMI) $I_\rho(A\!:\!C | B)$ being small for appropriate tripartitions $(A,B,C)$.  
In this near-Markovian regime, the global state is well-approximated by a quantum Markov chain, and one can replace the global Petz map by a \emph{local} version acting only on $S_j = A \cup B$.

The local Petz map $\widetilde{\mathcal{R}}_{\dd t,S_j}$ is simply constructed by restricting the {twirled Petz} formula \cref{eq:petz-twirled} to $S_j$, using the reduced density matrix $\rho_{t,S_j}$ as the prior and $\mathcal{F}_{dt,S_j}$ as the local measurement channel.
Overlapping regions $\{S_j\}$ are chosen to cover the system, and the local recoveries $\widetilde{\mathcal{R}}_{dt,S_j}$ are applied in sequence or in parallel following the scheme of Ref.~\onlinecite{sang2024stabilitymixedstatequantumphases}.  
This yields a finite-depth recovery circuit implementable on quantum hardware, with accuracy controlled by the residual CMI.

\begin{theorem}[Local Petz Recovery for Quantum Diffusion with Finite Markov Length]\label{thm:theorem2}
We consider an $n$-qubit state $\bar{\rho}_0$ evolving under a sequence of Trotterized measurement channels $\mathcal{F}_{\delta t}$ over time $T$, leading to a fully decohered state $\bar{\rho}_T$. By applying a sequence of local twirled Petz recovery maps---each acting on a local neighborhood of the measured qubit---we approximately reverse this process. Assuming exponential decay of conditional mutual information with distance for every intermediate $\bar{\rho}_t$, we show that for sufficiently large $T$ and recovery region size, the final reconstructed state $
\bar{\rho}_0'$ satisfies $\| \bar{\rho}_0' - \bar{\rho}_0 \|_1\leq \epsilon.$
\end{theorem}

\noindent\textit{Proof Sketch.}  
The error bound for Petz recovery is derived following the approach of \cite{sang2024stabilitymixedstatequantumphases}: one first bounds the error of a single recovery step in terms of the conditional mutual information (CMI), then sums over $N_T$ steps to obtain a telescoping bound. By the data-processing inequality, CMI is monotonic, implying that the total recovery error remains bounded by $\delta$. Meanwhile, repeated weak measurements drive the state towards a maximally mixed product state by rapidly damping higher-weight Pauli operators (see \cref{eq-measurement-channel-pauli-weight}), so the difference between the final state $\bar{\rho}_T$ and the maximally mixed product state decays exponentially with $T$. Combining these results yields the stated bound. Full proof details are provided in SM \cite{suppl}.
 

The following protocol constructs local Petz recovery maps from forward-diffusion measurement data and implements them as local quantum channels on quantum hardware:
\begin{enumerate}
\item[(i)] \emph{Collect data.}  
    Repeatedly prepare the unknown target average state $\bar{\rho}_0$ and apply the forward diffusion process via randomized weak measurements.  
    Record the complete measurement record $\mathcal{O} = \{(O_t, o_t)\}$ for each trajectory, forming a dataset $\{\mathcal{O}\}$.

\item[(ii)] \emph{Process data.}  
    Apply weak-measurement classical shadow tomography to the dataset $\{\mathcal{O}\}$ to reconstruct an estimate $\tilde{\rho}_0$ of the initial average state.  
    Using this estimate, infer the intermediate states $\tilde{\rho}_t = \mathcal{F}_t(\tilde{\rho}_0)$ via classical simulation of the known forward channel $\mathcal{F}_t$, see \figref{fig:petz}(b).

\item[(iii)] \emph{Compute local Petz maps.}  
    For each overlapping local region $S_j$ (with size exceeding the Markov length $\xi$), construct the local Petz recovery map $\widetilde{\mathcal{R}}_{\dd t, S_j}$ by restricting the \emph{twirled Petz} formula \cref{eq:petz-twirled} to $S_j$, using the reduced state $\tilde{\rho}_{t,S_j}$ and the corresponding local channel $\mathcal{F}_{\dd t, S_j}$.  
    Arrange and stack these local maps in the reverse order of the forward schedule to obtain the overall Petz recovery protocol $\widetilde{\mathcal{R}}_t$, as illustrated in \figref{fig:petz}(c) (see \cite{sang2024stabilitymixedstatequantumphases} for the detailed construction).

\item[(iv)] \emph{Implement reverse diffusion.}  
    Each step of $\widetilde{\mathcal{R}}_t$ is a local quantum channel that can, in principle, be implemented on quantum hardware.  
    Initialize the device in a state $\bar{\rho}'_T$ (e.g., a random product state evolved forward to time $T$), then apply $\bar{\rho}'_0 = \widetilde{\mathcal{R}}_T(\bar{\rho}'_T)$ to obtain a recovered approximation to the original state $\bar{\rho}_0$.  
    This procedure enables repeated state preparation directly on the quantum device.
\end{enumerate}

This protocol provides a systematic construction of the reverse process as a composition of local quantum channels, tailored for forward diffusion driven by randomized measurements.  
In contrast to prior approaches to average-state recovery based on learned reverse quantum channel \cite{Parigi2023Q2308.12013,Chen2024Q2401.07039,Kolle2024Q2401.07049,Zhang2024G2310.05866,Kwun2025M2411.17608,Huang2025C2506.19270,Zhang2026P}, the Petz recovery approach requires no learning: the recovery maps are constructed directly from data and knowledge of the forward channel.  
This observation offers an important theoretical guarantee for learning-based methods---namely, that a purely classical algorithm exists to build local recovery maps from measurement data.  
Therefore, any successful learning-based recovery can be viewed as approximating, in principle, a well-defined recovery process given by local Petz maps.

\subsubsection{Numerical Demonstration}
To illustrate the proposed protocol, we perform a numerical experiment on a 10-qubit transverse-field Ising chain with open boundary conditions,
\begin{equation}
H = -J \sum_{i} \sigma_{z,i} \sigma_{z,i+1} - B_x \sum_i \sigma_{x,i} ,
\end{equation}
using $J=1.0$, and $B_x=1.5, 2.0$ and $5.0$ as in \cref{fig:petz}.
This choice places the system well away from the critical point ($B_x=1.0$), ensuring a finite correlation length and a ground state close to the product state $\ket{+}^{\otimes n}$.
The chain is initialized in the ground state $\bar{\rho}_0$ of $H$ and evolved under randomized weak-measurement dynamics for a total time $T=10$ with step size $\dd t=0.01$ ($1000$ diffusion steps).

Local Petz maps $\widetilde{\mathcal{R}}_{dt,S_j}$ are constructed for all contiguous three-qubit regions $S_j$ (each exceeding the Markov length) and arranged in a constant-depth circuit, ensuring that the entire protocol inverts the forward diffusion approximately while remaining finite in depth.
The recovered state $\bar{\rho}'_0$ achieves fidelities of $0.911, 0.963,$ and $0.989$ with respect to $\bar{\rho}_0$ for $B_x=1.5, 2.0,$ and $5.0,$ respectively, as in \cref{fig:petz}.

\begin{figure}[t]
\centering
\includegraphics[width=0.95\linewidth]{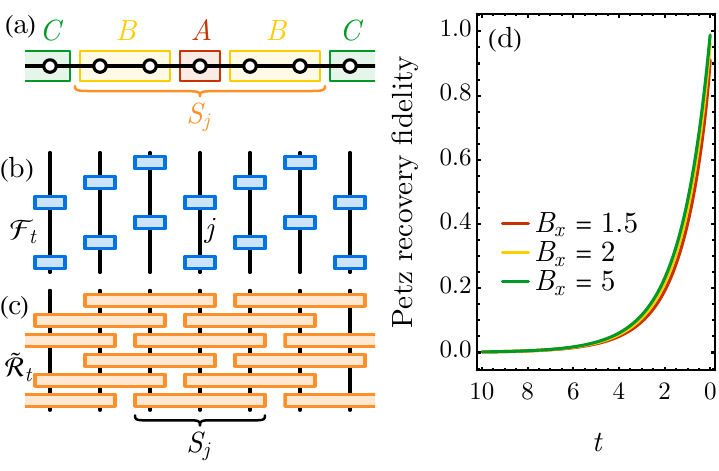}
    \caption{(a) Regions $A$ (qubit $j$), $B$, $C$ and $S_j$. When adjacent steps measure far-apart qubits, the circuit of (b) the forward weak measurement channel $\mathcal{F}$ and (c) the Petz recovery channel $\widetilde{\mathcal{R}}$ can be rearranged into a finite depth circuit. (d) Fidelity of Petz recovery at time $t$ with the initial state.}
    \label{fig:petz}
\end{figure}

\subsubsection{Petz Recovery as Time-Reversed Diffusion}  
We establish that the Petz recovery map provides a quantum generalization of classical time-reversed diffusion under weak measurements, and that in the classical large-spin limit the two become equivalent.

Consider an $n$-qubit separable state (no quantum entanglement), expressed as a mixture of tensor product of spin coherent states
\begin{equation}
\bar{\rho}_t = \int \dd\vect{n}\; p(\vect{n},t)\, \ket{\vect{n}}\!\bra{\vect{n}},
\end{equation}
where the spin configuration is described by $\vect{n} = (\vect{n}_1,\ldots,\vect{n}_n)$ with $|\vect{n}_j| = 1$, and the corresponding coherent state reads
\begin{equation}
|\vect{n}\rangle\langle\vect{n}| = \prod_{j=1}^n \frac{\id + \vect{n}_j \cdot \boldsymbol{\sigma}_j}{2},
\end{equation}
and $p(\vect{n},t)$ is a normalized probability density over the $n$-fold product of Bloch spheres.

Consider randomized Pauli weak measurements, whose action is equivalent to applying independent depolarizing channels on each qubit.  
Since such local measurements do not generate entanglement, every trajectory remains a pure product state throughout the process.  
Consequently, the ensemble evolution is fully captured by the probability density $p(\vect{n},t)$ over coherent state configurations $\vect{n}$.  
The resulting stochastic dynamics yield the forward Fokker--Planck equation  
\begin{equation}\label{eq:diffusion-product-states-n}
\partial_t p(\vect{n},t)=\frac{2\gamma}{3 n}\nabla_\perp^2 p(\vect{n},t),
\end{equation}
where $\nabla_\perp^2 = \sum_{j=1}^n \nabla_{\perp,j}^2$ is the sum of Laplacians on individual Bloch spheres.

Given the forward diffusion history of the average state 
$\bar{\rho}_t = \mathcal{F}_t(\bar{\rho}_0)$, 
the reverse process can be implemented as  
\begin{equation}
\bar{\rho}_0 = \mathcal{R}^0_t(\bar{\rho}_t),
\end{equation}
where $\mathcal{R}^0_t$ is the (untwirled) Petz map derived from \cref{eq:petz-twirled},  
\begin{equation}
\mathcal{R}^0_t(\sigma) 
= \bar{\rho}_t^{1/2}\,\mathcal{F}_t^\dagger\!\big(\bar{\rho}_{t+\dd t}^{-1/2}\,\sigma\,\bar{\rho}_{t+\dd t}^{-1/2}\big)\,\bar{\rho}_t^{1/2},
\end{equation}
which reweights $\sigma$ by $\bar{\rho}_t^{-1/2}$ through the adjoint channel.  
For $\sigma = \int q(\vect{n},t)\,\ket{\vect{n}}\!\bra{\vect{n}}\,\dd\vect{n}$, 
$\mathcal{R}^0_t(\sigma)$ induces a backward diffusion of the distribution $q(\vect{n},t)$ --- a quantum generalization of classical reverse diffusion.  
The PDE governing this backward evolution is given in SM~\cite{suppl}.

In the large-spin (classical) limit, coherent states with different $\vect{n}$ become orthogonal, $\langle\vect{n}|\vect{n}'\rangle \to \delta_{\vect{n},\vect{n}'}$. The Petz equation becomes the classical backward Fokker-Planck equation
\begin{equation}\label{eq:back-diff-classical}
\partial_t q(\vect{n},t) = \frac{2\gamma}{3n}\nabla_\perp^2 q(\vect{n},t) 
- \frac{4\gamma}{3n}\nabla_\perp\!\cdot\!\big[ q(\vect{n},t)\,\nabla_\perp \log p(\vect{n},t) \big],
\end{equation}
exactly reversing \cref{eq:diffusion-product-states-n}, with $\frac{4\gamma}{3n}\nabla_\perp \log p$ as the backward drift.

This result reveals a precise theoretical connection: in the classical limit, the Petz recovery map {is} time-reversed diffusion. This not only clarifies the physical meaning of Petz recovery, but also provides a rigorous bridge between quantum recovery channels and classical stochastic reversals, ensuring that any successful quantum Petz reconstruction has a well-defined classical counterpart.

\section{Conclusion and Discussion}

We developed a measurement-based framework for quantum diffusion that unifies trajectory-wise and ensemble-average reversal. On the forward side, randomized weak measurements generate stochastic pure-state trajectories while inducing a measurement channel that is locally scrambled and diagonal in the Pauli basis. On the reverse side, we presented two complementary routes: (i) trajectory-level reversal via score-matching/control--Hamiltonian learning; and (ii) ensemble-level reversal via either Petz recovery or classical shadow reconstruction.

\paragraph{Summary of contributions.}
(i) We formalized trajectory-level reverse diffusion as a control problem and showed its equivalence to quantum score matching, yielding a principled objective for learning the reverse flow. (ii) For ensemble averages, we introduced a practical \emph{local} Petz recovery when correlations are short-ranged, leveraging recoverability guarantees for Lindbladian dynamics and near-optimal Petz maps; this connects quantum channel inversion to reversing measurement-induced decoherence. 
(iii) We established a complementary \emph{classical} post-processing route via classical shadow tomography, exploiting that our Pauli-twirled measurement-and-prepare channel is diagonal in the Pauli basis; its inverse gives an unbiased estimator of $\bar{\rho}_0$ with rigorous concentration bounds for local observables. 
(iv) Conceptually, we clarified that Petz recovery implements a structured quantum analogue of time-reversed diffusion and reduces to classical backward diffusion in the large-spin/separable limit, linking recovery channels to Fokker-Planck reversals. This complements recent diffusion-style generative models for quantum data that operate in either state space or density-matrix space.

\paragraph{Broader implications.}
The measurement-driven formulation naturally couples data acquisition and control: the same weak-measurement records that drive forward diffusion also provide conditioning signals for reverse synthesis. Local Petz maps offer a hardware-friendly, finite-depth realization of ensemble reversal under finite Markov length, whereas classical shadows provide an efficient, device-agnostic alternative when the goal is inference rather than state preparation. Together, these tools place quantum diffusion models on firmer mathematical footing and connect them to established theories of recoverability and randomized tomography.

\paragraph{Limitations and outlook.}
Our trajectory-wise reversal presently relies on classical decoders (simulation- or learning-based) to infer sufficient state information for feedback, a key difference from classical denoising where the state is directly observable. Scaling such decoders, tightening end-to-end error bounds (from Hamiltonian learning error to ensemble Wasserstein/trace-distance guarantees), and extending beyond randomized Pauli measurements to other locally scrambled schemes (e.g., Clifford- or unitary-2-design--based protocols) are immediate next steps. It is also promising to explore recovery-aware measurement schedules and noise models that optimize sample complexity and circuit depth, and to interface these methods with structure-preserving diffusion for mixed states and with recent progress on designing open dynamics with prescribed steady states. 

Overall, this work provides a concise blueprint for quantum diffusion: use weak measurements to generate and learn from forward randomness; reverse trajectories via learned control when needed; and reverse ensembles via Petz recovery or classical shadows when sufficient statistics are available. We expect these ingredients to be directly useful in near-term platforms that already support weak, randomized measurements and local quantum channels.

\emph{Note Added}: Upon the completion of our work, we became aware of a complementary work by Hu et al.~\cite{hu2025local}, which also discussed local recovery mechanisms in diffusion processes. Their work also established a quantum-classical correspondence between Petz recovery maps and classical diffusion models, and uses conditional mutual information to determine the possibility of local denoising. While our work focuses on quantum state generation, their work addresses classical data generation through local denoising.

\begin{acknowledgments}
We acknowledge the helpful discussion with Quntao Zhuang, Shengqi Sang, Molei Tao, Lei Wang, Fangjun Hu, Xun Gao, Zihao Wang, Haimeng Zhao and Biao Lian. We thank Fangjun Hu and Xun Gao for sharing with us their unpublished work. JZ is supported by the Innovation Program for Quantum Science and Technology (No.2021ZD0301601). WH and YZY are supported by the NSF Grant No.~DMR-2238360. We also acknowledge the use of OpenAI's GPT\mbox{-}5 model for assistance in organizing the manuscript, refining the exposition, and aiding with \LaTeX{} typesetting; all interpretations, conclusions, and any remaining errors are solely the responsibility of the authors.
\end{acknowledgments}

\bibliography{QD_ref}

\onecolumngrid
\newpage
\appendix

\begin{center}
\textbf{\large Supplemental Material}
\end{center}
\vspace{0.5em}

\section{Derivations for the Score Matching of the General Diffusion Dynamics}

This appendix provides a detailed, self-contained derivation of the reverse-time dynamics for a general diffusion process and its application to the stochastic evolution of quantum states under weak measurement. We begin by outlining the general theory of stochastic differential equations (SDEs) and score-based modeling. We then connect this mathematical framework to the physical process of weak measurement. Finally, we establish a formal equivalence between learning the reverse diffusion process for an ensemble of pure states and learning an effective time-reversed unitary evolution.

\subsection{General Diffusion Dynamics}\label{sec:general-diffusion}

\subsubsection{The generic SDE, Fokker--Planck PDE, and ODE Descriptions}

A generalized diffusion process for a state vector $\boldsymbol{z}(t)$ in an $N$ dimensional space can be described by a stochastic differential equation (SDE) with a $M$-dimensional and state-dependent noise. The generic SDE takes the form (in which all quantities are real valued):

\begin{equation}\label{seq:general-sde-Mdim}
d\boldsymbol{z}_t = \boldsymbol{f}(\boldsymbol{z}_t, t) \mathrm{d}t + \boldsymbol{G}(\boldsymbol{z}_t, t) d\boldsymbol{w}_t\ ,
\end{equation}

where:

\begin{itemize}
  \item $\boldsymbol{z}_t$ is an $N$-dimensional \textbf{state vector} representing the system's state at time $t$.
  \item $\boldsymbol{f}(\boldsymbol{z}_t, t)$ is an $N$-dimensional \textbf{drift field} (drift force), describing the deterministic part of the evolution.
  \item $\boldsymbol{G}(\boldsymbol{z}_t, t)$ is an $N \times M$ matrix representing the \textbf{noise amplitude}, which couples the noise sources to the state vector.
  \item $d\boldsymbol{w}_t=(dw_{t,1},\cdots,dw_{t,M})^T$ is an $M$-dimensional vector representing a Wiener process \textbf{noise}, satisfying the a probability distribution (usually Gaussian) with the first and second moments ($\mathbb{E}$ stands for mean value): 
  $$\mathbb{E}[\mathrm{d}w_{t,i}] = 0\ ,\qquad \mathbb{E}[\mathrm{d}w_{t,i} \mathrm{d}w_{t,j}] = \delta_{ij} \mathrm{d}t\ .$$
\end{itemize}

The time evolution of the probability distribution $p(\boldsymbol{z}, t)$ of the state vector $\boldsymbol{z}$ is known to be governed by the corresponding Fokker--Planck partial differential equation (PDE):

\begin{equation}\label{seq:forward-PDE}
\partial_t p(\boldsymbol{z}, t) = -\nabla\cdot \left[\boldsymbol{f}(\boldsymbol{z}, t) p(\boldsymbol{z}, t)\right]
+ \frac{1}{2} \partial_i \partial_j \left(\Sigma_{ij}(\boldsymbol{z}, t) p(\boldsymbol{z}, t)\right)\ ,
\end{equation}
where Einstein summation convention over repeated indices are assumed, and we have introduced the $N\times N$ noise covariance matrix
\begin{equation}
 \boldsymbol{\Sigma}(\boldsymbol{z}, t) = \boldsymbol{G}(\boldsymbol{z}, t) \boldsymbol{G}^T(\boldsymbol{z}, t)\ .
\end{equation}

Using the probability distribution $p(\boldsymbol{z}, t)$, the SDE of state vector $\boldsymbol{z}_t$ is statistically equivalent to a deterministic ordinary differential equation (ODE), known as the probability flow ODE:
\begin{equation}
d\boldsymbol{z}_t = \boldsymbol{v}(\boldsymbol{z}_t, t) \mathrm{d}t\ ,
\end{equation}
with the velocity field $\boldsymbol{v} =(v_1,\cdots,v_N)^T$ given by:
\begin{equation}
v_i(\boldsymbol{z}, t) = f_i(\boldsymbol{z}, t) 
- \frac{1}{2} \Big[\partial_j \Sigma_{ij}(\boldsymbol{z}, t)
+ \Sigma_{ij}(\boldsymbol{z}, t) \partial_j \log p(\boldsymbol{z}, t)\Big]\ ,
\end{equation}
or more concisely in vector form,
\begin{equation}
\boldsymbol{v}(\boldsymbol{z}, t)=\boldsymbol{f}(\boldsymbol{z}, t)-\frac{1}{2}\Big[\nabla\cdot\boldsymbol{\Sigma}(\boldsymbol{z}, t)+ \boldsymbol{\Sigma}(\boldsymbol{z}, t)\cdot \nabla \log p (\boldsymbol{z}, t)\Big]\ ,
\end{equation}
where $(\nabla \cdot \boldsymbol{\Sigma})_i = \partial_j \Sigma_{ij}$.

\subsubsection{Reverse Diffusion and Score Matching}

A key result from stochastic calculus is that the forward diffusion process described by Eq.~(A1) has a corresponding reverse-time process. The reverse SDE, which allows for generating samples from the distribution $p(\boldsymbol{z}, 0)$ by starting from a simple prior at time $T$, at time $\widetilde{t}=T-t$ is given by:

\begin{equation}\label{seq:reverse-SDE-M-dim}
d\boldsymbol{z}_{\widetilde{t}} = \left[ -\boldsymbol{f}(\boldsymbol{z}_{\widetilde{t}},{\widetilde{t}}) + \nabla\cdot\boldsymbol{\Sigma}(\boldsymbol{z}_{\widetilde{t}}, {\widetilde{t}})+ \boldsymbol{\Sigma}(\boldsymbol{z}_{\widetilde{t}}, {\widetilde{t}})\cdot \nabla \log p (\boldsymbol{z}_{\widetilde{t}}, {\widetilde{t}}) \right] \mathrm{d}{\widetilde{t}} + \boldsymbol{G}(\boldsymbol{z}_{\widetilde{t}}) d\boldsymbol{w}_{\widetilde{t}}\ .
\end{equation}
where $d\boldsymbol{w}_{\widetilde{t}}$ is a standard Wiener process in the reverse time direction.

The generative power of this framework hinges on knowledge of the score function incorporating the information of $\nabla \log p(\boldsymbol{z}, t)$. In score-based modeling, we define a learnable model $s_\theta(\boldsymbol{z}, t)$ with variational parameters $\theta$ to approximate the score-dependent part of the reverse drift:
\begin{equation}
s_\theta(\boldsymbol{z}, t) \approx \nabla\cdot\boldsymbol{\Sigma}(\boldsymbol{z}, t)+ \boldsymbol{\Sigma}(\boldsymbol{z}, t)\cdot \nabla \log p (\boldsymbol{z}_t, t)\ ,
\end{equation}
by which one can implement the reverse SDE in \cref{seq:reverse-SDE-M-dim} by the score function $s_\theta(\boldsymbol{z}, t)$. The model is trained by minimizing the score-matching loss, which is the expected squared deviation between the model and the true score term:

\begin{equation}\label{seq:loss-M-dim}
\mathcal{L}_\theta = \mathbb{E}_{\boldsymbol{z}_t \sim p(\boldsymbol{z}, t)} \left[ \left\| s_\theta(\boldsymbol{z}_t, t) - (\nabla \cdot \boldsymbol{\Sigma})(\boldsymbol{z}_t) - \boldsymbol{\Sigma}(\boldsymbol{z}_t) \nabla \log p(\boldsymbol{z}, t) \right\|^2 \right]
\end{equation}
Directly computing this loss is often intractable. A practical and equivalent alternative is denoising score matching. Given a forward step $\boldsymbol{z}_{t + \mathrm{d}t} = \boldsymbol{z}_t + \boldsymbol{f}(\boldsymbol{z}_t) \mathrm{d}t + \boldsymbol{g}(\boldsymbol{z}_t) d\boldsymbol{w}_t$, it can be shown that the above loss function in \cref{seq:loss-M-dim} is equivalent to the following denoising loss function up to a constant:
\begin{equation}\label{seq:loss-denoising-M-dim}
\mathcal{L}_\theta = \frac{1}{2} \mathbb{E}_{\boldsymbol{z}_t, \boldsymbol{z}_{t + \mathrm{d}t}} \left[ \left\| s_\theta(\boldsymbol{z}_{t + \mathrm{d}t}-\boldsymbol{f}(\boldsymbol{z}_t) \mathrm{d}t, t) + \frac{\boldsymbol{z}_{t + \mathrm{d}t} - \boldsymbol{z}_t - \boldsymbol{f} \mathrm{d}t}{\mathrm{d}t} \right\|^2 \right]\ .
\end{equation}
This objective trains the model to predict the noise term that corrupted the state, thereby implicitly learning the score without needing to compute it directly.

Lastly, starting from a backward probability distribution $q(\boldsymbol{z},{\widetilde{t}})$, the backward diffusion PDE driven by the forward probability distribution $p(\boldsymbol{z},{\widetilde{t}})$ takes the form
\begin{equation}\label{seq:reverse-PDE}
\partial_{\widetilde{t}} q(\boldsymbol{z}, \widetilde{t}) = -\nabla\cdot \left[\left(-\boldsymbol{f}(\boldsymbol{z}, \widetilde{t})+\nabla\cdot\boldsymbol{\Sigma}(\boldsymbol{z}, \widetilde{t})+ \boldsymbol{\Sigma}(\boldsymbol{z}, \widetilde{t})\cdot \nabla \log p (\boldsymbol{z}, \widetilde{t})\right) q(\boldsymbol{z}, \widetilde{t})\right]
+ \frac{1}{2} \partial_i \partial_j \left(\Sigma_{ij}(\boldsymbol{z}, \widetilde{t}) q(\boldsymbol{z}, \widetilde{t})\right)\ ,
\end{equation}
where $\nabla\cdot\boldsymbol{\Sigma}(\boldsymbol{z}, \widetilde{t})+ \boldsymbol{\Sigma}(\boldsymbol{z}, \widetilde{t})\cdot \nabla \log p (\boldsymbol{z}, \widetilde{t})$ plays the role of a drifting force towards the distribution $p(\boldsymbol{z},{\widetilde{t}})$.



\subsection{The nonlinear SDE from the Weak Measurement Protocol}

\subsubsection{The Weak Measurement Protocol}\label{sec-weak-measurement-protocol}
We consider a quantum system initially prepared in the state $\rho_0$. Ancilla-assisted weak measurement of an observable $O_t$ is implemented by coupling the system to an ancilla qubit, prepared in $|0\rangle_A$, via the interaction Hamiltonian 
\begin{equation}
H = \lambda O_{t} \sigma_A^y\ .
\end{equation}
After evolving for a finite time $\delta t$, the ancilla is projectively measured in the $\sigma_A^x$ basis, yielding an outcome $o_t $. In the continuous limit, where $\delta t \to 0$ and $\lambda \to \infty$ such that the measurement rate $\gamma = \lambda^2 \delta t$ is constant, this interaction is described by the Kraus operator:
\begin{equation}
\label{eq:K_delt}
K_{\delta t}(O_t,\delta o_t)
\;=\;
\frac{1}{\sqrt{N_{O_t}}}\,
\exp\!\Big(-\gamma\,\delta t +2\gamma\,O_t\,\delta o_t\Big),
\end{equation}
where $N_{O_t}$ is a normalization factor such that $\sum_{\delta o_t} K^{\dagger}_{\delta t}(O_t,\delta o_t)K_{\delta t}(O_t,\delta o_t)=\id.$

\subsubsection{Forward Measurement-Based Diffusion}

Assuming the observable $O_t=O_t^\dagger$ is Hermitian, the corresponding weak measurement of observable $O_t$ in time $\delta t$ with measurement outcome $o_t$ will update the system's state via a Kraus operator $K_{\delta t}(O_t,\delta o_t)$ into:
\begin{equation}\label{eq-state-evolution}
\rho_{t+dt} = \frac{K_{\delta t}(O_t,\delta o_t)\, \rho_t\, K_{\delta t}^\dagger(O_t,\delta o_t)}{\mathrm{Tr} \left[ K_{\delta t}(O_t,\delta o_t)\, \rho_t\, K_{\delta t}^\dagger(O_t,\delta o_t) \right]}\ .
\end{equation}
The probability of observing outcome $o_t$ is state-dependent, and the mean value of $o_t$ from weak measurement defined above is given by
\begin{equation}
\EE[o_t]=2\sqrt{\gamma \delta t}\langle O_t\rangle\ ,
\end{equation}
where we denote $\langle O_{t} \rangle=\langle \psi_t|O_t|\psi_t\rangle$ for short. For instance, for one step of weak measurement of spin $1/2$ operator $O_t$ with outcomes $o_t=\pm1$ in time $dt$, the probability of obsereving outcome $o_t$ is:
\begin{equation}
p(o_t | \rho_t) = \mathrm{Tr} \left[ K_{\delta t}(O_t,\delta o_t)\, \rho_t\, K_{\delta t}^\dagger (O_t,\delta o_t) \right] = \frac{1}{2} + o_t \sqrt{\gamma \delta t} \langle O_{t} \rangle= \frac{1}{2} + 2\gamma \langle O_{t} \rangle \delta o_t\ ,\qquad \delta o_t= \frac{1}{2}\sqrt{\frac{\delta t}{\gamma}} o_t\ .
\end{equation}

\subsubsection{The Stochastic Process from Weak Measurement}
The evolution of the quantum state $\rho_t$ is found by expanding the update rule in \cref{eq-state-evolution} to the first order in $\delta t$ (for the derivations that follow, we denote the time increment as $\delta t= dt$ for convenience), after which we arrive at the nonlinear SDE for the normalized state:
\begin{equation}\label{eq-sde-rho}
d\rho_t = -\frac{\gamma}{2} [O_{t}, [O_{t}, \rho_t]] dt + \sqrt{\gamma} \{ O_{t} - \langle O_{t} \rangle, \rho_t \} dw_t
\end{equation}
where
\begin{equation}
 dw_t = o_t \sqrt{dt} - 2 \sqrt{\gamma}dt \langle O_{t} \rangle =2 \sqrt{\gamma}\big( \dd o_t  -  \langle O_{t} \rangle dt \big)  
\end{equation}
can be treated as a standard Wiener noise term satisfying $\mathbb{E}[dw_t] = 0$ and $\mathbb{E}[dw_t^2] = dt$. The first term in \cref{eq-sde-rho} is a deterministic drift representing decoherence, while the second is a stochastic term describing the random walk from quantum state collapse. If the state is pure, $\rho_t = |\psi_t\rangle \langle \psi_t|$, this process preserves purity, and the nonliear SDE in \cref{eq-sde-rho} is equivalent to the following evolution of the quantum state:
\begin{equation}\label{eq-pure-sde}
d|\psi_t\rangle = \left( -\frac{\gamma}{2} (O - \langle O \rangle)^2 dt + \sqrt{\gamma} (O - \langle O \rangle) dw_t \right) |\psi_t\rangle\ .
\end{equation}

\subsubsection{The nonlinear SDE in the Pauli basis}
The nonlinear SDE for the weak measurement in \cref{eq-sde-rho} can be written more explicitly by expanding the $n$-qubit density matrix in the Pauli basis: 
\begin{equation}
\rho_t = \frac{1}{2^n} \sum_i z_i P_i=\frac{1}{2^n}\boldsymbol{z}\cdot\boldsymbol{P}\ ,
\end{equation}
where $P_i$ ($1\le i\le N$) runs over all the $N=4^n$ Pauli string operators, with $P_0=I$ being the identity operator, $\boldsymbol{P}=(P_1,\cdots,P_N)^T$, and $z_i = \mathrm{Tr}(\rho P_i)$. This defines a state vector $\boldsymbol{z}=(z_1,\cdots,z_{N})^T=\mathrm{Tr}(\rho \boldsymbol{P})$ for the density matrix $\rho$. By taking trace of the nonlinear SDE in \cref{eq-sde-rho} with vector operator $\boldsymbol{P}$, and denoting the state vector at time $t$ by $\boldsymbol{z}_t$, we can rewrite it into a nonlinear SDE in the Pauli basis as:
\begin{equation}
d\boldsymbol{z}_t = \boldsymbol{f}(\boldsymbol{z}_t, t) \mathrm{d}t + \boldsymbol{g}(\boldsymbol{z}_t, t) dw_t\ ,
\end{equation}
where the the drift vector $\boldsymbol{f}(\boldsymbol{z}, t)$ and the noise vector $\boldsymbol{g}(\boldsymbol{z}, t)$ are defined by
\begin{equation}\label{seq:drift-noise-vector}
\boldsymbol{f}(\boldsymbol{z}_t, t)=-\frac{\gamma}{2} \text{Tr}\Big([O_{t}, [O_{t}, \rho_t]]\boldsymbol{P}\Big)\ ,\qquad \boldsymbol{g}(\boldsymbol{z}_t, t)=\sqrt{\gamma} \text{Tr}\Big( \{ O_{t} - \langle O_{t} \rangle, \rho_t \}\boldsymbol{P}\Big)\ ,
\end{equation}

Let $\vect{P}=\{P_i\}_{i=1}^{4^n}$ denote the $n$-qubit Pauli-operator basis with $\text{Tr}(P_i P_j)=2^n\delta_{ij}$.
Any state $\rho_t$ or observable $O_t$ can be expanded as
\begin{equation}
\label{eq:pauli-expansion}
\rho_t = \frac{\boldsymbol{z}_t\cdot\boldsymbol{P}}{2^n},\quad
O_t=\boldsymbol{x}_t\cdot \boldsymbol{P}, 
\end{equation}
with $z_{t,i}=\text{Tr}(\rho_t P_i)=\bra{\psi_t}P_i\ket{\psi_t}$, $x_{t,i}=\text{Tr}(O_t P_i)/2^n$. Then \cref{seq:drift-noise-vector}becomes (Einstein summation convention implied)
\begin{equation}
\begin{split}
f_l(\vect{z}_t,t) &= -\gamma(c_{ijm} c_{klm} - c_{ikm} c_{jlm}) x_{t,i}x_{t,j}z_{t,k}, \\
g_l(\vect{z}_t,t)  &= \sqrt{\gamma}((c_{ijl} + c_{jil})x_{t,i}z_{t,j} - 2 x_{t,i} z_{t,i} z_{t,l}),
\end{split}
\end{equation}
and $c_{ijk} = \mathrm{Tr}(P_i P_j P_k)/\mathrm{Tr}\,\id$ being the operator product expansion (OPE) coefficient.

Compared with the generic SDE in \cref{seq:general-sde-Mdim}, our weak measurement nonlinear SDE has $M=1$ dimensional noise $dw_t$ (which is a number) introduced by weak measurement in the $\boldsymbol{x}_t$ direction, and the matrix $\boldsymbol{G}(\boldsymbol{z}, t)$ reduces to the noise vector $\boldsymbol{g}(\boldsymbol{z}, t)$ here. Accordingly, we define the probability distribution of the state $\rho$ in terms of state vector $\boldsymbol{z}$ as $p(\boldsymbol{z},t)$.  

All the forward/backward SDE, ODE, PDE, score function and loss function specifically for our weak measurement model then follows from \cref{sec:general-diffusion} by setting the noise dimension $M=1$.


\subsection{The Unitary Reverse Procedure for Pure States}

\subsubsection{From learning unitary to learning score function}

For an ensemble of pure states, the forward diffusion process implemented by weak measurement must be purity-preserving. As is clear from \cref{eq-pure-sde}, a pure state remains pure along the quantum trajectory. Consequently, the time-evolution process is a unitary evolution.

For such pure state evolutions, we are therefore motivated to define a unitary score function being an infinitesimal time-$dt$ unitary evolution $V_\theta(\boldsymbol{z},t)$ (for some fixed small time step $dt$), such that it approximates the infinitesimal forward evolution
\begin{equation}
\rho_{t+dt}=\frac{K_{dt}(O_t, d o_t)\, \rho_t\, K_{dt}^\dagger(O_t, d o_t)}{\mathrm{Tr} \left[ K_{dt}(O_t, d o_t)\, \rho_t\, K_{dt}^\dagger(O_t, d o_t) \right]}\approx V_\theta(\boldsymbol{z}_t,t)\rho_{dt}V_\theta^\dagger(\boldsymbol{z}_t,t)=|\psi_{t+dt}\rangle \langle \psi_{}|\ ,\qquad \rho_{t}=|\psi_{t}\rangle \langle \psi_{t}|\ .
\end{equation}
Once this score function is learned, the time-reversed process can be generated by the inverse unitary transformations $V_\theta^\dagger(\boldsymbol{z}_t,t)$, which gives a purity-preserving reverse unitary evolution.

We are thus straightforwardly motivated to define a loss function
\begin{equation}\label{seq:loss-unitary}
\begin{split}
&\mathcal{L}_\theta = \frac{1}{2} \, \EE_{\boldsymbol{z}_t,\boldsymbol{z}_{t+\dd t}} 
\left\| 
2V_\theta(\boldsymbol{z}_{t+dt}, t) \, \rho_t \, V_\theta^\dagger(\boldsymbol{z}_{t+dt}, t) 
-\rho_t -\frac{K_{dt}(O_t, d o_t)\, \rho_t\, K_{dt}^\dagger(O_t, d o_t)}{\mathrm{Tr} \left[ K_{dt}(O_t, d o_t)\, \rho_t\, K_{dt}^\dagger(O_t, d o_t) \right]}+\frac{\gamma dt}{2}[O_{t}, [O_{t}, \rho_t]]
\right\|^2_F\\
&=\frac{1}{2} \, \EE_{\boldsymbol{z}_t,\boldsymbol{z}_{t+\dd t}} 
\left\| 
2V_\theta(\boldsymbol{z}_{t+dt}, t) \, \rho_t \, V_\theta^\dagger(\boldsymbol{z}_{t+dt}, t) 
-\rho_t-\rho_{t+dt}+\frac{\gamma dt}{2}[O_{t}, [O_{t}, \rho_t]]
\right\|^2_F\ , \\
&=\frac{1}{2} \, \EE_{\boldsymbol{z}_t,\boldsymbol{z}_{t+\dd t}} 
\left\| 
2V_\theta(\boldsymbol{z}_{t+dt}, t) \, \rho_t \, V_\theta^\dagger(\boldsymbol{z}_{t+dt}, t) 
-\rho_t-\rho_{t+dt}+\gamma dt\Big(\rho_t - O_{t}\rho_t O_{t}\Big)
\right\|^2_F\ , \\
\end{split}
\end{equation}
optimizing which gives the desired unitary score $V_\theta$. In the third line, we have used the fact that $O_{t}^2=\id$. 

We now show that this is equivalent to the denoising loss function in \cref{seq:loss-denoising-M-dim} usually defined in the diffusion problem. Using the fact that
\begin{equation}
\left\| \boldsymbol{z} \cdot \boldsymbol{P} \right\|^2 
= \sum_{ij} z_i z_j \, \mathrm{tr}(P_i P_j)
= 2^n \sum_i z_i^2 = 2^n \left\| \boldsymbol{z} \right\|^2\ ,
\end{equation}
the denoising loss function in \cref{seq:loss-denoising-M-dim} can be rewritten as
\begin{equation}\label{eq-loss-nonunitary}
\begin{split}
\mathcal{L}_\theta 
&= \frac{2^n}{2} \, \EE_{\boldsymbol{z}_t,\boldsymbol{z}_{t+\dd t}} \left\| 
\left[\boldsymbol{s}_\theta(\boldsymbol{z}_{t+dt} - \boldsymbol{f}(\boldsymbol{z}_t)dt,t) - \boldsymbol{f}(\boldsymbol{z}_t) \right] \cdot \frac{\boldsymbol{P}}{2^n} 
+ \left( \frac{\boldsymbol{z}_{t + dt} - \boldsymbol{z}_t}{dt} \right) \cdot \frac{\boldsymbol{P}}{2^n}
\right\|^2 \\
&= \frac{2^n}{2} \, \EE_{\boldsymbol{z}_t,\boldsymbol{z}_{t+\dd t}} \left\| 
\left[ \boldsymbol{s}_\theta(\boldsymbol{z}_{t+dt} - \boldsymbol{f}(\boldsymbol{z}_t)dt,t) - \boldsymbol{f}(\boldsymbol{z}_t) \right] \cdot \frac{\boldsymbol{P}}{2^n} 
+  \frac{\rho_{t + dt} - \rho_t}{dt} 
\right\|^2\ .
\end{split}
\end{equation}
Comparing with \cref{seq:loss-unitary}, we find if we make the following identification:
\begin{equation}
\left[\boldsymbol{s}_\theta(\boldsymbol{z}_{t+dt} - \boldsymbol{f}(\boldsymbol{z}_t)dt,t) - \boldsymbol{f}(\boldsymbol{z}_t)
\right] \cdot \frac{\boldsymbol{P}}{2^n}=\frac{2\rho_t-2V_\theta(\boldsymbol{z}_{t+dt}, t) \, \rho_t \, V_\theta^\dagger(\boldsymbol{z}_{t+dt}, t)}{dt}-\frac{\gamma}{2}[O_{t}, [O_{t}, \rho_t]]\ ,
\end{equation}
then the loss function $\mathcal{L}_\theta$ we defined in \cref{seq:loss-unitary} is up to a global factor equivalent to the denoising loss function in \cref{seq:loss-denoising-M-dim}! Moreover, using \cref{seq:drift-noise-vector}, we have $-\frac{\gamma}{2}[O_{t}, [O_{t}, \rho_t]]=\boldsymbol{f}(\boldsymbol{z}_t)\cdot \frac{\boldsymbol{P}}{2^n}$, and thus we can rewrite the above equation as
\begin{equation}
\left[\boldsymbol{s}_\theta(\boldsymbol{z}_{t+dt} - \boldsymbol{f}(\boldsymbol{z}_t)dt,t) - 2\boldsymbol{f}(\boldsymbol{z}_t)
\right] \cdot \frac{\boldsymbol{P}}{2^n}=\frac{2\rho_t-2V_\theta(\boldsymbol{z}_{t+dt}, t) \, \rho_t \, V_\theta^\dagger(\boldsymbol{z}_{t+dt}, t)}{dt}\ .
\end{equation}

Further, to the linear order of $dt$, the infinitesimal unitary $V_\theta$ can be parametrized by 
\begin{equation}
V_\theta = e^{i H_\theta dt} \approx I + i H_\theta dt\ ,\qquad H_\theta=\sum_j\eta_{\theta,j}P_j=\boldsymbol{\eta}_\theta \cdot \boldsymbol{P}\ ,
\end{equation}
where $H_\theta$ is a Hermitian generator. (Keeping only up to the first order in $dt$ of $V_\theta$ because our SDE is accurate in the first order in $dt$.) Thus, 
\begin{equation}
V_\theta \rho_t V_\theta^\dagger-\rho_t\approx i[H_\theta,\rho_t]dt\ , 
\end{equation}
and we can identify the following:
\begin{align}
\left[
\boldsymbol{s}_\theta(\boldsymbol{z}_{t+dt} - \boldsymbol{f}(\boldsymbol{z}_t)dt,t) - 2\boldsymbol{f}(\boldsymbol{z}_t)
\right] \cdot \frac{\boldsymbol{P}}{2^n}
&= -2i \left[ H_\theta, \rho_t \right] 
= -\frac{2i}{2^n} \left[ \boldsymbol{\eta}_\theta \cdot \boldsymbol{P}, \boldsymbol{z}_t \cdot \boldsymbol{P} \right] \nonumber \\
&= -\frac{2i}{2^n} \sum_{jk} \left[ P_j, P_k \right] \eta_{\theta,j} z_{t,k}\ .
\end{align}
We see that $V_\theta = e^{i H_\theta dt}$
encodes the information of the loss function $\frac{1}{2}\boldsymbol{s}_\theta - \boldsymbol{f}$. Accordingly, the reverse unitary evolution is equivalent to the reverse ODE at reverse time $\widetilde{t}=T-t$ given by the score function:
\begin{equation}
d\rho_{\widetilde{t}}=V_\theta^\dagger\rho_{\widetilde{t}}V_\theta-\rho_{\widetilde{t}} \approx -i[H_\theta,\rho_{\widetilde{t}}]d{\widetilde{t}}=\left(\frac{1}{2}\boldsymbol{s}_\theta - \boldsymbol{f}\right)\cdot \frac{\boldsymbol{P}}{2^n}d{\widetilde{t}}\ ,\qquad d\boldsymbol{z}_{\widetilde{t}}=\left(\frac{1}{2}\boldsymbol{s}_\theta - \boldsymbol{f}\right)d{\widetilde{t}}\ .
\end{equation}

\paragraph*{Equivalent score-operator form.} The loss form \cref{seq:loss-unitary} carries the Kraus drift as an explicit additive term inside the squared norm and parametrizes $V_\theta = e^{i H_\theta dt}$ as a strictly unitary object. An equivalent and more compact presentation absorbs both pieces into a single (in general non-unitary) score operator
\begin{equation}\label{seq:V-option-A}
\widetilde{V}_\theta(\boldsymbol{z}_{t+dt}, t) := \exp\!\Big(\big(-\tfrac{\gamma}{2}\,\delta O_t^2 + 2i\,H_\theta(\boldsymbol{z}_{t+dt}, t)\big)\,dt\Big),
\end{equation}
with $\delta O_t = O_t - \langle\psi_{t+dt}|O_t|\psi_{t+dt}\rangle$. The corresponding loss is the pure-state inner-product infidelity
\begin{equation}\label{seq:loss-A}
\mathcal{L}_\theta = \sum_t \Bigl(\,1 - \mathbb{E}_{\psi_t,\,\psi_{t+dt}}\big|\!\bra{\psi_t}\widetilde{V}_\theta(\boldsymbol{z}_{t+dt}, t)\ket{\psi_{t+dt}}\!\big|^2\Bigr),
\end{equation}
which up to a global multiplicative constant agrees with \cref{seq:loss-unitary} expanded to leading order in $dt$: the $-\tfrac{\gamma}{2}\delta O_t^2 dt$ exponent reproduces the $\frac{\gamma dt}{2}[O_t,[O_t,\rho_t]]$ Kraus correction in the squared-norm form, and the $2iH_\theta dt$ exponent encodes the same controller--score identification derived above, modulo the same factor of $2$ inherited from the score--flow relation $\boldsymbol{s} = 2\boldsymbol{f} - 2\boldsymbol{v}_\theta$. Optimizing either loss therefore yields the same effective backward unitary; we adopt \cref{seq:V-option-A} and \cref{seq:loss-A} for numerics because this compact score-operator form bundles the Kraus drift and the unitary rotation into a single object whose physical meaning---the averaged $dt$-step of the backward probability-flow update---is more transparent.

\subsubsection{Error bound for Unitary Learning}

The equivalence between learning the score function and learning a unitary generator allows us to bound the error of the generative model. Here we consider the error bound of the unitary reverse generative model. 

First, for a pair of pure states  $|\phi\rangle$ and $\psi\rangle$, we define their trace distance as
\begin{equation}
D(|\phi\rangle,\,|\psi\rangle)=\text{Tr}\sqrt{(\rho_\phi-\rho_\psi)^2}=\sqrt{2(1-|\langle\phi|\psi\rangle|^2)}\ ,\qquad \rho_\phi=|\phi\rangle\langle\phi|\ ,\quad \rho_\psi=|\psi\rangle\langle\psi|\ .
\end{equation}

Assume the learned probability distribution from the unitary reverse generative model is $p_\theta$, and the  true probability distribution of pure states $p_{\mathrm{true}}$. 
We then define the
\textbf{$p$-Wasserstein distance} with trace distance cost is defined by \cite{Zhang2024G2310.05866}
\begin{equation}\label{seq-Wp-distance2}
\begin{split}
&W_p(p_\theta, p_{\mathrm{true}})= \left( \inf_{\pi \in \Pi(p_\theta, p_{\mathrm{true}})} 
\int D(|\phi\rangle,\,|\psi\rangle)^p \;
\pi\big(|\phi\rangle,\, |\psi\rangle \big)d\phi d\psi \right)^{\frac{1}{p}} ,
\end{split}
\end{equation}
where $\Pi(p_\theta, p_{\mathrm{true}})$ runs over all  the \textbf{couplings} (joint distributions) $\pi$ 
on pairs of states $(|\phi\rangle,\,|\psi\rangle)$ whose marginals are $p_\theta$ and $p_{\mathrm{true}}$, namely satisfying $\int \pi\big(|\phi\rangle,\, |\psi\rangle \big) d\phi=p_{\mathrm{true}}(|\psi\rangle)$, and $\int \pi\big(|\phi\rangle,\, |\psi\rangle \big)d\psi=p_\theta(|\phi\rangle)$. 

We now show the following theorem:

\begin{theorem*}[Restatement of Theorem~1 of the main text] Distribution Convergence Theorem. Let $p_{\mathrm{true}}$ be the true probability distribution of initial pure states, and $p_\theta$ denote the learned probability distribution obtained after a reverse diffusion process of total evolution time $T$ generated by the learned control unitary $V_\theta$.  
Assume the learned control unitary satisfies  
\begin{equation}\label{seq-V-ep}
\EE_{\psi_t,\psi_{t+\dd t}}\sqrt{1-  
\left|\langle \psi_{t+dt}|V_\theta(\boldsymbol{z}_{t+dt}, t)|\psi_t\rangle \right|^2} \le \epsilon dt\ ,
\end{equation}
where $\boldsymbol{z}_t$ is the state vector of $|\psi_t \rangle$. Also, assume $V_\theta$ has a Lipschitz bound of constant $L_V dt$ such that for any pure state $|\alpha\rangle$: 
\begin{equation}
D(V_\theta^\dagger(\boldsymbol{z},t)|\alpha\rangle,V_\theta^\dagger(\boldsymbol{z}',t)|\alpha\rangle)\le L_Vdt D(|\psi\rangle,|\psi^{\prime}\rangle)\ ,
\end{equation}
where $\boldsymbol{z}$ and $\boldsymbol{z}'$ are the state vectors of $|\psi\rangle$ and $|\psi^{\prime}\rangle$, respectively. Then by choosing a stopping time $T = \frac{C}{L_v}\ln(1/\epsilon)$ with $0 < C < 1$, the $1$-Wasserstein distance between the learned ensemble $p_{\theta}$ and the true ensemble $p_{\mathrm{true}}$ is bounded by
\begin{equation}\label{eq-W-bound}
W_{1}(p_{0}, p_{\mathrm{true}}) 
\leq \frac{\sqrt{2}}{L_{v}}  \epsilon^{1-C} 
+ \epsilon^{-C} \, E_{\mathrm{diff}}(T)
\end{equation}
with $E_{\mathrm{diff}}(T)$ decaying exponentially as $n e^{-\frac{4\gamma}{3n}T}$ with respect to the diffusion time $T$. If we choose a sufficiently large measurement strength $\gamma$ such that $\frac{4\gamma}{3n L_{v}} \ge 1$, the leading-order behavior of the bound for small $\epsilon$ is dominated by the first term in \cref{eq-W-bound}, and $
W_{1}(p_{0}, p_{\mathrm{true}}) \;\to\; 0$ as $\epsilon \to 0.$
\end{theorem*}

\begin{proof}We first note that once the learned control unitary $V_\theta$ is determined, the reversed state $|\phi_t\rangle$ at any time $0\le t\le T$ generated by $V_\theta$ backward from a state $|\phi_T\rangle$ at time $T$ is deterministic:
\begin{equation}
|\phi_t\rangle=\mathcal{T}^\dagger\prod_{t'=t}^{T-dt} V_\theta^\dagger (\boldsymbol{z}_{t'+dt}^\phi,t')|\phi_T\rangle=V^\dagger(\phi_T,t,T)|\phi_T\rangle\ ,
\end{equation}
where $\mathcal{T}^\dagger$ stands for the anti-time ordering, and $\boldsymbol{z}_{t}^\phi$ is the state vector of state $|\phi_{t}\rangle$.

We denote the true probability distribution at time $t$ as $p_{\mathrm{true},t}$, which is evolved by the forward weak measurement diffusion from the initial true probability distribution $p_{\mathrm{true}}$. Similarly, we denote the learned probability distribution at time $t$ as $p_{\theta,t}$, which is generated by reverse diffusion $V_\theta$ from the true probability distribution $p_{\mathrm{true},T}$ at time $T$.

Using the fact that $|\phi_t\rangle=V_\theta^\dagger(\boldsymbol{z}_{t+dt}^\phi,t)|\phi_{t+dt}\rangle$, 
%
\begin{equation}\label{seq-D-ineq}
\begin{split}
&D(|\phi_t\rangle,\,|\psi_t\rangle)=D(V_\theta^\dagger(\boldsymbol{z}_{t+dt}^\phi,t)|\phi_{t+dt}\rangle,\,|\psi_t\rangle)=D(|\phi_{t+dt}\rangle,\,V_\theta(\boldsymbol{z}_{t+dt}^\phi,t)|\psi_t\rangle)\\
&\le D(|\phi_{t+dt}\rangle,|\psi_{t+dt}\rangle)+D(|\psi_{t+dt}\rangle, V_\theta(\boldsymbol{z}_{t+dt}^\phi,t)|\psi_t\rangle) \\
&=D(|\phi_{t+dt}\rangle,\,|\psi_{t+dt}\rangle) +\sqrt{2\left(1- \left|\langle \psi_{t+dt}| V_\theta(\boldsymbol{z}_{t+dt}^\phi,t)|\psi_t\rangle\right|^2\right)}\ .
\end{split}
\end{equation}
which holds for any state $|\psi_{t+dt}\rangle$. 

Assume the Wasserstein distance $W_1(p_{\theta,t+dt}, p_{\mathrm{true},t+dt})$ for time $t+dt$ is achieved by the coupling $\pi_{t+dt}^{\text{inf}} \in \Pi(p_{\theta,t+dt}, p_{\mathrm{true},t+dt})$, namely,
\begin{equation}
W_1(p_{\theta,t+dt}, p_{\mathrm{true},t+dt})=\int D(|\phi_{t+dt}\rangle,\,|\psi_{t+dt}\rangle) \;
\pi_{t+dt}^{\text{inf}}\big(|\phi_{t+dt}\rangle,\, |\psi_{t+dt}\rangle \big)d\phi_{t+dt} d\psi_{t+dt}\ .
\end{equation}
We define an induced coupling at time $t$ as $\pi^d_{t}$:
\begin{equation}\label{seq-pi-ind}
\pi^d_t\big(|\phi_{t}\rangle,\, |\psi_{t}\rangle \big)d\phi_t=d\phi_{t+dt}\int_{\psi_{t+dt}} \pi_{t+dt}^{\text{inf}}\big(\phi_{t+dt}\rangle,\, |\psi_{t+dt}\rangle \big) 
\frac{\Gamma_{\mathrm{true},t}(|\psi_{t+dt}\rangle,|\psi_t\rangle)}{p_{\mathrm{true},t+dt}(|\psi_{t+dt}\rangle)}
d\psi_{t+dt} \ ,
\end{equation}
where $|\phi_{t+dt}\rangle =V_\theta(\boldsymbol{z}_{t+dt}^\phi,t)|\phi_{t}\rangle$. The function $\Gamma_{\mathrm{true},t}(|\psi_{t+dt}\rangle,|\psi_t\rangle)$ is the true joint distribution between states $|\psi_{t+dt}\rangle,|\psi_t\rangle$, satisfying $$\int \Gamma_{\mathrm{true},t}(|\psi_{t+dt}\rangle,|\psi_t\rangle) d\psi_{t+dt} = p_{\mathrm{true},t}(|\psi_t\rangle)\ ,\quad \int \Gamma_{\mathrm{true},t}(|\psi_{t+dt}\rangle,|\psi_t\rangle) d\psi_{t} = p_{\mathrm{true},t+dt}(|\psi_{t+dt}\rangle)\ .$$
Note that $\frac{\Gamma_{\mathrm{true},t}(|\psi_{t+dt}\rangle,|\psi_t\rangle)}{p_{\mathrm{true},t+dt}(|\psi_{t+dt}\rangle)}=p_{\mathrm{true}}(|\psi_t\rangle \big| |\psi_{t+dt}\rangle)$ is the conditional probability of $|\psi_t\rangle$ with respect to $|\psi_{t+dt}\rangle$ of the true distribution, which satisfies $\int p_{\mathrm{true}}(|\psi_t\rangle \big| |\psi_{t+dt}\rangle) d\psi_t=1$.

By \cref{seq-D-ineq,seq-pi-ind}, we arrive at the following inequality:
\begin{equation}\label{seq-pi-d-ineq}
\begin{split}
&\int D(|\phi_t\rangle,\,|\psi_t\rangle) \;
\pi_t^d\big(|\phi_t\rangle,\, |\psi_t\rangle \big)d\phi_t d\psi_{t}\\
\le &\int D(|\phi_{t+dt}\rangle,\,|\psi_{t+dt}\rangle) \;
\pi_{t+dt}^{\text{inf}}\big(|\phi_{t+dt}\rangle,\, |\psi_{t+dt}\rangle \big) d\phi_{t+dt} d\psi_{t+dt}\\
&\ + \int \sqrt{2\left(1- \left|\langle \psi_{t+dt}| V_\theta(\boldsymbol{z}_{t+dt}^\phi,t)|\psi_t\rangle\right|^2\right)} \;
\pi_{t+dt}^{\text{inf}}\big(|\phi_{t+dt}\rangle,\, |\psi_{t+dt}\rangle \big) 
\frac{\Gamma_{\mathrm{true},t}(|\psi_{t+dt}\rangle,|\psi_t\rangle)}{p_{\mathrm{true},t+dt}(|\psi_{t+dt}\rangle)}
d\phi_{t+dt}d\psi_{t+dt}d\psi_t \\
=&W_1(p_{\theta,t+dt}, p_{\mathrm{true},t+dt})+\Delta(\pi_{t+dt}^{\text{inf}})\ ,
\end{split}
\end{equation}
where
\begin{equation}
\Delta(\pi_{t+dt}^{\text{inf}})=\int \sqrt{2\left(1- \left|\langle \psi_{t+dt}| V_\theta(\boldsymbol{z}_{t+dt}^\phi,t)|\psi_t\rangle\right|^2\right)} \;
\frac{\pi_{t+dt}^{\text{inf}}\big(|\phi_{t+dt}\rangle,\, |\psi_{t+dt}\rangle \big) 
}{p_{\mathrm{true},t+dt}(|\psi_{t+dt}\rangle)} \Gamma_{\mathrm{true},t}(|\psi_{t+dt}\rangle,|\psi_t\rangle)
d\phi_{t+dt}d\psi_{t+dt}d\psi_t\ .
\end{equation}

We assume $V_\theta$ is Lipschitz with coefficient $L_V dt$, that for any state $|\alpha\rangle$,
\begin{equation}
D(V_\theta^\dagger(\boldsymbol{z},t)|\alpha\rangle,V_\theta^\dagger(\boldsymbol{z}',t)|\alpha\rangle)\le L_Vdt D(|\psi\rangle,|\psi^{\prime}\rangle)\  .
\end{equation}
Then, using \cref{seq-V-ep}, assuming $\boldsymbol{z}_{t+dt}$ denotes the state vector of state $|\psi_{t+dt}\rangle$, we find
\begin{equation}
\begin{split}
&\Delta(\pi_{t+dt}^{\text{inf}})\le \sqrt{2}\epsilon dt + \int \left(D(| \psi_{t}\rangle, V_\theta^\dagger(\boldsymbol{z}_{t+dt},t)|\psi_{t+dt}\rangle)- D(| \psi_{t}\rangle , V_\theta^\dagger(\boldsymbol{z}_{t+dt}^\phi,t)|\psi_{t+dt}\rangle)\right) \;
\\
&\qquad\qquad \times \frac{\pi_{t+dt}^{\text{inf}}\big(|\phi_{t+dt}\rangle,\, |\psi_{t+dt}\rangle \big) 
}{p_{\mathrm{true},t+dt}(|\psi_{t+dt}\rangle)} \Gamma_{\mathrm{true},t}(|\psi_{t+dt}\rangle,|\psi_t\rangle)
d\phi_{t+dt}d\psi_{t+dt}d\psi_t\\
&\le \sqrt{2}\epsilon dt + \int D( V_\theta^\dagger(\boldsymbol{z}_{t+dt},t)|\psi_{t+dt}\rangle, V_\theta^\dagger(\boldsymbol{z}_{t+dt}^\phi,t)|\psi_{t+dt}\rangle)
\frac{\pi_{t+dt}^{\text{inf}}\big(|\phi_{t+dt}\rangle,\, |\psi_{t+dt}\rangle \big) 
}{p_{\mathrm{true},t+dt}(|\psi_{t+dt}\rangle)} \Gamma_{\mathrm{true},t}(|\psi_{t+dt}\rangle,|\psi_t\rangle)
d\phi_{t+dt}d\psi_{t+dt}d\psi_t\\
&\le \sqrt{2}\epsilon dt +
L_v dt\int D(|\phi_{t+dt}\rangle,\,|\psi_{t+dt}\rangle) \;
\pi_{t+dt}^{\text{inf}}\big(|\phi_{t+dt}\rangle,\, |\psi_{t+dt}\rangle \big)d\phi_{t+dt} d\psi_{t+dt} \\
&=\sqrt{2}\epsilon dt +L_vdt W_1(p_{\theta,t+dt}, p_{\mathrm{true},t+dt})  \ .
\end{split}
\end{equation}
Thus, by \cref{seq-pi-d-ineq} and the definition of the $1$-Wasserstein distance, we find
\begin{equation}\label{seq-W2-bound-diff}
W_1(p_{\theta,t}, p_{\mathrm{true},t}) \le  \big(1+L_V dt\big) W_1(p_{\theta,t+dt}, p_{\mathrm{true},t+dt}) +\sqrt{2}\epsilon dt\ .
\end{equation}
This implies
\begin{equation}\label{seq-W2-bound-int}
e^{-L_V (T-t)}W_1(p_{\theta,t}, p_{\mathrm{true},t}) \le  e^{-L_V (T-t-dt)} W_1(p_{\theta,t+dt}, p_{\mathrm{true},t+dt}) +\sqrt{2} e^{-L_V (T-t)}\epsilon dt\ .
\end{equation}

At $t=T$, we start to perform the unitary reverse from a uniform random product state distribution $p_{\theta,T}$. Meanwhile, $E_{\mathrm{diff}}(T) \sim n \exp \left(-\frac{4\gamma}{3n}T\right)$, because for non-identity Pauli strings, each Pauli weight decays exponentially to zero with increasing diffusion time $T$, as shown in \cref{seq-measurement-pauli-weight2}. So
\begin{equation}
W_1(p_{\theta,T}, p_{\mathrm{true},T}) = E_{\mathrm{diff}}(T) \sim n e^{-\frac{4\gamma}{3n}T}\ ,
\end{equation}
as is evident from the Pauli weight of the weak measurement channel in \cref{seq-measurement-pauli-weight2}. This together with integrating \cref{seq-W2-bound-int} from $T$ to $t$ yields an accumulated bound:
\begin{equation}
W_1(p_{\theta,t}, p_{\mathrm{true},t}) \le e^{L_V (T-t)}\Big[\frac{\sqrt{2}\epsilon}{L_V}\left(1-e^{-L_V (T-t)}\right) +E_{\mathrm{diff}}(T) \Big] \ .
\end{equation}
Thus, particularly, for $t=0$, the $1$-Wasserstein distance between $p_\theta=p_{\theta,0}$ and $p_{\mathrm{true}}=p_{\mathrm{true},0}$ is bounded by
\begin{equation}\label{bound-no-T}
W_1(p_{\theta}, p_{\mathrm{true}}) \le e^{L_VT}\left[\frac{\sqrt{2}\epsilon}{L_V} \left(1-e^{-L_V T}\right) +E_{\mathrm{diff}}(T) \right]\ .
\end{equation}

Now if we choose a stopping time $T = \frac{C}{L_v}\ln(1/\epsilon)$, then \cref{bound-no-T} becomes
\begin{equation}\label{bound-T}
W_{1}(p_{0}, p_{\mathrm{true}}) 
\leq \frac{\sqrt{2}}{L_{v}}  \epsilon^{1-C} 
+ \epsilon^{-C} \, E_{\mathrm{diff}}(T)
\end{equation}

In the leading-order behavior of this expression for small $\epsilon,$ the second term in \cref{bound-T} becomes
\begin{equation}
E_{\mathrm{diff}}(T) 
\sim n \exp\left( - \frac{4\gamma}{3n L_{v}} \ln \frac{1}{\epsilon} \right) 
= n \epsilon^{\frac{4\gamma}{3n L_{v}}}
\end{equation}

In choosing a sufficiently large measurement strength $\gamma$, we expect $E_{\text{diff}}(T)$ to decay sufficiently fast that $\frac{4\gamma}{3n L_{v}} \ge C,$ and the leading-order behavior of the bound for small $\epsilon$ is then dominated by the first term in \cref{bound-T}, which is a constant. As $\epsilon \to 0,$
\begin{equation}
W_{1}(p_{0}, p_{\mathrm{true}}) \;\to\; \frac{\sqrt{2}}{L_{v}} \epsilon^{1-C} 
\end{equation}

This means that the distribution distance $W_1 \to 0$ when the training error $\epsilon$ is made arbitrarily tiny.

\end{proof}

\section{Pauli Weights Evolution under Measurement Channel and Measurement-and-prepare channel}

This appendix derives the analytical expression of the measurement channel $\mathcal{F}_t$ and the measurement-and-prepare channel $\mathcal{M}_t$ we defined in the main text, by solving the differential equation of their Pauli weights.

\subsection{The Measurement channel and the linear SDE}

\subsubsection{The measurement channel}

We next examine the generic weak measurement protocol of \cref{sec-weak-measurement-protocol} in the context of qubit systems. We consider a system of $n$ qubits, initialized in a state $\rho_0$. The forward diffusion is implemented by applying a sequence of ancilla-assisted weak measurements. At each infinitesimal time step $\delta t$, a single-qubit observable $O_{ t}$ is measured once on a qubit $j_t$ (either randomly chosen or consecutively). The observable is selected from the set of single-qubit Pauli operators $\{\,\sigma_{x,j},\sigma_{y,j}, \sigma_{z,j} \mid j=1,\dots,n\,\}$ or their linear superpositions, ensuring $O_{ t}^2 = I$. The corresponding Kraus operator is given by:
\begin{equation}
\label{eq:K_delt2}
K_{\delta t}(O_t,\delta o_t)
\;=\;
\frac{1}{\sqrt{2}}\,
\exp\!\Big(-\gamma\,\delta t +2\gamma\,O_t\,\delta o_t\Big)\ ,\qquad \delta o_t= \frac{1}{2}\sqrt{\frac{\delta t}{\gamma}} o_t\ ,\qquad o_t=\pm 1\ .
\end{equation}

In the following derivations of this appendix, for convenience, we do not distinguish between the discrete time $\delta t$ and $dt$, and we will denote \cref{eq:K_delt} as
\begin{equation}\label{seq-Kdt-tilde}
K_{dt}(O_t, do_t)=\frac{1}{\sqrt{2}}\exp\left(-\gamma dt+ o_t \sqrt{\gamma dt} O_{j_{t, t}}\right) \approx \frac{1}{\sqrt{2}} \left[ \left(1 - \frac{\gamma dt}{2}\right) I + o_t \sqrt{\gamma dt} O_{t} \right]\ .
\end{equation}
Note that $K_{dt}(O_t, do_t)$ differs from a continuous-time Kraus increment in normalization: a continuous-time formulation accumulates many weak-measurement substeps within each $dt$ (so the accumulated readout $do_t$ can be treated as continuous), whereas here in \cref{seq-Kdt-tilde}, for simplicity, we assume a single weak-measurement substep per $dt$ with discrete outcomes $do_t=\frac{1}{2}\sqrt{\frac{d t}{\gamma}}$.

For a sequence of consecutive weak measurements on the $n$-qubit system from time $t=0$ to $t=T$ ($t=kdt$, $1\le k\le N_T$, $N_T=T/dt$), we define a trajectory $\mathcal{O}$ as the sequence of measured observables $O_{ t}$ and their corresponding outcomes $o_t\in\{\pm 1\}$ at each time step:
\begin{equation}
\mathcal{O}=\left\{(O_{ t}, o_t) \mid 0 \leqslant t \leqslant T\right\}
\end{equation}

The associated Kraus operator up to time $t$
along the trajectory $\mathcal{O}$ is
\begin{equation}
\begin{split}
&K_t(\mathcal{O})=\mathcal{T}\prod_{t^{\prime}=0}^t K_{d t}\left(O_t^{\prime}, do_t^{\prime}\right)\ ,
\end{split}
\end{equation}
where $\mathcal{T}$ denotes the time-ordering operator, and $K_{dt}$ is defined in \cref{seq-Kdt-tilde}. At time $t$, the system state $\rho_t$ conditional on trajectory $\mathcal{O}$ is given by
\begin{equation}\label{eq-rho_t}
\rho_t(\mathcal{O})=\frac{\widetilde{\rho}_t(\mathcal{O})}{\operatorname{Tr}\left[\widetilde{\rho}_t(\mathcal{O})\right]} \ ,  \quad \widetilde{\rho}_t(\mathcal{O})=2^{N_t}K_t(\mathcal{O}) \rho_0 K_t^\dagger\left(\mathcal{O}\right)\ ,\qquad N_t=\frac{t}{dt}\ .
\end{equation}
where $\rho_0$ is the initial state at $t=0$, and we have defined a tilded state $\widetilde{\rho}_t(\mathcal{O})$ with unnormalized trace. The trace of $\widetilde{\rho}_t(\mathcal{O})$ has the physical meaning of being the probability for trajectory $\mathcal{O}$ to occur:
\begin{equation}\label{eq-pt}
p_t(\mathcal{O})=\frac{1}{2^{N_t}}\operatorname{Tr}\left[\widetilde{\rho}_t(\mathcal{O})\right]\ .
\end{equation}
As defined in the main text, the measurement channel $\mathcal{F}_t$ is the average of the final states of all trajectories respecting their probabilities:
\begin{equation}\label{eq-m-channel}
\begin{split}
\mathcal{F}_t(\rho_0)& = \sum_{\mathcal{O}} p_t(\mathcal{O}) \, \rho_t(\mathcal{O}) =\frac{1}{2^{N_t}}\sum_{\mathcal{O}} \widetilde{\rho}_t(\mathcal{O})=  \frac{1}{2^{N_t}}\sum_{\mathcal{O}} K_t(\mathcal{O}) \, \rho_0 \, K_t^\dagger(\mathcal{O})\ .
\end{split}
\end{equation}
Particularly, in the final expression, all the $2^{N_t}$ trajectories $\mathcal{O}$ is summed at equal weight, as their probability $p_t(\mathcal{O})$ is readily included in the trace of $\widetilde{\rho}_t(\mathcal{O})$. Therefore, we say that the unnormalized states $\widetilde{\rho}_t(\mathcal{O})$ of different trajectories $\mathcal{O}$ effectively occur at equal probability.

\subsubsection{The linear SDE}

The unnormalized state $\widetilde{\rho}_t(\mathcal{O})$ defined in \cref{eq-rho_t} plays an important role in the above quantum channels, so it is useful to take a closer look at it. Its evolution under the sequence of weak measurements is governed by a linear SDE \cite{Jacobs_2006}:
\begin{equation}\label{eq-lSDE}
d \widetilde{\rho}_t=-\frac{\gamma}{2}[O_{t},[O_{t}, \widetilde{\rho}_t]] d t+\sqrt{\gamma}\{O_{t}, \widetilde{\rho}_t\} d\widetilde{w}_t\ ,\qquad d\widetilde{w}_t=o_t\sqrt{dt}=2\sqrt{\gamma}\dd o_t\ ,
\end{equation}
which can be derived from the expression of Kraus operators $K_{dt}$, and $d\widetilde{w}_t$ is the stocastic noise variable. Note that for later convenience, we choose to rescale the variable $\dd o_t$ into $d\widetilde{w}_t$ according to the above here. An important difference compared to the normalized state $\rho_t$, since each unnormalized state $\widetilde{\rho}_t$ occurs at equal probability (\cref{eq-m-channel}), $d\widetilde{w}_t$ in the linear SDE \cref{eq-lSDE} for $\widetilde{\rho}_t$ has a probability distribution with first and second moment
\begin{equation}\label{seq-dwt-moments}
\mathbb{E}[d\widetilde{w}_t]=0\ ,\qquad \mathbb{E}[d\widetilde{w}_t^2]=dt\ .
\end{equation}
In the limit $dt\rightarrow 0$, it is sufficient to keep all the calculations up to order $dt\sim d\widetilde{w}_t^2$, so all probability distributions of $d\widetilde{w}_t$ with the same first and second moments are equivalent for the SDE. Thus, we can simply set $d\widetilde{w}_t$ as the standard random Wiener noise obeying a Gaussian distribution satisfying \cref{seq-dwt-moments}.

in terms of $d\widetilde{w}_t$, the infinitesimal Kraus operator in \cref{seq-Kdt-tilde} can be rewritten as
\begin{equation}\label{seq-Kdt-dw}
K_{d t}\left(t, d\widetilde{w}_t\right)=\frac{1}{\sqrt{2}}\exp\left(-\gamma dt+ \sqrt{\gamma } O_{t}d\widetilde{w}_t\right)\ ,
\end{equation}
which will be used hereafter in quantum channel calculations. 

\subsubsection{Solving the measurement channel}

As the equal probability average over all $\widetilde{\rho}_t(\mathcal{O})$, the measurement channel satisfies a differential equation being the average of the linear SDE in \cref{eq-lSDE}, namely, for infinitesimal evolution:
\begin{equation}
d\rho_t=\mathcal{F}_{dt}(\rho_t)-\rho_t=-\frac{\gamma}{2}[O_{t},[O_{t}, \rho_t]] d t\ .
\end{equation}
Assume $O_t$ acts on qubit $j_t$. By averaging $O_{t}$ over all possible single-qubit Pauli operators ($\sigma_{j_t,x}=X_j,\sigma_{j_t,y}=Y_j,\sigma_{j_t,z}=Z_j$ each with $1/3$ probability) on qubit $j_t$, we find
\begin{equation}
d\rho_t=\mathcal{F}_{dt}(\rho_t)-\rho_t=-\frac{\gamma}{6}\sum_{\mu=x,y,z}[\sigma_{j_t,\mu},[\sigma_{j_t,\mu}, \rho_t]] d t\ .
\end{equation}
Particularly, this shows that $\mathcal{F}_{dt}$ maps a Pauli string operator $P$ to itself, and thus the measurement channel is diagonal in the Pauli basis. More explicitly, we can expand the density matrix in the Pauli basis as
\begin{equation}
\rho_t = \frac{1}{2^n} \sum_i z_{t,i} P_i=\frac{1}{2^n}\boldsymbol{z}_t\cdot\boldsymbol{P}\ ,\qquad \rho_{t+dt}=\mathcal{F}_{dt}(\rho_t)\ ,
\end{equation}
and the above differential equation implies
\begin{equation}
dz_{t,i}=-\frac{\gamma dt}{6}z_{t,i}\sum_{\mu=x,y,z}[\sigma_{j_t,\mu},[\sigma_{j_t,\mu}, P_i]]=-\frac{4\gamma dt}{3}z_{t,i}\delta_{|P_i|,j_t}\ ,
\end{equation}
where $|P|$ denotes the support of Pauli operator $P$, while we define $\delta_{|P|,j_t}=1$ if qubit $j_t\in|P|$, and $\delta_{|P|,j_t}=0$ otherwise. Solving the above equation gives the following solution:
\begin{equation}\label{seq-measurement-pauli-weight1}
z_{t,i}=w_{\mathcal{F}_t}(P_i)z_{0,i}\ ,\qquad w_{\mathcal{F}_t}(P)=\exp\left(-\frac{4\gamma}{3}\sum_{j\in|P|}t^{(j)}\right)\ ,
\end{equation}
where we have defined $t^{(j)}$ as the total length of time that the $j$-th qubit is measured (in other words, the $j$-th qubit is measured for $t^{(j)}/dt$ times), which satisfy $\sum_{j=1}^n t^{(j)}=t$. The coefficient $w_{\mathcal{F}_t}(P)$ is called the \emph{Pauli weight} of the measurement channel, which is defined as
\begin{equation}
w_{\mathcal{F}_t}(P)=\frac{\text{Tr}[P\mathcal{F}_t(P)]}{\text{Tr}{\id}}\ ,\qquad \rightarrow\qquad \mathcal{F}_t(P)=w_{\mathcal{F}_t}(P)P\ .
\end{equation}
Thus, the action of the measurement channel on a density matrix $\rho_0 = \frac{1}{2^n} \sum_i z_{0,i} P_i$ is
\begin{equation}
\mathcal{F}_t\left(\frac{1}{2^n} \sum_i z_{0,i} P_i\right)=\frac{1}{2^n} \sum_i w_{\mathcal{F}_t}(P)z_{0,i} P_i\ .
\end{equation}
For weak measurements either uniformly randomly on a qubit, or sequentially on each of the $n$ qubits, in the long time limit $t/dt\gg n$, each qubit is on average measured for a total length of time $t^{(j)}=t/n$. Thus, the Pauli weight in \cref{seq-measurement-pauli-weight1} becomes
\begin{equation}\label{seq-measurement-pauli-weight2}
w_{\mathcal{F}_t}(P)=w_{\mathcal{F},m}(t)=\exp\left(-\frac{4\gamma}{3n}mt\right)\ ,\qquad m=|P|\ ,
\end{equation}
where $m=|P|$ denotes the support size of the Pauli string operator $P$.


\subsection{The Measurement-and-prepare channel}

\subsubsection{Definition of the channel}

We now turn to the measurement-and-prepare channel, which is essential for our method to extract weak measurement shadow tomography from the outcomes gathered through weak measurements. 

We define the following unnormalized classical snapshot state ($\text{Tr}[\sigma_t(\mathcal{O})]\neq1$):
\begin{equation}
\sigma_t(\mathcal{O})=K_t^{\dag}(\mathcal{O}) K_t(\mathcal{O}) \ ,
\end{equation}
which tends to unnormalized random pure products state when $t\rightarrow \infty$. The random single qubit measurements then form a measurement-and-prepare quantum channel $\mathcal{M}_t$ that maps the initial state $\rho_0$ into
\begin{equation}\label{eq-mp-channel}
\begin{split}
    &\mathcal{M}_t(\rho_0)=\frac{1}{2^{N_t}}\sum_{\mathcal{O}} \sigma_t(\mathcal{O}) \operatorname{Tr}\left[\sigma_t(\mathcal{O}) \rho_0\right]=\frac{1}{2^{N_t}}\sum_{\mathcal{O}} \sigma_t(\mathcal{O})p_t(\mathcal{O})=\underset{\mathcal{O}}{\mathbb{E}}\  \sigma_t(\mathcal{O})p_t(\mathcal{O})=\EE_{\mathcal{O} \sim p(\mathcal{O} \mid \bar{\rho}_0)} \sigma_t(\mathcal{O})\ ,
    \end{split}
\end{equation}
where $p_t(\mathcal{O})$ is the probability of trajectory $\mathcal{O}$ as defined in \cref{eq-pt}, $N_t=t/dt$ is the number of steps. 
The channel in \cref{eq-mp-channel} is defined not preserving the trace of $\rho_0$. The randomized measurement channel maps the underlying state $\rho$ to an ensemble of classical snapshot states $\sigma \sim p(\sigma|\rho)\sim\text{Tr}(\sigma\rho)$.

In practice, $\mathcal{M}_t(\rho_0)$ can be measured by summing over the output snapshot states $\sigma_t(\mathcal{O})$ of all the trajectories measured. By definition, each trajectory $\mathcal{O}$ occurs with probability $p_t(\mathcal{O})$, thus the resulting sum approaches $\mathcal{M}_t(\rho_0)$. Note that the output $\sigma_t(\mathcal{O})$ of the measurement-and-prepare channel can be calculated solely from the measurement outcomes which determines $K_t(\mathcal{O})$, and no information about $\rho_0$ is needed. By reversing the measurement-and-prepare channel, one can deduce $\rho_0$, or the shadow of it. Thus, it is important to derive the channel $\mathcal{M}_t$ analytically, which we shall do below.

As local scrambled channels (because of randomized single-qubit measurements invariant under local SU(2) rotations), both channels are diagonal in the Pauli basis, and we can define their Pauli weight as:
\begin{equation}
w_{\mathcal{M}_t}(P)=\frac{\text{Tr}[P\mathcal{M}_t(P)]}{\text{Tr}{\id}}\ .
\end{equation}
Accordingly, their action on Pauli operator $P$ is given by
\begin{equation}
\mathcal{M}_t(P)=w_{\mathcal{M}_t}(P) P\ .
\end{equation}
Particularly, $\text{Tr}\left[\mathcal{M}_t(\rho_0)\right]=w_{\mathcal{M}_t}(\id)$ is unnormalized. In the below, we derive the Pauli weight $w_{\mathcal{M}_t}(P)$ of the unnormalized channel.

\subsubsection{Pauli Transfer Matrix}

For simplicity, we ignore the trajectory label $\mathcal{O}$ in this subsection. We want to find a differential equation satisfied by the unnormalized Pauli weight $w_{\mathcal{M}_t}(P)$. To do this, we consider the time evolution of $n$ steps in time $ndt$, in which we assume each of the $n$ qubits is measured once. We note that the unnormalized channel in \cref{eq-mp-channel} at time $t$ and $t+ndt$ can be rewritten as
\begin{equation}
\mathcal{M}_t(\rho_0)=\underset{\sigma_t}{\mathbb{E}}\, \sigma_t \operatorname{Tr}\left[\sigma_t \rho_0\right]\ ,\qquad \mathcal{M}_{t+ndt}(\rho_0)=\underset{\mathcal{N}}{\mathbb{E}}\,\underset{\sigma_t}{\mathbb{E}}\, \mathcal{N}(\sigma_t) \operatorname{Tr}\left[\mathcal{N}(\sigma_t) \rho_0\right]\ ,
\end{equation}
where we have defined a channel $\mathcal{N}$ as
\begin{equation}
\sigma_{t+ndt}=\mathcal{N}(\sigma_t)= 2^{n}K_{ndt}^\dagger  \sigma_t K_{ndt}\ .
\end{equation}
Here $K_{ndt}$ is a random Kraus operator corresponding to $n$ steps of weak measurement, which has the form
\begin{equation}
K_{n \mathrm{d}t} = \exp \left( -n\gamma\,\mathrm{d}t + \sum_{j=1}^n \sqrt{\gamma}\, O_j\, d\widetilde{w}_j \right)
\simeq 1 + \sum_{j=1}^n \sqrt{\gamma}\, O_j\, d\widetilde{w}_j - \frac{n\gamma}{2}\,\mathrm{d}t\ ,
\end{equation}
with $O_j$ being the random Pauli operator measured on qubit $j$ ($O_j^2=1$), and $d\widetilde{w}_j=\pm \sqrt{dt}$ are random measurement outcomes.

It is clear that the channel $\mathcal{N}$ above is also locally scrambled, namely, its distribution is invariant under local SU(2) rotations of each qubit. Therefore, by \cite{akhtar2023measurementinducedcriticalitytomographicallyoptimal}, the evolution of Pauli weights from $\mathcal{M}_t$ to $\mathcal{M}_{t+ndt}$ defined above is given by a linear transformation govered by the \emph{Pauli transfer matrix} (PTM):
\begin{equation}\label{ptm}
w_{\mathcal{M}_{t+ndt}}(P) = \sum_{P'} W(P, P')\, w_{\mathcal{M}_t}(P',t)\ ,
\end{equation}
where $P$ runs over all the Pauli string operators, and the PTM is defined by
\begin{equation}\label{seq-def-PTM}
W(P, P')=\underset{\mathcal{N}}{\mathbb{E}} \left( \frac{\operatorname{Tr} \left[ P\, \mathcal{N}(P') \right]}{\operatorname{Tr}{\id}} \right)^2\ .
\end{equation}
Calculating the PTM therefore gives the differential equations satisfied by Pauli weights $w_{\mathcal{M}_t}(P)$.

\subsubsection{Deriving the Pauli weight differential equation}

For a Pauli string operator $P'$ with support on $m'=|P'|$ qubits, the action of channel $\mathcal{N}$ to linear order of $\mathrm{d}t$ reads:
\begin{equation}
\begin{aligned}
& \mathcal{N}\left(P^{\prime}\right) \simeq\left(1+\sum_{j=1}^n \sqrt{\gamma} O_j \mathrm{d}\widetilde{w}_j-\frac{\gamma}{2} n \mathrm{d}t\right) P^{\prime}\left(1+\sum_{j=1}^n \sqrt{\gamma} O_j \mathrm{d}\widetilde{w}_j-\frac{\gamma}{2} n \mathrm{d}t\right) \\
& \simeq(1-\gamma q \mathrm{d}t) P^{\prime}+\sum_{j=1}^n \sqrt{\gamma} \mathrm{d}\widetilde{w}_j \left\{O_j, P^{\prime}\right\}+\sum_{j, j^{\prime}} \gamma \mathrm{d}\widetilde{w}_j \mathrm{d}\widetilde{w}_j^{\prime} O_j P^{\prime} O_{j^{\prime}}\ ,
\end{aligned}
\end{equation}
from which we see the channel $\mathcal{N}$ can map a Pauli operator $P$ into another Pauli operator $P'$ with the support size (1) increase by $1$, or (2) invariant, or (3) decrease by $1$. Therefore, we expect the PTM $W(P, P')\neq 0$ only if $|m-m'|\le 1$, where $m=|P|$ and $|m'|=|P'|$ are the support size of Pauli operators $P$ and $P'$.

Explicitly, using the definition of PTM in \cref{seq-def-PTM}, up to linear order of $dt$ (note that $\mathrm{d}\widetilde{w}_j$ is of order $\sqrt{dt}$), and using the fact that $\mathbb{E}[\mathrm{d}\widetilde{w}_j]=0$, and $\mathbb{E}[ \mathrm{d}\widetilde{w}_j\mathrm{d}\widetilde{w}_{j'}]=dt\delta_{jj'}$, we find
\begin{equation}
\begin{aligned}
W_\mathcal{N}(P, P') &= \left(1 - \frac{8}{3} m\, \gamma\, \mathrm{d}t \right) \delta_{P, P'} + 4 \gamma\, \mathrm{d}t \sum_{j=1}^n \underset{O_j}{\mathbb{E}} \left[ \delta_{P, \frac{1}{2} \{O_j, P'\}} \right]\ .
\end{aligned}
\end{equation}
By further acting on $w_{\mathcal{M}_t}(P')$ and averaging over $O_j$ among all single-qubit Pauli operators with equal probability, we arrive at
\begin{equation}\label{seq-wP-PTM}
\begin{split}
& w_{\mathcal{M}_{t+ndt}}(P) = \sum_{P'} W(P, P')\, w_{\mathcal{M}_t}(P') \\
& =\left(1-\frac{8}{3} m \gamma \mathrm{d}t\right) w_{\mathcal{M}_t}(P)+4 \gamma \mathrm{d}t \sum_{P^{\prime}} \mathbb{E}_{\mathcal{N}\in\mathcal{F}_{\mathcal{N}}} \sum_{j=1}^n \delta_{P, \frac{\left\{O_j, P^{\prime}\right\}}{2} }w_{\mathcal{M}_t}(P')\\
&=\left(1-\frac{8}{3} m \gamma \mathrm{d}t\right) w_{\mathcal{M}_t}(P)+4 \gamma \mathrm{d}t \sum_{P^{\prime}} w_{\mathcal{M}_t}(P')\left[\frac{1}{3} \delta_{m^{\prime}, m-1} \delta_{P, O P^{\prime}}+\delta_{m^{\prime}, m+1} \delta_{O P, P^{\prime}}\right]\ ,
\end{split}
\end{equation}
where we have defined $\delta_{P,OP'}=1$ if $P=OP'$ for some single-qubit Pauli operator $O$, and $m=|P|$, $|m'|=|P'|$ are the support sizes.

The sum over $P'$ term accounts for transitions into the state $P$ from other Pauli states, in the bracket of which the first term describes a weight-increasing process, where a support-$(m - 1)$ Pauli operator $P'$ becomes $P$, while the second term describes a weight-decreasing process, where a support-$(m + 1)$ Pauli operator $P'$ becomes $P$.

Dividing the above \cref{seq-wP-PTM} by $ndt$, we arrive at a differential equation for $w_{\mathcal{M}_t}(P)$ in the continuous-time limit:
\begin{equation}\label{seq-wP-ODE}
\frac{dw_{\mathcal{M}_t}(P)}{\mathrm{d}t}=-\frac{8m}{3n} \gamma w_{\mathcal{M}_t}(P)+\frac{4 \gamma}{n} \sum_{P^{\prime}}w_{\mathcal{M}_t}(P')\left[\frac{1}{3} \delta_{m^{\prime}, m-1} \delta_{P, O P^{\prime}}+\delta_{m^{\prime}, m+1} \delta_{O P, P^{\prime}}\right].
\end{equation}
This differential equation can be further simplified. Due to the permutation symmetry among all qubits and the initial condition symmetric among all qubits, we can assure that the Pauli weight depends only on the support size $m = |P|$ of Pauli operator $P$: 
\begin{equation}
w_{\mathcal{M}_t}(P)=w_{\mathcal{M},m}(t)\ ,\qquad m=|P|\ ,
\end{equation}
and the Pauli weights satisfy a closed system of ODEs from \cref{seq-wP-ODE}:
\begin{equation}\label{seq-wm-ODE}
\frac{\mathrm{d}w_{\mathcal{M}, m}}{\mathrm{d}t} = -\frac{8 m}{3 n} \gamma\, w_{\mathcal{M}, m} + \frac{4 m}{3 n} \gamma\, w_{\mathcal{M}, m-1} + \frac{4 (n - m)}{n} \gamma\, w_{\mathcal{M}, m+1}\ ,\qquad (0\le m\le n)
\end{equation}
This set of ODEs govern the time evolution of the Pauli weights $w_{\mathcal{M},m}(t)$ ($0\le m\le n$) for the weak measurement-and-prepare channel, with boundary conditions $w_{\mathcal{M},-1} = w_{\mathcal{M},n+1} = 0$, and initial conditions
\begin{equation}
w_{\mathcal{M},m}(0)=\delta_{m,0}\ ,
\end{equation}
as is evident from the channel definition.

\subsubsection{Analytical solution of the Pauli weights}

The coupled linear ODEs in \cref{seq-wm-ODE} can be compactly expressed in matrix form as
\begin{equation}
\frac{\mathrm{d}\boldsymbol{w}_{\mathcal{M}}}{\mathrm{d}t} = D\  \boldsymbol{w}_{\mathcal{M}}\ ,\qquad \boldsymbol{w}_{\mathcal{M}}=(w_{\mathcal{M},0},w_{\mathcal{M},1},\cdots,w_{\mathcal{M},n})^T\ ,
\end{equation}
where $\boldsymbol{w}_{\mathcal{M}}$ is the column vector of Pauli weights across all support sizes $0\le m\le n$. The evolution is governed by the $(n + 1) \times (n + 1)$ tridiagonal matrix $D$, the nonzero entries of which are given by:
\begin{equation}
D_{m,m'} = 
\begin{cases}
-\dfrac{8}{3} \dfrac{m}{n} \gamma, & \text{if } m' = m \\[2ex]
\ \ \dfrac{4}{3} \dfrac{m}{n} \gamma,  & \text{if } m' = m - 1 \\[2ex]
\dfrac{4(n - m)}{n} \gamma,         & \text{if } m' = m + 1 \\[2ex]
0,                                & \text{otherwise}
\end{cases}
\end{equation}


The matrix $D$ is non-Hermitian, the eigenvalues and eigenvectors of which are typically complicated to obtain. However, due to its special form, an analytical solution can be obtained by mapping the problem to a well-understood physical system: a quantum spin in a magnetic field. This mapping is realized through a similarity transformation $D' = V D V^{-1}$, which converts $D$ into a Hermitian matrix without changing its eigenvalues.

More generically, we consider a tridiagonal matrix $D$ of the form with nonzero entries $D_{m,m} = c_1 m + d_0$, $D_{m,m+1} = c_2 (n - m)$, and $D_{m,m-1} = c_0 m$. We define a diagonal matrix $V$ with elements
\[
V_{m,m'} = \delta_{m,m'} \, \frac{\left( \sqrt{c_2 / c_0} \right)^m}{\sqrt{m! (n - m)!}}\ .
\]
A straightforward calculation of the similarity transformation $D' = V D V^{-1}$ yields the following matrix elements:
\begin{align}
D'_{m,m'}=(V D V^{-1})_{m,m'} =
\begin{cases}
\sqrt{c_0 c_2} \sqrt{(m + 1)(n - m)}, & \text{if } m' = m + 1 \\[1ex]
\sqrt{c_0 c_2} \sqrt{m(n - m + 1)}, & \text{if } m' = m - 1 \\[1ex]
c_1 m + d_0, & \text{if } m' = m \\[1ex]
0, & \text{otherwise}
\end{cases}
\end{align}

These matrix elements are structurally identical to those of spin operators in a spin-$s$ system, where the total number of qubits is $n = 2s$, and the weight index $m$ is related to the $z$-component of spin via $m = s + s_z$. In particular, the off-diagonal elements correspond to the matrix elements of the spin raising and lowering operators $S_+$ and $S_-$, while the diagonal term matches $S_z$. As a result, the transformed matrix $D'$ can be written as a linear combination of spin-$s$ operators:
\begin{equation}\label{spin_op}
D' = V D V^{-1} = (c_1 s + d_0)\, I + c_1 S_z + 2 \sqrt{c_0 c_2}\, S_x\ ,
\end{equation}
which corresponds to the Hamiltonian of a spin-$s$ particle in an effective magnetic field $\boldsymbol{B}=- (2 \sqrt{c_0 c_2},\, 0,\, c_1)$. Particularly, it is known that the eigenvalues of such a Hamiltonian are equally spaced, given by

\begin{equation}
\lambda_p=p\Lambda_1+(n-p)\Lambda_2\ ,\qquad  \Lambda_{1,2}=\frac{1}{2}\left(c_1 \pm \sqrt{c_1^2+4 c_0 c_2}\right)\ .
\end{equation}

In our case here, the constants are given by
\begin{equation}
c_0=\frac{4\gamma}{3n}\ ,\quad c_1=-\frac{8\gamma}{3n}\ ,\quad c_2=\frac{4\gamma}{n}\ ,\quad d_0=0\ .
\end{equation}
Substituting the coefficients from the table, we obtain the exact $p$-th eigenvalue $\lambda_p$ of the evolution matrix $D$:
\begin{equation}
\lambda_p = \frac{4\gamma}{3} - \frac{4\gamma}{n} p \qquad \text{for } p = 0, 1, \ldots, n\ .
\end{equation}
The largest eigenvalue corresponds to $p = 0$, giving $\lambda_0 = \frac{4\gamma}{3}$. The general solution for the Pauli weight vector is a linear combination of the corresponding eigenvectors $\boldsymbol{v}^{(p)}$ of $D$:
\begin{equation}
\widetilde{\boldsymbol{w}}_{\mathcal{M}}(t) = \sum_{p=0}^n c_p\, e^{\lambda_p t} \boldsymbol{v}^{(p)},
\end{equation}
where the coefficients $c_p$ are determined by the initial condition $\boldsymbol{w}_{\mathcal{M}}(0) = (1, 0, \ldots, 0)^T$. After solving for the coefficients, one arrives at the solution:
\begin{equation}
w_{\mathcal{M},m}(t)=w_{\mathcal{M},0}(t) \left( \frac{1 - \ee^{-\frac{16\gamma t}{3n}}}{3 + \ee^{-\frac{16\gamma t}{3n}}} \right)^m\ ,\qquad w_{\mathcal{M},0}(t)=w_{\mathcal{M}_t}(\id)=\left(3\,\ee^{\frac{4\gamma t}{3n}}+\ee^{-\frac{4\gamma t}{n}}\right)^n\ .
\end{equation}
We can define the Pauli weight of the channel after normalization as
\begin{equation}
\widetilde{w}_{\mathcal{M},m}(t)=\frac{w_{\mathcal{M},m}(t)}{w_{\mathcal{M},0}(t)}=\left( \frac{1 - \ee^{-\frac{16\gamma t}{3n}}}{3 + \ee^{-\frac{16\gamma t}{3n}}} \right)^m\ ,\qquad \rightarrow \qquad \widetilde{w}_{\mathcal{M},m}(t)=\big(\widetilde{w}_{\mathcal{M},1}(t)\big)^m\ .
\end{equation}
In the long-time limit $t \to \infty$, one has
\begin{equation}
\widetilde{w}_{\mathcal{M},m}(t\rightarrow\infty)=\frac{w_{\mathcal{M},m}(t\rightarrow\infty)}{w_{\mathcal{M},0}(t\rightarrow\infty)}=\frac{1}{3^m}\ ,
\end{equation}
This agrees with the fact that $\sigma_t/\mathrm{Tr}(\sigma_t)=K_t^{\dagger} K_t/\text{Tr}(K_t^{\dagger} K_t)$ is a product state. It is a product state of mixed state unless $t\rightarrow \infty$, when it becomes a pure product state. 

\cref{Pauli_weight_m} shows the Pauli weight $\tilde{w}_{\sigma,m}(t)$ at finite time $t$ calculated analytically from the above (see caption for detailed descriptions).

\begin{figure}[htbp]
\begin{center}
\includegraphics[width=240pt]{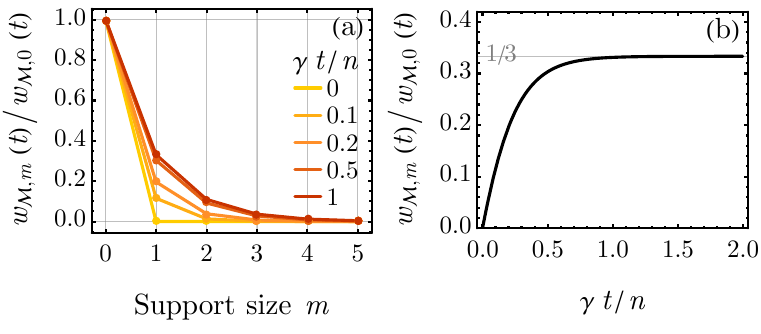}
\caption{
Rescaled Pauli weight $w_{\mathcal{M},m}(t)/w_{\mathcal{M},0}(t)$ (a) as a function of  the support size $m$ for different fixed time $t$ (see legend), and (b) as a function of time $t$, which approaches $1/3$ as $t\rightarrow\infty$.}
\label{Pauli_weight_m}
\end{center}
\end{figure}


\subsection{Error and shadow norm}

We recover the initial state $\rho_0$ by defining shadow of the measured trajectories $\mathcal{O}$:
\begin{equation}\label{seq:shadow-C}
\hat{\rho}_{0,\mathcal{O}}=\mathcal{M}_t^{-1}(\sigma_t(\mathcal{O}))\ .
\end{equation}
It is straightforward to show that the mean value of the shadow of various trajectories $\mathcal{O}$ from the same initial state $\rho_0$ is the initial state $\rho_0$:
\begin{equation}
\underset{p_t(\mathcal{O})}{\mathbb{E}}\ \hat{\rho}_{0,\mathcal{O}}=\underset{\mathcal{O}}{\mathbb{E}}\ p_t(\mathcal{O})\mathcal{M}_t^{-1}(\sigma_t(\mathcal{O})) = \mathcal{M}_t^{-1}\Big[ \underset{\mathcal{O}}{\mathbb{E}}\ p_t(\mathcal{O})\sigma_t(\mathcal{O})\Big] =\mathcal{M}_t^{-1} \Big[\mathcal{M}_t(\rho_0)\Big]=\rho_0\ .
\end{equation}
This is the key of our method.

We now examine the variance of an operator $O$ being measured from the shadow. We define
\begin{equation}
\hat{o}=\text{Tr}(\hat{\rho}_{0,\mathcal{O}}O)\ ,\qquad \rightarrow\qquad \mathbb{E}\ [\hat{o}]=\text{Tr}(\rho_0 O)\ .
\end{equation}
Its variance is given by
\begin{equation}
\begin{split}
\text{Var}[\hat{o}]&=\mathbb{E}[(\text{Tr}(\hat{\rho}_{0,\mathcal{O}}O))^2]-(\text{Tr}(\rho_0O))^2 \le \mathbb{E}[(\text{Tr}(\hat{\rho}_{0,\mathcal{O}}O))^2] =\mathbb{E}[(\text{Tr}(\mathcal{M}_t^{-1}(O)\sigma_t(\mathcal{O}))^2]\\
&=2^{N_t}\underset{\mathcal{O}}{\mathbb{E}}\ p_t(\mathcal{O})(\text{Tr}(\mathcal{M}_t^{-1}(O)\sigma_t(\mathcal{O}))^2\ .
\end{split}
\end{equation}
Note that
\begin{equation}
p_t(\mathcal{O})=\text{Tr}(\sigma_t(\mathcal{O})\rho_0)\sim\frac{\text{Tr}(\sigma_t(\mathcal{O}))}{\text{Tr}{\id}}\ ,\qquad \sigma_t(\mathcal{O})=K_t^{\dag}(\mathcal{O}) K_t(\mathcal{O})\ ,
\end{equation}
we can estimate the Variance of $\hat{o}$ as
\begin{equation}
\text{Var}[\hat{o}]\le \underset{\mathcal{O}}{\mathbb{E}}\ \frac{2^{N_t}\text{Tr}(\sigma_t(\mathcal{O}))}{\text{Tr}{\id}}(\text{Tr}(\mathcal{M}_t^{-1}(O)\sigma_t(\mathcal{O}))^2=\text{Tr}\Big[\mathcal{M}_t^{-1}(O) \underset{\mathcal{O}}{\mathbb{E}}\ \frac{2^{N_t}\text{Tr}(\sigma_t(\mathcal{O}))}{\text{Tr}{\id}}\sigma_t(\mathcal{O}) \text{Tr}(\mathcal{M}_t^{-1}(O)\sigma_t(\mathcal{O}))\Big]\ .
\end{equation}
If we estimate $2^{N_t}\text{Tr}(\sigma_t(\mathcal{O}))\simeq C_t$ as a constant, we approximately have
\begin{equation}
\begin{split}
\text{Var}[\hat{o}]&\le \frac{C_t}{\text{Tr}{\id}}\text{Tr}\Big[\mathcal{M}_t^{-1}(O) \underset{\mathcal{O}}{\mathbb{E}}\ \sigma_t(\mathcal{O})\text{Tr}(\mathcal{M}_t^{-1}(O)\sigma_t(\mathcal{O}))\Big]=\frac{C_t}{\text{Tr}{\id}}\text{Tr}\Big[\mathcal{M}_t^{-1}(O) \mathcal{M}_t(\mathcal{M}_t^{-1}(O))\Big]=\frac{C_t}{\text{Tr}{\id}}\text{Tr}\Big[\mathcal{M}_t^{-1}(O) O\Big]\ .
\end{split}
\end{equation}
The constant $C_t$ can be fixed by noting that the above upper bound, which is the expectation value of the square of operator $O$, should be $1$ if $O=\id$ is the identity:
\begin{equation}
\frac{C_t}{\text{Tr}{\id}}\text{Tr}\Big[\mathcal{M}_t^{-1}(\id)\Big]=1\ ,\qquad \rightarrow \qquad \frac{C_t}{\text{Tr}{\id}}=\frac{1}{\text{Tr}\Big[\mathcal{M}_t^{-1}(\id)\Big]}\ .
\end{equation}
This gives the variance of each shadow:
\begin{equation}
\text{Var}[\hat{o}]\le \|O\|^2_\text{sh}\ ,
\end{equation}
where we have defined the shadow norm of operator $O$ as
\begin{equation}
\|O\|^2_\text{sh}=\frac{\text{Tr}\Big[\mathcal{M}_t^{-1}(O) O\Big]}{\text{Tr}\Big[\mathcal{M}_t^{-1}(\id)\Big]}\ .
\end{equation}
For Pauli operators $P$ with support size $m=|P|$, the shadow norm is
\begin{equation}
\|P\|^2_\text{sh}=\frac{\text{Tr}\Big[\mathcal{M}_t^{-1}(P) P\Big]}{\text{Tr}{\id}}= \frac{w_{\mathcal{M}_t}(\id)}{w_{\mathcal{M}_t}(P)} \qquad\rightarrow 3^m\qquad \text{as}\qquad t\rightarrow\infty\ .
\end{equation}

For $M$ shadows, the total variance will be $1/M$ factor smaller:
\begin{equation}
\text{Var}[\hat{o}]\le \frac{\|O\|^2_\text{sh}}{M}\ .
\end{equation}

\subsection{Efficient calculation of the shadow}

The shadow in \cref{seq:shadow-C} can be calculated efficiently due to the nice properties we showed earlier. First, recall that the Kraus operator $K_t(\mathcal{O})$ associated with trajectory $\mathcal{O}$ is a product of Kraus operators over each individual qubits:
\begin{equation}
\begin{split}
&K_t(\mathcal{O})=\mathcal{T}\prod_{t^{\prime}=0}^t K_{d t}\left(O_t^{\prime}, d o_t^{\prime}\right)=\otimes_{j=1}^n K_t^{(j)},\quad K_t^{(j)}=\mathcal{T}\prod_{j_{t^{\prime}}=j} K_{d t}\left(O_t^{\prime}, d o_t^{\prime}\right)\ ,\quad K_{dt}(O_t, d o_t)=\frac{1}{\sqrt{2}}\exp\left(-\gamma dt+ o_t \sqrt{\gamma dt} O_{t}\right)
\end{split}
\end{equation}
where $\mathcal{T}$ denotes the time-ordering operator, and each Kraus operator $K_t^{(j)}$ is the accumulated contributions of measurements of observables
$O_{t}$ acting on the $j_t=j$-th qubit. Note that each $K_t^{(j)}$ is simply a $2\times 2$ matrix acting on the $j$-th qubit.

Therefore, $\sigma_t(\mathcal{O})$ can be reexpressed in a tensor product form:
\begin{equation}
\sigma_t(\mathcal{O})=\otimes_{j=1}^n \sigma_t^{(j)}(\mathcal{O})\ ,\qquad \sigma_t^{(j)}(\mathcal{O})=K_t^{(j)\dagger}(\mathcal{O}) K_t^{(j)}(\mathcal{O})\ ,
\end{equation}
where each $\sigma_t^{(j)}(\mathcal{O})$ is simply a $2\times 2$ matrix. 

Then, because the channel $\mathcal{M}_t$ has a Pauli weight $w_{\sigma, m}(t) \equiv\left(w_{\sigma, 1}(t)\right)^m$, the action of $\mathcal{M}_t$ and $\mathcal{M}_t^{-1}$ simply decomposes into the action on each qubit. If we decompose $\sigma_t^{(j)}(\mathcal{O})$ in single-qubit Pauli basis as
\begin{equation}
\sigma_t^{(j)}(\mathcal{O})=\kappa_{0}^{(j)}I_j+\kappa_{0}^{(j)}X_j+\kappa_{0}^{(j)}Y_j+\kappa_{0}^{(j)}Z_j\ ,
\end{equation}
then we can define the single-qubit shadow (which is a $2\times 2$ matrix)
\begin{equation}
\hat{\rho}_{0,\mathcal{O}}^{(j)}=\mathcal{M}_t^{-1}(\sigma_t^{(j)})=\kappa_{0}^{(j)}I_j+\frac{1}{w_{\sigma,1}(t)}\Big(\kappa_{0}^{(j)}X_j+\kappa_{0}^{(j)}Y_j+\kappa_{0}^{(j)}Z_j\Big)\ ,
\end{equation}
and the shadow of the total system is simply the tensor product:
\begin{equation}
\hat{\rho}_{0,\mathcal{O}}=\otimes_{j=1}^n
\hat{\rho}_{0,\mathcal{O}}^{(j)}\ .
\end{equation}

\section{From measurement channel and reverse Petz channel to quantum diffusion}

In this appendix, we show how the weak measurement channel and its Petz recovery channel correspond to the partial differential equations (PDEs) of the quantum diffusion picture of our weak measurements. The key observation is that the density matrix after the quantum channels can be understood as a quantum generalization of the probability distribution in diffusion problem, which satisfies the PDEs of the diffusion.

\subsection{The measurement channel and the forward diffusion}

We first focus on the forward evolution of the system by a sequence of weak measurements. 

We consider a system of $n$ qubits. As introduced in the main text, the measurement channel $\mathcal{F}_t$ describes the averaged state of all the weak measurement trajectories $\mathcal{O}$. Earlier, we have derived the analytical expression of the measurement channel. For a density matrix $\rho_0=\frac{1}{2^n}\sum_i z_i P_i$, where $P_i$ runs over all the Pauli string operators
\begin{equation}\label{seq:m-channel}
\mathcal{F}_t(\rho_0)=\frac{1}{2^n}\sum_i w_{\mathcal{F}_t}(P_i)(t)z_i P_i\ ,\qquad w_{\mathcal{F}_t}(P)=w_{\mathcal{F},m}(t)=\exp\left(-\frac{4\gamma}{3n}m t\right)\ ,\qquad m=|P|\ .
\end{equation}
Here $m=|P|$ is the support size of Pauli operator $P$. Note that the Pauli weight $w_{\mathcal{F}_t}(P)=w_{\mathcal{F},m}(t)$ only depends on the support size $m$ of operator $P$.

We consider the measurement channel evolved state $\rho_t=\mathcal{F}_t(\rho_0)$ of some initial state (mixed or pure) $\rho_0$. In practice, the random weak measurement along each trajectory $\mathcal{O}$ leads to a random pure product state, since the single-qubit weak measurements eventually collapse each qubit to a single-qubit pure state. This motivates us to define the following  pure product coherent state basis:
\begin{equation}
|\boldsymbol{n}\rangle\langle\boldsymbol{n}|=\prod_{j=1}^n \left(\frac{I_j+\boldsymbol{n}_j\cdot\boldsymbol{\sigma}_j}{2}\right)\ ,\qquad \boldsymbol{n}=(\boldsymbol{n}_1,\boldsymbol{n}_2,\cdots,\boldsymbol{n}_n)\ ,\quad |\boldsymbol{n}_j|=1\ ,\quad \boldsymbol{n}_j=(n_{j,x},n_{j,y},n_{j,z})\ ,
\end{equation}
where $I_j$ is the identity matrix of the $j$-th qubit, and $\boldsymbol{\sigma}_j=(\sigma_{j,x},\sigma_{j,y},\sigma_{j,z})$ of the $j$-th qubit. 
As defined, $\boldsymbol{n}_j$ is a unit vector living on the Bloch sphere of the $j$-th qubit ($1\le j\le n$). the coherent states $|\boldsymbol{n}\rangle$ form an overcomplete basis of the system, satisfying 
\begin{equation}
\int d\boldsymbol{n}|\boldsymbol{n}\rangle\langle\boldsymbol{n}|=\frac{I}{2^n}\ ,\qquad \int d\boldsymbol{n}=\text{Tr}\left(\frac{I}{2^n}\right)=1\ ,
\end{equation}
where $d\boldsymbol{n}$ is the Haar measure in the space of $\boldsymbol{n}$, and $I$ is the identity operator. Generically, the density matrix $\rho_t$ can be expanded in terms of the coherent state projector basis as:
\begin{equation}\label{eq-rho-Prep}
\rho_t=\int d\boldsymbol{n}p(\boldsymbol{n},t)|\boldsymbol{n}\rangle\langle\boldsymbol{n}|\ ,\qquad \int d\boldsymbol{n}p(\boldsymbol{n},t)=1\ ,
\end{equation}
which is known as the Glauber-Sudarshan P-representation, and $p(\boldsymbol{n},t)$ can be called the quasi-probability. For Hermitian operator $\rho_t$, $p(\boldsymbol{n},t)$ is real. We also note that the choice of $p(\boldsymbol{n},t)$ is not unique due to the overcompleteness of coherent state basis.

Generically $p(\boldsymbol{n},t)$ is not necessarily positive. Our discussion below does not require $p(\boldsymbol{n},t)$ to be non-negative; but we will show $p(\boldsymbol{n},t)$ is similar to the probability distribution in the classical diffusion problem. 

\emph{Physical Meaning and relation to previous considerations}. If all the $p(\boldsymbol{n},t)\ge0$, then $\rho_t$ has no entanglement between qubits, but only has classical correlations. In this case, $\rho_t$ can be exactly understood as an esemble of random pure random product states $|\boldsymbol{n}\rangle\langle\boldsymbol{n}|$ with probability distribution $p(\boldsymbol{n},t)$, which are evolved from an initial state esemble of random pure product states with probability distribution $p(\boldsymbol{n},0)$. We note that for generic states, we have been denoting the probability distribution of states as $p(\boldsymbol{z},t)$, in terms of variable $\boldsymbol{z}$ which has components $z_P$ being the coefficients of Pauli matrices $P$ in the density matrix $\rho=\frac{1}{2^n}\sum_i z_i P_i$ (see \cref{seq:m-channel}). When the initial states of our setup are restricted to random pure product states $|\boldsymbol{n}\rangle$, they will remain random pure product states during the weak measurement evolution; their corresponding variable $\boldsymbol{z}$ is a function of $\boldsymbol{n}$ determined by $|\boldsymbol{n}\rangle\langle \boldsymbol{n}|=\frac{1}{2^n}\sum_i z_i P_i$, and by changing of variables we can re-express $p(\boldsymbol{z},t)$ as $p(\boldsymbol{n},t)$. Accordingly, the nonlinear SDE for variable $\boldsymbol{z}$ becomes a nonlinear SDE for variable $\boldsymbol{n}$ on the Bloch spheres.

\subsubsection{The single-qubit example}

We first consider the example of a single-qubit system, for which $\boldsymbol{n}=\boldsymbol{n}_1$. In this case, the state $\rho_t=\mathcal{F}_t(\rho_0)$ at time $t$ is given by
\begin{equation}\label{eq-rho-Prep-v1}
\rho_t=\int d\boldsymbol{n}_1 p(\boldsymbol{n}_1,t)\frac{I_1+\boldsymbol{n}_1\cdot \boldsymbol{\sigma}_1}{2}\ .
\end{equation}
By \cref{seq:m-channel}, we have $p(\boldsymbol{n}_1,t)$ satisfying
\begin{equation}\label{eq:p-t-dependence}
\int d\boldsymbol{n}_1 p(\boldsymbol{n}_1,t)=1\ ,\qquad \int d\boldsymbol{n}_1 p(\boldsymbol{n}_1,t)\boldsymbol{n}_1= e^{-\frac{4\gamma}{3}t} \int d\boldsymbol{n}_1 p(\boldsymbol{n}_1,0)\boldsymbol{n}_1\ .
\end{equation}

We now show that this correspond to a diffusion problem of quasi-probability distribution $p(\boldsymbol{n}_1,t)$. As we have explained, $\rho_t$ in the form of \cref{eq-rho-Prep-v1} can be understood as the average of an ensemble of random pure single-qubit states $|\boldsymbol{n}_1\rangle$ with a ``probability" distribution $p(\boldsymbol{n}_1,t)$. Particularly, $\rho_0$ corresponds to an esemble of initial states $|\boldsymbol{n}_1\rangle$ with a ``probability" distribution $p(\boldsymbol{n}_1,0)$. In this setup, the nonlinear SDE in terms of variable $\boldsymbol{z}$, when transformed into variable $\boldsymbol{n}_1$, gives a diffusion PDE for the ``probability" distribution $p(\boldsymbol{n}_1,t)$ restrited on the Bloch sphere:
\begin{equation}\label{eq:diffusion-S2}
\partial_t p(\boldsymbol{n}_1,t)=D\nabla_\perp^2 p(\boldsymbol{n}_1,t)\ ,
\end{equation}
where $\nabla_\perp^2$ is the Laplacian $\nabla_{\boldsymbol{n}_1}^2$ restricted on the Bloch sphere $|\boldsymbol{n}_1|=1$, and $D$ is the diffusion constant which can be derived from the nonlinear SDE restricted on Bloch sphere. Here we shall determine $D$ by matching the evolution of the measurement channel. Note that the PDE in \cref{eq:diffusion-S2} has no drifting force, since the diffusion is spherical symmetric on the Bloch sphere. (The drifting force $\boldsymbol{f}$ in the nonlinear SDE for generic state $\rho$ is solely driving the state towards random pure product states. So for random pure product states, the drifting force is zero.)

We can solve the diffusion equation by expanding $p(\boldsymbol{n}_1,t)$ in terms of spherical harmonics as
\begin{equation}\label{eq:p-Ylm}
p(\boldsymbol{n}_1,t)=\sum_{lm}\widetilde{p}_{lm}(t)Y_{lm}(\theta_1,\varphi_1)\ ,\quad \widetilde{p}_{lm}(t)=\int d\boldsymbol{n}_1 p(\boldsymbol{n}_1,t)Y_{lm}(\theta_1,\varphi_1)\ ,\quad \boldsymbol{n}_1=(\sin\theta_1\cos\varphi_1,\sin\theta_1\sin\varphi_1,\cos\theta_1)\ .
\end{equation}
Here $(\theta_1,\varphi_1)$ is the spherical coordinate of $\boldsymbol{n}_1$. The angular momentum $l=0,1,2,\cdots$, and $|m|\le l$. \cref{eq:diffusion-S2} tells us
\begin{equation}\label{eq:sol-Ylm}
\partial_t\widetilde{p}_{lm}(t)=-Dl(l+1)\widetilde{p}_{lm}(t)\ ,\quad \rightarrow \quad \widetilde{p}_{lm}(t)=e^{-Dl(l+1)t}\widetilde{p}_{lm}(0)\ .
\end{equation}
Note that $Y_{10}(\theta_1,\varphi_1)=\sqrt{\frac{3}{8\pi}}v_{1,z}$, and $Y_{1\pm1}(\theta_1,\varphi_1)=\sqrt{\frac{3}{8\pi}}(v_{1,x}\pm i v_{1,y})$, from \cref{eq:p-t-dependence,eq:p-Ylm}, we find that $\int d\boldsymbol{n}_1 p(\boldsymbol{n}_1,t)\boldsymbol{n}_1\propto \widetilde{p}_{1m}(t)$, so for $l=1$,
\begin{equation}
\widetilde{p}_{1m}(t)=e^{-\frac{4\gamma}{3}t}\widetilde{p}_{1m}(0)\ .
\end{equation}
Comparing with \cref{eq:sol-Ylm} at $l=1$, we find the diffusion constant $D=2\gamma/3$. Thus, $\rho_t$ from the measurement channel is equivalent to the forward diffusion problem of probability distribution $p(\boldsymbol{n}_1,t)$ satisfying the forward diffusion PDE:
\begin{equation}\label{eq:diffusion-S2-final}
\partial_t p(\boldsymbol{n}_1,t)=\frac{2\gamma}{3}\nabla_\perp^2 p(\boldsymbol{n}_1,t)\ .
\end{equation}
Lastly, $\widetilde{p}_{00}(t)\equiv \frac{1}{\sqrt{4\pi}}$, which gives the normalized total probability $\int d\boldsymbol{n}_1p(\boldsymbol{n}_1,t)=\sqrt{4\pi}\widetilde{p}_{00}(t)=1$.

\subsubsection{$n$-qubits}

For $n$-qubit systems, the generalization is straightforward: the quasi-probability $p(\boldsymbol{n},t)$ of $\rho_t$ diffuses simultaneously on the Bloch sphere of every qubit. Accordingly, it satisfy a forward diffusion equation
\begin{equation}\label{eq:diffusion-nS2}
\partial_t p(\boldsymbol{n},t)=\frac{2\gamma}{3 n}\nabla_\perp^2 p(\boldsymbol{n},t)\ ,\qquad \nabla_\perp^2=\nabla_{\perp,1}^2+\nabla_{\perp,2}^2+\cdots +\nabla_{\perp,n}^2\ ,
\end{equation}
where $\nabla_{\perp,j}^2$ is the Laplacian with respect to $\boldsymbol{n}_j$ restricted on the bloch sphere $|\boldsymbol{n}_j|=1$ of the $j$-th qubit. Namely, $\nabla_\perp^2$ is the Laplacian on the manifold of $\boldsymbol{n}\in \otimes_{j=1}^n S_j^2$, where $S_j^2$ is the Bloch sphere of the $j$-th qubit. It is straightforward to verify that such a diffusion equation for the quasi-probability $p(\boldsymbol{n},t)$ is equivalent to the state $\rho_t$ evolved by the measurement channel as given in \cref{seq:m-channel}.

Physically, this is equivalent to setting the initial state essemble to be the random pure product states $|\boldsymbol{n}\rangle\langle \boldsymbol{n}|$ with a probability distribution $p(\boldsymbol{n},0)$, and evolve every state by a sequence of random weak measurements, where the resulting probability will be $p(\boldsymbol{n},t)$ given by \cref{eq:diffusion-nS2}. Here $p(\boldsymbol{n},t)$ is related to the probability distribution $p(\boldsymbol{z},t)$ we studied earlier for generic states by the variable transformation $|\boldsymbol{n}\rangle\langle \boldsymbol{n}|=\frac{1}{2^n}\sum_i z_i P_i$ from $\boldsymbol{z}$ to $\boldsymbol{n}$. Accordingly, the SDE for $\boldsymbol{z}$ transforms into a SDE for $\boldsymbol{n}$, and the PDE for $p(\boldsymbol{z},t)$ transforms to the PDE in \cref{eq:diffusion-nS2} for $p(\boldsymbol{n},t)$.

\subsection{The Petz recovery channel and the backward diffusion}

We now consider the Petz recovery channel of the measurement channel. For states $\rho_t=\mathcal{F}_t(\rho_0)$ evolved by the measurement channel $\mathcal{F}_t$ defined above, the Petz recovery channel reversing from final time $T$ to time $T-t$ is defined as
\begin{equation}
\mathcal{R}_t(\sigma)=\rho^{1/2}_{T-t}\mathcal{F}_t^\dagger (\rho_T^{-1/2}\sigma \rho_T^{-1/2})\rho^{1/2}_{T-t}\ ,\qquad \rightarrow \qquad \mathcal{R}_t(\rho_T)=\rho_{T-t}\ ,
\end{equation}
where $\mathcal{F}_t^\dagger$ is the conjugate of the measurement channel $\mathcal{F}_t$, defined by $\text{Tr}(\mathcal{F}_t(\rho_1)\rho_2)=\text{Tr}(\rho_1\mathcal{F}_t^\dagger(\rho_2))$. In the case of measurement channel in \cref{seq:m-channel}, it is easy to see that $\mathcal{F}_t^\dagger=\mathcal{F}_t$.

We consider the infinitesimal Petz recovery channel from $\rho_{t+dt}$ to $\rho_t$ given by
\begin{equation}\label{eq-Rdt}
\mathcal{R}_{dt}(\sigma)=\rho^{1/2}_{t}\mathcal{F}_{dt}^\dagger (\rho_{t+dt}^{-1/2}\sigma \rho_{t+dt}^{-1/2})\rho^{1/2}_{t}\ ,
\end{equation}
and show it is analogous to the backward diffusion in the diffusion model. To the linear order of $dt$, we have
\begin{equation}\label{eq:dRdt}
\frac{d\sigma}{dt}=\frac{\mathcal{R}_{dt}(\sigma)-\sigma}{dt}=\rho^{1/2}_{t}\frac{\mathcal{F}_{dt}-1}{dt} (\rho_{t}^{-1/2}\sigma \rho_{t}^{-1/2})\rho^{1/2}_{t}+\frac{\rho^{1/2}_{t}\rho_{t+dt}^{-1/2}-1}{dt}\sigma +\sigma \frac{\rho_{t+dt}^{-1/2}\rho^{1/2}_{t}-1}{dt}\ .
\end{equation}
We rewrite the operators involved above in the Glauber-Sudarshan P-representation as:
\begin{equation}\label{eq:rhot-P-rep}
\rho_t=\int d\boldsymbol{n}p(\boldsymbol{n},t)|\boldsymbol{n}\rangle\langle\boldsymbol{n}|\ ,\quad \rho_t^{1/2}=\int d\boldsymbol{n}\eta(\boldsymbol{n},t)|\boldsymbol{n}\rangle\langle\boldsymbol{n}|\ ,\quad \rho_t^{-1/2}=\int d\boldsymbol{n}\lambda(\boldsymbol{n},t)|\boldsymbol{n}\rangle\langle\boldsymbol{n}|\ ,
\end{equation}
and
\begin{equation}
\frac{\rho^{1/2}_{t}\rho_{t+dt}^{-1/2}-1}{dt}= \int d\boldsymbol{n}\beta(\boldsymbol{n},t)|\boldsymbol{n}\rangle\langle\boldsymbol{n}|\ .
\end{equation}
To better understand them, we first define the multiplication of operators in the Glauber-Sudarshan P-representation:
\begin{equation}\label{eq:P-rep-multiply}
A=\int d\boldsymbol{n}p_A(\boldsymbol{n},t)|\boldsymbol{n}\rangle\langle\boldsymbol{n}|\ ,\quad B=\int d\boldsymbol{n}p_B(\boldsymbol{n},t)|\boldsymbol{n}\rangle\langle\boldsymbol{n}|\ ,\quad AB=\int d\boldsymbol{n}p_{AB}(\boldsymbol{n},t)|\boldsymbol{n}\rangle\langle\boldsymbol{n}|\ ,
\end{equation}
where we can define a multiplication operation
\begin{equation}
p_{AB}(\boldsymbol{n},t)=(p_A*p_B)(\boldsymbol{n},t)\ ,
\end{equation}
where $*$ being an operation similar to convolution. In this sense, the functions defined in \cref{eq:rhot-P-rep} satisfy $\eta*\eta=p$, and $\lambda*\lambda*p=1$. 

By expressing state $\sigma$ as
\begin{equation}
\sigma=\rho_t=\int d\boldsymbol{n}q(\boldsymbol{n},t)|\boldsymbol{n}\rangle\langle\boldsymbol{n}|\ ,
\end{equation}
in terms of a function $q(\boldsymbol{n},t)$, we can formally rewrite \cref{eq:dRdt} in the Glauber-Sudarshan P-representation as a PDE:
\begin{equation}\label{eq:back-diff}
\partial_tq=\frac{2\gamma}{3n}\eta*[\nabla_\perp^2(\lambda*q*\lambda)]*\eta +\beta*q+q*\bar{\beta}\ ,
\end{equation}
where $\bar{\beta}$ is the complex conjugate of $\beta$, and we have used the fact that the forward evolution $\frac{\mathcal{F}_{dt}-1}{dt}=\frac{2\gamma}{3n}\nabla_\perp^2$ from \cref{eq:diffusion-nS2}.

We claim that \cref{eq:back-diff} is the quantum analog of the backward diffusion PDE of the classical diffusion model. To see this, we show it reduces to the classical backward diffusion in a properly defined classical limit.

We define the classical limit as the case where all the coherent states $|\boldsymbol{n}\rangle$ are orthonormal, namely,
\begin{equation}\label{eq:ortho-coh}
\langle\boldsymbol{n}|\boldsymbol{n}'\rangle=\delta(\boldsymbol{n}-\boldsymbol{n}')\ .
\end{equation}
For qubits, this is an unphysical limit, but we can understand it as being approached when we consider many copies of the coherent states, namely, $|\boldsymbol{n}\rangle\rightarrow |\boldsymbol{n}\rangle^{\otimes k}$ with $k\rightarrow\infty$, or equivalently, the large spin limit of each qubit (which is spin $1/2$). In this classical limit, the states in the Glauber-Sudarshan P-representation become effectively diagonal, and by the definition in \cref{eq:P-rep-multiply}, the $*$ multiplication simply becomes the normal multiplication. The functions we defined then reduces to the following (the variable $(\boldsymbol{n},t)$ are omitted for simplicity):
\begin{equation}
\eta=p^{1/2}\ ,\qquad \lambda=p^{-1/2}\ ,\qquad \beta=\frac{p^{1/2}(\boldsymbol{n},t)p^{-1/2}(\boldsymbol{n},t+dt)-1}{dt}=-\frac{\partial_t p}{2p}=-\frac{\gamma}{3n}\frac{\nabla_\perp^2 p}{p}\ .
\end{equation}
Substituting these function into \cref{eq:back-diff}, we find
\begin{equation}
\partial_tq=\frac{2\gamma}{3n}\left[p\nabla_\perp^2\left(\frac{q}{p}\right)-q\frac{\nabla_\perp^2 p}{p}\right]=\frac{2\gamma}{3n}\left[\nabla_\perp^2q-2\nabla_\perp\cdot\left(q\frac{\nabla_\perp p}{p}\right)\right]\ ,
\end{equation}
or in a more familiar form,
\begin{equation}\label{eq:back-diff-classical2}
\partial_tq(\boldsymbol{n},t)=\frac{2\gamma}{3n}\nabla_\perp^2q(\boldsymbol{n},t)-\frac{4\gamma}{3n}\nabla_\perp\cdot\left[q(\boldsymbol{n},t)\nabla_\perp \log p(\boldsymbol{n},t)\right]\ .
\end{equation}
This is exactly the backward diffusion PDE corresponding to the forward diffusion PDE in \cref{eq:diffusion-nS2}, where $\frac{4\gamma}{3n}\nabla_\perp \log p(\boldsymbol{n},t)$ plays the role of the backward drifting force.

\emph{Approaching the classical limit}. As we have argued above, if we replace each spin $1/2$ qubit by a large spin $S$ qudit (which has $2S+1$ states), \cref{eq:ortho-coh} will be asymptotically satisfied, and we can achieve the classical limit. In this case, state $|\boldsymbol{n}\rangle$ denotes the pure product of large spin coherent states along direction $\boldsymbol{n}_j$ on the $j$-th spin. One then prepares an esemble of initial states being pure product coherent states $|\boldsymbol{n}\rangle$ with a probability distribution $p(\boldsymbol{n},t)$, and perform a sequence of weak measurement for these large spin qudits. In the large spin limit, the measurement channel $\mathcal{F}_t$ matches exactly with the forward classical diffusion governed by \cref{eq:diffusion-nS2}, and the Petz recovery channel $\mathcal{R}_t$ matches exactly with the backward classical diffusion governed by \cref{eq:back-diff-classical2}.

\section{Lindbladian for twirled Petz map}

In this section, we give the expression of the Lindbladian corresponding to the twirled Petz recovery channel.

\subsection{Lindbladian of the measurement channel}

For an single $dt$ step of the measurement channel acting on the $j$-th qubit, the Lindbladian can be derived to be:
\begin{equation}
\mathcal{F}_{dt}(\rho_t)=e^{\mathcal{L}dt}(\rho_t)=\sum_{o_t=\pm1}K_{dt}(O_t, d o_t)\rho_t K_{dt}^\dag(O_t, d o_t)\ ,
\end{equation}
which gives the Lindbladian equation
\begin{equation}
\frac{d\rho_t}{dt}=\mathcal{L}[\rho_t]=\sum_{\mu=x,y,z}\left(L_\mu\rho_t L_\mu^\dag -\frac{1}{2}\{L_\mu^\dag L_\mu,\rho_t\}\right)\ ,\qquad L_\mu=L_\mu^\dag=\sqrt{\frac{\gamma}{6}}\sigma_{j,\mu}\ .
\end{equation}
Here the $3$ jump operators $L_\mu$ runs over all the single-qubit Pauli operators $\sigma_{j,\mu}$ on qubits $j$. Note that there is no Hamiltonian in the Lindbladian of the measurement channel. This can be obtained by averaging over the linear SDE in the below over all random single qubit operators $O_{t}$ and averaging over $d\widetilde{w}_t$:
\begin{equation}
d\widetilde{\rho}_t = -\frac{\gamma}{2} [O_{t}, [O_{t}, \widetilde{\rho}_t]] dt + \sqrt{\gamma} \{ O_{t}, \widetilde{\rho}_t \} d\widetilde{w}_t\ .
\end{equation}
Since $L_\mu=L_\mu^\dagger$, we conclude that the conjugate channel $\mathcal{F}_{dt}=\mathcal{F}_{dt}^\dagger$, and the conjugate Lindbladian $\mathcal{L}^\dag=\mathcal{L}$.

\subsection{Lindbladian of the twirled Petz recovery channel}

The generic infinitesimal twirled Petz recovery channel from a generic $\rho_{t+dt}$ to $\rho_t$ is defined as
\begin{equation}\label{eq-tiwrled-petz}
\widetilde{\mathcal{R}}_{dt}(\sigma)=\int_{-\infty}^\infty f(\tau) \mathcal{R}_{dt}^\tau(\sigma)d\tau\ ,\qquad \mathcal{R}_{dt}^\tau(\sigma)=\rho^{\frac{1-i\tau}{2}}_{t}\mathcal{F}_{dt}^\dagger \left(\rho_{t+dt}^{\frac{-1+i\tau}{2}}\sigma \rho_{t+dt}^{\frac{-1-i\tau}{2}}\right)\rho^{\frac{1+i\tau}{2}}_{t}\ ,
\end{equation}
where $f(\tau)=\frac{1}{2[\cosh(\pi\tau)+1]}$. Setting $\tau=0$ gives back the untwirled Petz recovery channel.
According to Ref. \cite{Kwon_2019}, the infinitesimal rotated Petz map $\mathcal{R}_{dt}^\tau(\sigma)$ can be written as a Lindbladian form $\mathcal{R}_{dt}^\tau(\sigma)=e^{\mathcal{L}_{B,\tau}dt}(\sigma)$, with the Lindbladian
\begin{equation}
\mathcal{L}_{B,\tau}(\sigma)=-i[H_{B,\tau},\sigma]+\sum_\mu \left(L_{B,\tau,\mu}\sigma L_{B,\tau,\mu}^\dag -\frac{1}{2}\{L_{B,\tau,\mu}^\dag L_{B,\tau,\mu},\sigma\}\right)\ ,
\end{equation}
where
\begin{equation}
L_{B,\tau,\mu}=\rho^{\frac{1-i\tau}{2}}_{t}L_\mu^\dagger \rho^{\frac{-1+i\tau}{2}}_{t}\ ,\quad H_{B,\tau}=-\frac{i}{2}\left[\frac{d (\rho^{\frac{1-i\tau}{2}}_{t})}{dt}\rho^{\frac{-1+i\tau}{2}}_{t}+\frac{1}{2}\sum_{\mu}\rho^{\frac{1-i\tau}{2}}_{t}L_\mu^\dagger L_\mu \rho^{\frac{-1+i\tau}{2}}_{t}\right]+h.c.\ .
\end{equation}
Particularly, for our measurement channel, the expression of $H_{B,\tau}$ can be further simplified by noting that $L_\mu^\dagger L_\mu=\frac{\gamma}{6}$ here is a constant. The Lindbladian for the twirled Petz channel is then given by
\begin{equation}
\widetilde{\mathcal{L}}_{B}(\sigma)=\int_{-\infty}^\infty f(\tau)\mathcal{L}_{B,\tau}(\sigma)d\tau\ ,\qquad \widetilde{\mathcal{R}}_{dt}(\sigma)=e^{\widetilde{\mathcal{L}}_{B}dt}(\sigma)\ .
\end{equation}

For our purpose here, we apply the twirled Petz recovery channel $\widetilde{\mathcal{R}}_{dt}$ with $\rho_{t+dt}$ and $\rho_t$ in \cref{eq-tiwrled-petz} replaced by the local reduced density matrix $\rho^{S_j}_{t}$ to $\rho^{S_j}_{t-dt}$ in subregion $S_j$ around the $j$-th qubit (defined below). The resulting local twirled Petz recover channel is denoted as $\widetilde{\mathcal{R}}_{dt}^{S_j}$, which acts locally on the subregion $S_j$.

\subsection{Subsystem-Based Density Matrix Estimation}

We now focus on a forward weak measurement step $\mathcal{F}_{\delta t}$ acting on qubit $j$ in a finite small time $\delta t$ (a discretized or Trotterized Linbladian dynamics) and explain how to construct its Petz recovery map $\widetilde{R}_{\delta t}^{S_j}$. Consider a subsystem $S_j$ centered around qubit $j$, chosen to include the region that qubit $j$ is entangled with (which is denoted as region $B$). Denoting qubit $j$ as $A$, we define $S_j = A \cup B$, with the remainder of the system denoted as $C$. We impose the condition
\begin{equation}
I(A : C | B) \simeq 0,
\end{equation}
which ensures that $A$ is not entangled with $C$.

\begin{figure}[htbp]
\begin{center}
\includegraphics[width=30mm]{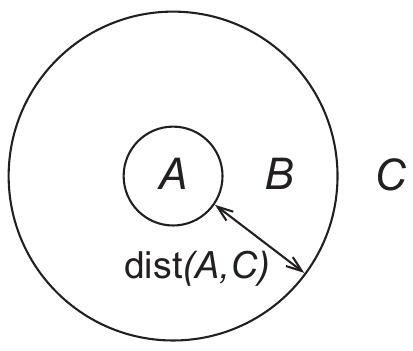}
\caption{Illustration of regions $A, B$ and $C.$ Region $A$ is the qubit $j$ that $\mathcal{F}_{\delta t}$ acts on, and $S_j= A \cup B$.}
\label{figS_ABC}
\end{center}
\end{figure} 

To characterize the reduced forward dynamics, we compute the reduced density matrix (RDM) on each subsystem $S_j$ via the following procedure:

\begin{enumerate}
    \item For each $S_j$, compute the initial reduced density matrix:
    \begin{equation}
        \rho_0^{S_j} = \frac{1}{2^{|S_j|}} \sum_{P \in S_j} z_{P,0} P.
    \end{equation}

    \item Verify whether $\rho_0^{S_j}$ is semi-positive definite (SPD). If it is not, deform $\rho_0^{S_j}$ into a valid SPD matrix. This step modifies the Pauli coefficients:
    \begin{equation}
        z_{P,0} \longrightarrow z_{P,0}^{(S_j)}.
    \end{equation}

    \item Evolve the RDM under the decoherence model:
    \begin{equation}
        \rho_t^{S_j} = \frac{1}{2^{|S_j|}} \sum_{P \in S_j} z_{P,t}^{(S_j)} P,
        \quad \text{with} \quad
        z_{P,t}^{(S_j)} = w_{\mathcal{F}_t}(P) z_{P,0}^{(S_j)} \quad \text{for } P \in S_j.
    \end{equation}
    
    Similarly, the RDM at time $t - \delta t$ is given by:
    \begin{equation}
        \rho_{t-\delta t}^{S_j} = \frac{1}{2^{|S_j|}} \sum_{P \in S_j} z_{P,t-\delta t}^{(S_j)} P.
    \end{equation}
\end{enumerate}

By construction, the SPD property of $\rho_0^{S_j}$ ensures the SPD of $\rho_t^{S_j}$. We assume that $\rho_t^{S_j}$ is a matrix of size $N_{S_j} \times N_{S_j}$, where $N_{S_j} = 2^{|S_j|}$.

Using the approximate data processing inequality~\cite{sang2024stabilitymixedstatequantumphases}, there exists a map channel $\widetilde{\mathcal{R}}_{\delta t}^{\mathcal{S}_j}$ acting on $S_j = A \cup B$ whose recovery quality satisfies:
\begin{equation}
\frac{1}{2 \ln 2} \left\| \widetilde{\mathcal{R}}^{\mathcal{S}_j}_{\delta t} \circ \mathcal{F}_{\delta t}[\rho_t] - \rho_t \right\|_1^2 
\leq I_{\rho_t}(A : C | B) - I_{ \mathcal{F}_{\delta t}[\rho_t]}(A : C | B) \ .
\end{equation}

 The explicit form of $\widetilde{\mathcal{R}}_{\delta t}^{\mathcal{S}_j}$ is given by the twirled Petz recovery map in \cref{eq-tiwrled-petz}, which depends on the forward channel $\mathcal{F}_{\delta t}$ and the local reduced density matrix $\rho_t^{S_j}$. The above bound shows that if $\rho$ has $\xi$-FML, local recovery is possible because the recovery error $\left\| \widetilde{\mathcal{R}}_{\delta t}^{\mathcal{S}_j}  \circ \mathcal{F}_{\delta t}[\rho_t] - \rho_t \right\|_1$ decays exponentially with the width of $B$.
 
The state $\rho$ has Markov length $\xi$ if its conditional mutual information (CMI) satisfies \cite{sang2024stabilitymixedstatequantumphases}
\begin{equation}
    I_{\rho}(A:C \,|\, B) \le \mathrm{poly}(|A|,|C|) \, e^{-\mathrm{dist}(A,C)/\xi}
    \label{eq:markov-length}
\end{equation}

For three regions $A, B, C$ arranged as in \cref{figS_ABC}---where $A$ is simply connected, $B$ is an annulus surrounding $A$, and $C = \overline{A \cup B}$ is the rest of the system---if $\rho$ is well-defined on arbitrarily large lattices of size $n$ with $\xi$ independent of $n$, we say $\rho$ has \emph{$\xi$-finite Markov length} ($\xi$-FML).

\begin{theorem*}[Restatement of Theorem~2 of the main text]
Consider an initial state $\rho_0$ of $n$ qubits. We evolve this state forward by sequentially applying an infinitesimal measurement channel $\mathcal{F}_{\delta t}$ over a total duration $T$. Here $T$ is divided into $N_T$ small time steps ($T = N_T , \delta t$). During each infinitesimal interval $\delta t$, the channel $\mathcal{F}_{\delta t}$ acts on a single qubit $j$ (which may depend on time $t$). After $N_T$ such steps, the final state $\rho_T = \mathcal{F}_{T}(\rho_0)$ is reached. 

Starting from the decohered state $\rho_T$, let the reverse diffusion process be approximated by applying $N_T$ sequential operations of a local twirled Petz recovery channel, denoted $\widetilde{\mathcal{R}}_{\delta t}^{\mathcal{S}_j}$. Each such recovery operation acts on a local subsystem $S_j = A \cup B$, where $A$ represents the specific qubit $j$ that was acted on by $\mathcal{F}_{\delta t}$ at the corresponding forward step, and $B$ is the set of additional qubits in the immediate neighborhood of $A$ that we include in the recovery operation. Assume a locality condition such that for every intermediate state $\rho_t$ ($0 \le t \le T$) and every choice of subsystem $S_j$, the conditional mutual information is bounded by $I_{\rho_t}(A:C|B) \le \mathrm{poly}(|A|,|C|) \, e^{-\mathrm{dist}(A,C)/\xi}$, where $C$ is the complement of $S_j$. Then by requiring
\begin{equation}
    \mathrm{dist}(A,C) \ge 
\xi \cdot \log \left( 
\frac{\mathrm{poly}(n) \cdot \, \ln\!\left( \frac{2n}{\epsilon} \right)}
{\gamma \epsilon^2 \cdot \delta t} 
\right)\ ,\qquad T \geq \frac{3n}{4\gamma} \ln\left( \frac{2n}{\epsilon} \right) \ , 
\end{equation}
the final state $\rho_0'$ produced by applying all $N_T$ Petz recovery steps remains close to the initial state $\rho_0$, with the error bounded by the trace distance:
\begin{equation}
\big\| \rho_0' - \rho_0 \big\|_1 \leq \epsilon\ .
\end{equation}
\end{theorem*}

\noindent
\begin{proof}
The total error can be bounded by the sum of single-step errors:
\begin{equation}
\big\| \rho_0' - \rho_0 \big\|_1 
= E_{\mathrm{petz}}(T) + E_{\mathrm{diff}}(T), \qquad E_{\mathrm{petz}}(T) = \sum_{k=1}^{N_T} 
\big\| \widetilde{\mathcal{R}}^{S_k}_{\delta t} \big( \mathcal{F}_{\delta t}(\rho_{t_k}) \big) - \rho_{t_k} \big\|_1\ ,
\end{equation}
where $\rho_{t_k} = \mathcal{F}_{\delta t}(\rho_{t_{k-1}})$ is the state after the $k$th forward step of duration $\delta t$.  
By the Cauchy--Schwarz inequality, this implies
\begin{equation}\label{eq-total-cmi}
E_{\mathrm{petz}}(T)^2
\le N_T \sum_{k=1}^{N_T} 
\big\| \widetilde{\mathcal{R}}^{S_k}_{\delta t} \big( \mathcal{F}_{\delta t}(\rho_{t_k}) \big) - \rho_{t_k} \big\|_1^2.
\end{equation}

Each single-step recovery error is controlled by the drop in the local CMI over that step:
\begin{equation}
\big\| \widetilde{\mathcal{R}}^S_{\delta t}(\mathcal{F}_{\delta t}(\rho)) - \rho \big\|_1^2
\le 2 \ln 2 \,\Big[\, I_\rho(A:C \,|\, B) - I_{\mathcal{F}_{\delta t}(\rho)}(A:C \,|\, B) \,\Big] ,
\end{equation}
for any subsystem $S = A \cup B$ (with complement $S^C$).  

Summing these squared errors over $k = 1,\dots,N_T$ yields a telescoping sum of local CMI differences at each recovery step.  
By the data-processing inequality, the conditional mutual information is non-increasing under the forward channel, 
$I_{\rho_{t_k}}(A:C|B) \le I_{\rho_{t_{k-1}}}(A:C|B)$, so the sum of CMI differences telescopes to the initial-minus-final value:
\begin{equation}
\sum_{k=1}^{N_T} \Big( I_{\rho_{t_k-1}}(A:C|B) - I_{\rho_{t_k}}(A:C|B) \Big)
= I_{\rho_0}(A:C|B) - I_{\rho_T}(A:C|B) \leq I_{\rho_0}(A:C|B),
\end{equation}
using the fact that $I_{\rho_T} \ge 0$. 

Meanwhile, $E_{\mathrm{diff}}(T) \sim n \exp \left(-\frac{4\gamma}{3n}T\right)$, because for non-identity Pauli strings, each Pauli weight decays exponentially to zero with increasing diffusion time $T$, as shown in \cref{seq-measurement-pauli-weight2}. To achieve $E_{\mathrm{diff}}(T) \leq \frac{\epsilon}{2},$ it suffices to require
\begin{equation}
T \geq \frac{3n}{4\gamma} \ln\left( \frac{2n}{\epsilon} \right)  
\end{equation}

Using the assumed locality bound $I_{\rho_0}(A:C|B) \le \mathrm{poly}(|A|,|C|) \, e^{-\mathrm{dist}(A,C)/\xi}$ and noting $N_T = T/\delta t,$ \cref{eq-total-cmi} becomes
\begin{equation}
E_{\mathrm{petz}}(T)^2 \leq \frac{T}{\delta t} \mathrm{poly}(|A|,|C|) \, e^{-\mathrm{dist}(A,C)/\xi}   
\end{equation}

To achieve an error $E_{\mathrm{petz}}(T)^2 \leq (\frac{\epsilon}{2})^2$, it suffices to require
\begin{equation}
    \mathrm{dist}(A,C) \ge \xi \cdot \log \left( \frac{4\mathrm{poly}(n)T}{\epsilon^2 \cdot \delta t} \right) \ge 
\xi \cdot \log \left( 
\frac{\mathrm{poly}(n) \cdot \, \ln\!\left( \frac{2n}{\epsilon} \right)}
{\gamma \epsilon^2 \cdot \delta t} 
\right)
    \label{eq:recovery-condition}
\end{equation}

Then we have the global recovery error
\begin{equation}
\big\| \rho_0' - \rho_0 \big\|_1 = E_{\mathrm{petz}}(T) + E_{\mathrm{diff}}(T) \leq \frac{\epsilon}{2}+\frac{\epsilon}{2}= \epsilon\ .
\end{equation}
\end{proof}

\end{document}